\pdfoutput=1

\documentclass[english,11pt]{article}
\include{macros}

\begin{document}

\title{\Large Pseudo-Labeling for Kernel Ridge Regression under Covariate Shift}

\author{Kaizheng Wang\thanks{Department of IEOR and Data Science Institute, Columbia University. Email: \texttt{kaizheng.wang@columbia.edu}.}
}

\date{This version: July 2025}

\maketitle

\begin{abstract}
We develop and analyze a principled approach to kernel ridge regression under covariate shift. The goal is to learn a regression function with small mean squared error over a target distribution, based on unlabeled data from there and labeled data that may have a different feature distribution. We propose to split the labeled data into two subsets, and conduct kernel ridge regression on them separately to obtain a collection of candidate models and an imputation model. We use the latter to fill the missing labels and then select the best candidate accordingly. Our non-asymptotic excess risk bounds demonstrate that our estimator adapts effectively to both the structure of the target distribution and the covariate shift. This adaptation is quantified through a notion of effective sample size that reflects the value of labeled source data for the target regression task. Our estimator achieves the minimax optimal error rate up to a polylogarithmic factor, and we find that using pseudo-labels for model selection does not significantly hinder performance.
\end{abstract}


\section{Introduction}\label{sec-intro}

Covariate shift is a phenomenon that occurs when the feature distribution of the test data differs from that of the training data. It can cause performance degradation of the model, as the training samples may have poor coverage of challenging cases to be encountered during deployment. 
Such issue arises from a variety of scientific and engineering applications \citep{Hec79,Zad04}.
For instance, a common task in personalized medicine is to predict the treatment effect of a medicine given a patient's covariates. However, the labeled data are oftentimes collected from particular clinical trials or observational studies, which may not be representative of the population of interest. 
Direct uses of traditional supervised learning techniques, such as empirical risk minimization and cross-validation, could yield sub-optimal results. 
Indeed, their theories are mostly built upon the homogeneity assumption that the training and test data share the same distribution \citep{Vap99}.

Acquisition of labeled data from the target population can be costly or even infeasible. 
In the above example of personalized medicine, it requires conducting new clinical trials.
On the other hand, unlabeled data are cheaply available. This leads to the following fundamental question that motivates our study: 
{ \setlist{rightmargin=\leftmargin} \begin{itemize} \item[]\centering\emph{How to train a predictive model using only unlabeled data from the distribution of interest and labeled data from a relevant distribution?}
\end{itemize} }
\noindent It is a major topic in domain adaptation and transfer learning \citep{PYa10,SKa12}, where the distribution of interest is called the \emph{target} and the other one is called the \emph{source}.
The missing labels in the target data create a significant obstacle for model assessment and selection, calling for principled approaches without ad hoc tuning.

In this paper, we provide a solution to the above problem in the context of kernel ridge regression (KRR). We work under the standard assumption of covariate shift that the source and the target share the same conditional label distribution. In addition, the conditional expectation of the label given covariates is described by a function in a reproducing kernel Hilbert space (RKHS). The goal is to learn the regression function with small mean squared error over the target distribution. While KRR is a natural choice, its performance depends on a penalty parameter $\lambda$ that is usually selected by hold-out validation or cross-validation. Neither of them is directly applicable without target labels. Ignoring the target data, however, may lead to suboptimal results if the source samples do not adequately represent the target distribution. To integrate both datasets, we propose to fill the missing target labels using an imputation model trained on part of the source data. Meanwhile, we train a collection of candidate models on the rest of the source data, using KRR with different penalty parameters. Finally, we select the best candidate based on pseudo-labels.

\paragraph*{An illustrative example}

We use a numerical simulation to demonstrate both the challenge and our proposed solution.\footnote{The code for reproducing the results is available at  \url{https://github.com/kw2934/KRR}.} The collection of absolutely continuous functions on $[0, 1]$ with square integrable derivatives forms an RKHS, which we will discuss further in \Cref{example-Sobolev}. Define $f^{\star} (x) =  \cos ( 2 \pi x ) - 1$. Let $\nu_0 = \cU [0, 1/2]$ and $\nu_1 = \cU [1/2, 1]$ be two uniform distributions. Based on them, we construct the source distribution $\source = \frac{5}{6} \nu_0 + \frac{1}{6} \nu_1$ and the target distribution $\target = \frac{1}{6} \nu_0 + \frac{5}{6} \nu_1$ that concentrate on $[0, 1/2]$ and $[1/2, 1]$, respectively. We generate $250$ i.i.d.~labeled samples $\{ (x_i, y_i) \}_{i=1}^n$ with $x_i \sim \source$ and $y_i | x_i \sim N (  f^{\star} (  x_i )  , 1 )$. Then, we run KRR with the penalty parameter $\lambda \in \{ 0.0008, 0.0032, 0.0128 \}$ to obtain three candidate models. Their excess risks on $\source$ and $\target$ are estimated using 10000 new samples.  We report the means and standard errors in \Cref{table-candidates}. The source and target favor the medium and small penalties, respectively.

\begin{table}[ht]
	\centering
	\begin{tabular}{|c|c|c|c|}
		\hline
		Test distribution & Small Penalty & Medium Penalty & Large Penalty \\
		\hline
		$\source$ & 0.021 (2.4e-4) & \textbf{0.013 (1.5e-4)} & 0.046 (5.9e-4) \\
		\hline
		$\target$ & \textbf{0.017 (2.0e-4)} & 0.024 (2.5e-4) & 0.100 (7.5e-4) \\
		\hline
	\end{tabular}
	\caption{Excess risks of three candidates on $\source$ and $\target$}\label{table-candidates}
\end{table}

The results are visualized in the left panel of \Cref{fig-candidates}. 
On the interval $[0, 1/2]$ with abundant data, the red and cyan curves show signs of undersmoothing and oversmoothing, respectively. In contrast, the blue curve provides a good fit. However, on the interval $[1/2, 1]$ with sparse data, the blue curve becomes oversmoothed while the red curve fits better. Such observation is consistent with the quantitative summary in \Cref{table-candidates}.

\begin{figure}[t]
	\centering
	\includegraphics[width=0.49\linewidth]{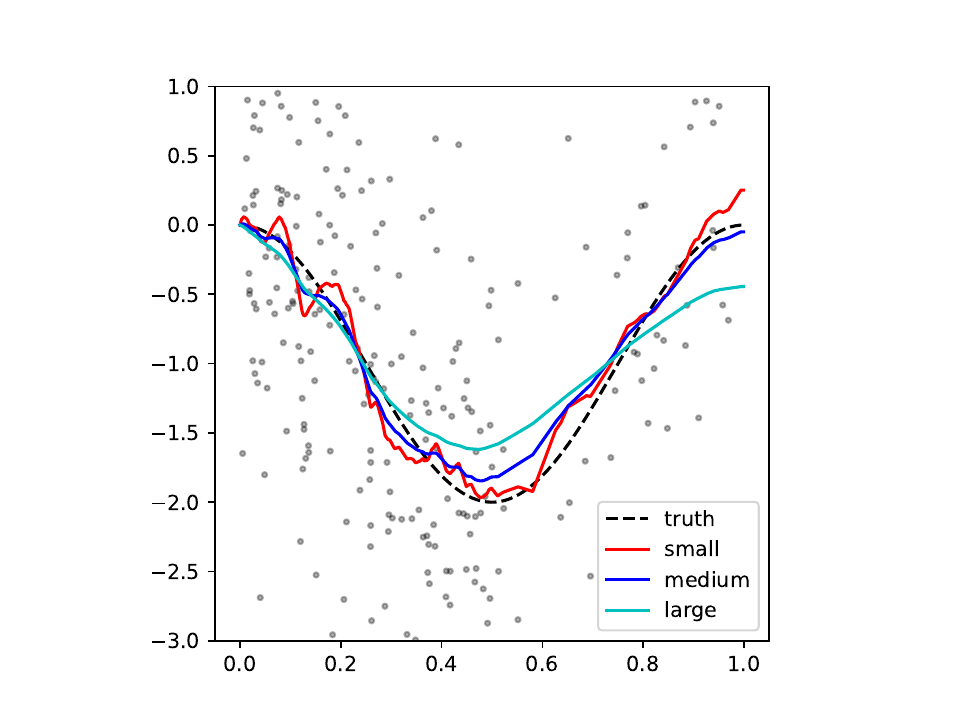}
	\includegraphics[width=0.49\linewidth]{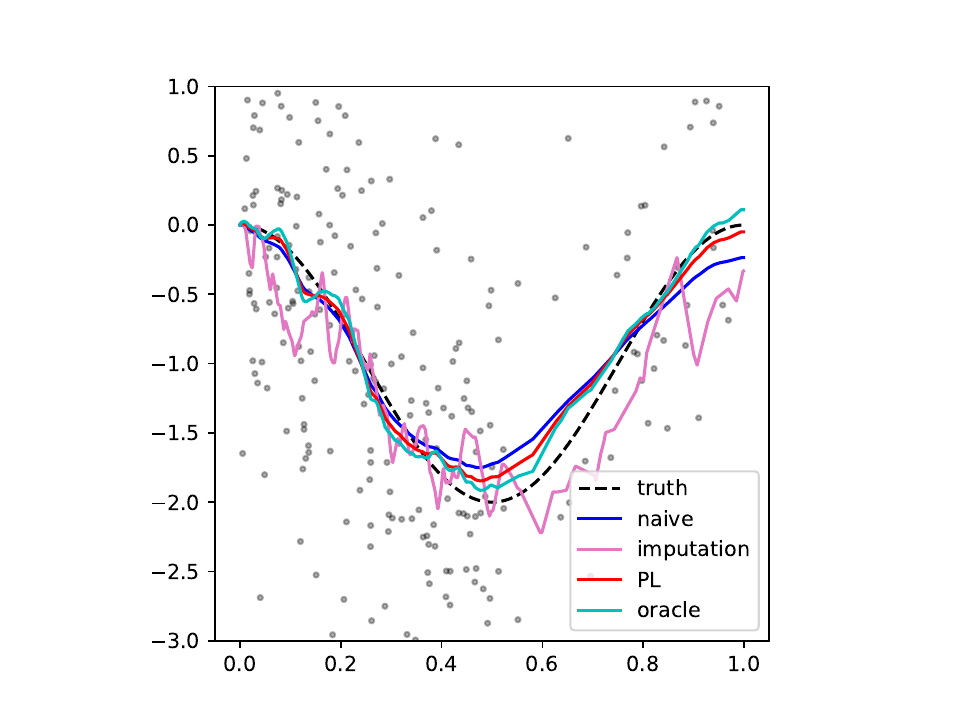}
	\caption{Covariate shift and its adaptation. In both panels, the black dashed curves and the gray dots show the true response function $f^{\star}$ and the source data. Left panel: comparison of three candidate models, where the red, blue, and cyan curves correspond to small, medium and large penalties. Right panel: comparison of three model selection approaches, where the red, blue, and cyan curves correspond to models chosen by the pseudo-labeling, na\"{i}ve, and oracle methods. In addition, the pink curve visualizes the imputation model.}
	\label{fig-candidates}
\end{figure}

To illustrate our proposed method for covariate shift adaptation, denote by $\cD_1$ the aforementioned 250 samples, $\cD_2$ another collection of 250 labeled source samples, and $\cD_0$ a collection of 500 unlabeled target samples. 
We run KRR on $\cD_1$ with $\lambda$ in a geometric sequence ranging from 0.0002 to 1.6384, with a common ratio of 2. 
Meanwhile, we run KRR on $\cD_2$ with $\lambda = 0.0002$ to get an imputation model, use it to generate pseudo-labels for $\cD_0$, and choose the best candidate model based on them. 
We compare our selection procedure against a na\"{i}ve approach that uses $\cD_2$ as the validation set, and an oracle approach that uses noiseless responses of the target data $\cD_0$ for validation. The results are shown in the right panel of \Cref{fig-candidates}. On the interval $[1/2, 1]$ where $\target$ has most of its mass, selection by pseudo-labeling (red) outperforms the na\"{i}ve method (blue), and it is comparable to the oracle method (cyan). It is important to note that the undersmoothed imputation model (pink) is not suitable for direct use but serves as a useful benchmark for model selection.

\paragraph*{Main contributions} We develop a pseudo-labeling approach to KRR with labeled source data and unlabeled target data. 
Our theory confirms its adaptivity to the unknown covariate shift in common scenarios. 

To characterize the performance of KRR on the target distribution, we introduce a notion of \emph{effective sample size} $\effectivesamplesize$ and show that
\begin{theorem}[Informal]\label{thm-intro}
	$n$ labeled source samples $\approx$ $\effectivesamplesize$ labeled target samples.
\end{theorem}
\noindent
This quantity gives a sharp measure of the information transfer, even if the covariate distribution of the target is singular with respect to that of the source.

To analyze our model selection method, we derive a bias-variance decomposition of the imputation model's impact. The result holds for general function classes and is of independent interest. It reveals that the variance has a much weaker influence compared to the bias. To the best of our knowledge, this is the first general theory of model selection with pseudo-labels.
For KRR, the finding suggests using a small penalty parameter for training the imputation model. Although such model often has sub-optimal mean squared error, it produces pseudo-labels that help find a near-optimal model among a collection of candidates.

\paragraph*{Related work}

Here we give a non-exhaustive review of related work in the literature.
As we mentioned above, our final estimator is selected by hold-out validation with pseudo-labels.
Simple as it is, hold-out validation is one of the go-to methods for model selection in supervised learning. When there is no covariate shift, \cite{BMa06} proved that the selected model adapts to the unknown noise condition in many classification and regression problems. \cite{CYa10} showed that KRR tuned by hold-out validation achieves the minimax optimal error rate.
Our study further reveals that faithful model selection is possible even if the validation labels are synthetic and the imputation model is trained on a different distribution. 

Our method is based on regression imputation, a powerful tool for missing data analysis. It has been widely applied to estimating linear functionals such as the average treatment effect \citep{Che94,HMa19,DJW20,HWa21,DWW24,MWB22,WKi23,MDW23}. In particular, \cite{WKi23} considered estimating the mean response (label) over the target data in our setting, and proposed to use an imputation model trained by KRR with a small penalty. The goals in the aforementioned works are different from ours, as we aim to estimate the whole function and use the imputed data for model selection.

In the machine learning literature, filling the missing labels by a trained model is referred to as pseudo-labeling \citep{Lee13}. It has achieved great empirical success in reducing the labeling cost. Theoretical investigations focus on classification problems under mixture assumptions that the class-conditional covariate distributions have nice properties (e.g., log-concavity) and are well-separated from each other, which are not suitable for regression problems with continuous responses.
A line of works in this direction \citep{KML20,CGL21,LWL21} consider classification under covariate shift, with labeled source data and unlabeled target data. Their algorithms alternate between predicting the missing labels and using the completed data to train a classifier, similar to the Expectation-Maximization algorithm for latent variable models.

In the presence of covariate shift, the empirical risk defined by the source data is generally a biased estimate of the target population risk. A natural way of bias correction is to reweight the source data by the likelihood ratio between the target covariate distribution and the source one. 
The reweighted empirical risk then serves as the objective function. The weights may be truncated so as to reduce the variance \citep{Shi00,CMM10,SKa12,MPW22}. 
Such methods need absolute continuity of the target distribution with respect to the source, as well as knowledge about the likelihood ratio. The former is a strong assumption, and the latter is notoriously hard to estimate in high dimensions. Matching methods came as a remedy \citep{HGB06}. Yet, they require solving large optimization problems to obtain the importance weights. Moreover, theoretical understanding of the resulting model is largely lacking.

There is a recent surge of interest in the theory of transfer learning under distribution shift. Most works assume either the target data are labeled, or the target covariate distribution is known \citep{BBC10,HKp19,MFA20,YZW20,KMa21,CWe21,RCS21,TAP21,MSB22,SZa22,PMW22,WZB22}.
When there are only finitely many unlabeled target samples, one needs to extract information about the covariate shift and incorporate that into learning. \cite{LZL20} studied doubly-robust regression under covariate shift. \cite{LHL21} obtained results for ridge regression in a Bayesian setting. A recent work \cite{MPW22} is the most closely related to ours, and we will discuss it in detail in \Cref{sec-excess-risk-general}.

\paragraph*{Outline} The rest of the paper is organized as follows. \Cref{sec-setup} describes the problem of covariate shift and its challenges.
\Cref{sec-methodology} introduces the methodology. \Cref{sec-excess-risk-general} presents excess risk bounds for our estimator and discusses its adaptivity. 
\Cref{sec-oracle} explains the power of pseudo-labeling. 
\Cref{sec-proofs-main} provides the proofs of our main results.
Finally, \Cref{sec-discussion} concludes the paper and discusses possible future directions.

\paragraph*{Notation}
The constants $c_1, c_2, C_1, C_2, \cdots$ may differ from line to line. 
We use the symbol $[n]$ as a shorthand for $\{ 1, 2, \cdots, n \}$ and $| \cdot |$ to denote the absolute value of a real number or cardinality of a set. 
For nonnegative sequences $\{ a_n \}_{n=1}^{\infty}$ and $\{ b_n \}_{n=1}^{\infty}$, we write $a_n \lesssim b_n$ or $a_n = O(b_n)$ 
if there exists a positive constant $C$ such that $a_n \leq C b_n$. We use $\widetilde{O}$ 
to hide polylogarithmic factors. In addition, we write $a_n \asymp b_n$ if $a_n \lesssim b_n$ and $b_n \lesssim a_n$; $a_n = o(b_n)$ if $a_n = O(c_n b_n)$ for some $c_n \to 0$.
Notations with tildes (e.g.,~$\widetilde O$) hide logarithmic factors.
For a matrix $\bA$, we use $\| \bA \|_2 = \sup_{\|\bx\|_2 = 1 } \| \bA \bx \|_2$ to denote its spectral norm.
For a bounded linear operator $\bA$ between two Hilbert spaces $\HH_1$ and $\HH_2$, we define its operator norm $\| \bA \| = \sup_{\| \bu \|_{\HH_1} = 1} \| \bA \bu \|_{\HH_2}$. 
For a bounded, self-adjoint linear operator $\bA$ mapping a Hilbert space $\HH$ to itself, we write $\bA \succeq 0$ if $\langle \bu , \bA \bu \rangle \geq 0$ holds for all $\bu \in \HH$; in that case, we define $\| \bu \|_{\bA} = \sqrt{ \langle \bu , \bA \bu \rangle } $.
For any $\bu $ and $ \bv $ in a Hilbert space $\HH$, their tensor product $\bu \otimes \bv$ is a linear operator from $\HH$ to itself that maps any $\bw \in \HH$ to $\langle \bv , \bw \rangle \bu$. 
Define $\| X \|_{\psi_{\alpha}} = \sup_{p \geq 1} \{ p^{-1/{\alpha}} \EE^{1/p} |X|^p \}$ for a random variable $X$ and $\alpha \in \{ 1, 2 \}$. If $\bX$ is a random element in a separable Hilbert space $\HH$, then we let $\|\bX\|_{\psi_2} = \sup_{\| \bu \|_{\HH} = 1} \| \langle \bu , \bX \rangle \|_{\psi_2} $.

\section{Problem setup}\label{sec-setup}

\subsection{Linear model and covariate shift}

Let $\{ (\bx_i , y_i) \}_{i=1}^n$ and $\{ ( \bx_{0i} , y_{0i} ) \}_{i=1}^{n_0}$ be two datasets named the \emph{source} and the \emph{target}, respectively. Here $\bx_i, \bx_{0i} \in \cX$ are covariate vectors in some feature space $\cX$; $y_{i}, y_{0i} \in \RR$ are responses. 
However, the target responses $\{ y_{0i} \}_{i=1}^{n_0}$ are \emph{unobserved}. The goal is to learn a predictive model from $\{ (\bx_i , y_i) \}_{i=1}^n$ and $\{ \bx_{0i} \}_{i=1}^{n_0}$ that works well on the target data.

The task is impossible when the source and the target datasets are arbitrarily unrelated. To make the problem well-posed, we adopt a standard assumption that the source and the target datasets share a common regression model. 
In our basic setup, the feature space $\cX$ is a finite-dimensional Euclidean space $\RR^d$, the data are random samples from an unknown distribution, and the conditional mean response is a linear function of the covariates. Below is the formal statement.

\begin{assumption}[Common linear model and random designs]\label{assumption-linear-model}
The two datasets $\{ ( \bx_i , y_i ) \}_{i=1}^n$ and $\{ ( \bx_{0i} , y_{0i} )  \}_{i=1}^{n_0}$ are independent, each of which consists of i.i.d.~samples.
\begin{itemize}
\item The distributions of covariates $\bx_{i}$ and $\bx_{0i}$ are $\source$ and $\target$, respectively. Furthermore, $\bSigma = \EE  (\bx_i \bx_i^{\top})$ and $\bSigma_0 = \EE  (\bx_{0i} \bx_{0i}^{\top})$ exist.
\item There exists $\btheta^{\star} \in \RR^d$ such that $\EE ( y_i | \bx_i = \bx ) = \EE ( y_{0i} | \bx_{0i} = \bx ) =  \langle   \bx , \btheta^{\star} \rangle$, $\forall \bx \in \RR^d$. 
\item Let $\varepsilon_i = y_i - \langle \bx_i , \btheta^{\star} \rangle$ and $\varepsilon_{0i} = y_{0i} - \langle   \bx_{0i} , \btheta^{\star} \rangle$. Conditioned on $\{ \bx_i  \}_{i=1}^n$ and $\{ \bx_{0i} \}_{i=1}^{n_0}$, the errors $\{ \varepsilon_i \}_{i=1}^n$ and $\{ \varepsilon_{0i} \}_{i=1}^{n_0}$ have finite second moments. 
\end{itemize}
\end{assumption}

More generally, we will study linear model in a reproducing kernel Hilbert space (RKHS) with possibly infinite dimensions. 
The feature space $\cX$ can be any set. Suppose that we have a symmetric, positive semi-definite kernel $K:~ \cX \times \cX \to \RR$ satisfying the conditions below:
\begin{itemize}
	\item (Symmetry) For any $ \bz \in \cX$ and $\bw \in \cX$, $K(\bz , \bw) = K(\bw,\bz)$;
	\item (Positive semi-definiteness) For any $m \in \ZZ_+$ and $\{ \bz_i \}_{i=1}^m \subseteq \cX$, the $m\times m$ matrix $[ K(\bz_i, \bz_j) ]_{m\times m}$ is positive semi-definite.
\end{itemize}
According to the Moore-Aronszajn Theorem \citep{Aro50}, there exists a Hilbert space $\HH$ with inner product $\langle \cdot , \cdot \rangle$ and a mapping $\phi:~\cX\to \HH$ such that $K(\bz, \bw) = \langle \phi(\bz), \phi(\bw) \rangle$, $\forall \bz, \bw \in \cX$. 
The kernel $K$ and the associated space $\HH$ are called a \emph{reproducing kernel} and a \emph{reproducing kernel Hilbert space}, respectively.
Now, we upgrade Assumption \ref{assumption-linear-model} to the RKHS setting.

\begin{assumption}[Common linear model in RKHS and random designs]\label{assumption-linear-model-rkhs}
The datasets $\{ ( \bx_i , y_i ) \}_{i=1}^n$ and $\{ ( \bx_{0i} , y_{0i} )  \}_{i=1}^{n_0}$ are independent, each of which consists of i.i.d.~samples.
\begin{itemize}
\item The distributions of covariates $\bx_{i}$ and $\bx_{0i}$ are $\source$ and $\target$, respectively. Furthermore, $\bSigma = \EE [ \phi(\bx_i) \otimes \phi(\bx_i)  ]$ and $\bSigma_0 = \EE [ \phi(\bx_{0i}) \otimes \phi(\bx_{0i})  ]$ are trace class.
\item There exists $\btheta^{\star} \in \HH$ such that $\EE ( y_i | \bx_i = \bx ) = \EE (y_{0i} | \bx_{0i} = \bx ) = \langle  \phi(\bx) , \btheta^{\star} \rangle$, $\forall \bx \in \cX$. 
\item Let $\varepsilon_i = y_i - \langle \phi (\bx_i) , \btheta^{\star} \rangle$ and $\varepsilon_{0i} = y_{0i} - \langle   \phi ( \bx_{0i} ) , \btheta^{\star} \rangle$. 
Conditioned on $\{ \bx_i  \}_{i=1}^n$ and $\{ \bx_{0i} \}_{i=1}^{n_0}$, the errors $\{ \varepsilon_i \}_{i=1}^n$ and $\{ \varepsilon_{0i} \}_{i=1}^{n_0}$ have finite second moments.
\end{itemize}
\end{assumption}

Under Assumption \ref{assumption-linear-model-rkhs}, the conditional mean function of the response $f^{\star} (\cdot) =  \langle \phi(\cdot), \btheta^{\star} \rangle$ belongs to a function class 
\begin{align*}
\functionclass =  \{
f_{\btheta}  :~ \cX \to \RR
~ |~ 
f_{\btheta}  (\bx) = \langle \phi(\bx) , \btheta \rangle \text{ for some } \btheta \in \HH
 \}
\end{align*}
generated by the kernel $K$. Define $\| f_{\btheta} \|_{\functionclass} = \| \btheta \|_{\HH}$. It is easily seen that $\| \cdot \|_{\functionclass}$ is a norm and $\functionclass$ becomes a Hilbert space that is isomorphic to $\HH$. Below are a few common examples \citep{Wai19}. 

\begin{example}[Linear and affine kernels]\label{example-linear}
Let $\cX = \RR^d$. The linear kernel $K(\bz, \bw) = \bz^{\top} \bw$ gives the standard inner product. In that case, $\HH = \RR^d$, $\phi (\bx) = \bx$ is the identity map, and $\functionclass$ is the class of all linear functions (without intercepts). An extension is the affine kernel $K(\bz , \bw) = 1 + \bz^{\top} \bw$ that generates all affine functions (with intercepts).
\end{example}

\begin{example}[Polynomial kernels]\label{example-polynomial}
Let $\cX = \RR^d$. The homogeneous polynomial kernel of degree $m \geq 2$ is $K(\bz, \bw) = ( \bz^{\top} \bw )^{m}$. The corresponding $\HH$ is the Euclidean space of dimension ${ m + d - 1 \choose m }$, the coordinates of $\phi$ are degree-$m$ monomials, and $\functionclass$ is the class of all homogeneous polynomials of degree $m$. An extension is the inhomogeneous polynomial kernel $K(\bz, \bw) = (1 +  \bz^{\top} \bw )^{m}$ that generates all polynomials of degree $m$ or less.
\end{example}

\begin{example}[Laplace and Gaussian kernels]\label{example-Gaussian}
Let	$\cX = \RR^d$ and $\alpha > 0$. The Laplace and Gaussian kernels are $K(\bz, \bw) = e^{ - \alpha  \| \bz - \bw \|_2  }$ and $K(\bz, \bw) = e^{ - \alpha  \| \bz - \bw \|_2^{2} }$, respectively. In either case, $\functionclass$ is infinite-dimensional.
\end{example}

\begin{example}[First-order Sobolev kernel]\label{example-Sobolev}
Let	$\cX = [0, 1]$ and $K(z, w) =\min \{ z, w \}$. Then, $\functionclass$ is the first-order Sobolev space
	\[
	\functionclass = \bigg\{ f :~[0, 1] \to \RR ~ \bigg| ~ f(0) = 0,~ \int_{0}^{1} |f'(x)|^2 \rd x < \infty \bigg\} .
	\]
Any $f \in \functionclass$ is absolutely continuous on $[0, 1]$ with square integrable derivative function.	
\end{example}

The function spaces are finite-dimensional in Examples \ref{example-linear} and \ref{example-polynomial}, and infinite-dimensional in the others. For any $f \in \functionclass$ serving as a predictive model, we measure its \emph{risk} by the out-of-sample mean squared prediction error on the target population:
\begin{align*}
\risk  (f) = \EE | f (   \bx_{\mathrm{new}}  )   - y_{\mathrm{new}}  |^2,
\end{align*}
where $(\bx_{\mathrm{new}} , y_{\mathrm{new}})$ is a new sample. Under Assumption \ref{assumption-linear-model-rkhs}, $f^{\star} $ minimizes the functional $ \risk (\cdot)$, and 
\begin{align*}
 \risk  ( f ) -  \risk  ( f^{\star} ) =  \EE_{\bx \sim \target}  | f ( \bx ) - f^{\star} (\bx)  |^2 .
\end{align*}
In words, the \emph{excess risk} is equal to the mean squared estimation error over $\target$.

We emphasize that neither $\source$ nor $\target$ is known to us, and we merely have a finite number of samples. 
The difference between $\source$ and $\target$ is called the \emph{covariate shift}. Below we explain why it brings challenges to learning in high dimensions.

\subsection{Ridge regression and new challenges}\label{sec-ridge-challenges}

Consider a finite-dimensional linear model (Assumption \ref{assumption-linear-model}). When $\mathrm{span} \{ \bx_i \}_{i=1}^n = \RR^d$, the ordinary least squares (OLS) regression
\begin{align*}
\min_{\btheta \in \RR^d} \bigg\{ 
\frac{1}{n} \sum_{i=1}^n (\langle \bx_i , \btheta \rangle - y_{i} )^2
\bigg\}
\end{align*}
has a unique solution, and it is an unbiased estimate of $\btheta^{\star}$. When $\mathrm{span} \{ \bx_i \}_{i=1}^n$ is a proper subspace of $\RR^d$ (e.g., if $d > n$), a popular approach is ridge regression \citep{HKe70}
\begin{align}
\min_{\btheta \in \RR^d} \bigg\{ 
\frac{1}{n} \sum_{i=1}^n (\langle \bx_i , \btheta \rangle - y_{i} )^2 + \lambda \| \btheta \|_2^2 
\bigg\}.
\label{eqn-ridge}
\end{align}
The tuning parameter $\lambda > 0$ ensures the uniqueness of solution and enhances its noise stability. Ridge regression is suitable if we have prior knowledge that $\| \btheta^{\star} \|_2$ is not too large.

An extension to the RKHS setting (Assumption \ref{assumption-linear-model-rkhs}) is kernel ridge regression (KRR)
\begin{align}
\min_{f \in \functionclass}
\bigg\{  \frac{1}{n} \sum_{i=1}^n [  f( \bx_i )  - y_{i} ]^2 + \lambda \| f \|_{\functionclass}^2 
\bigg\} ,
\label{eqn-krr}
\end{align}
which reduces to the ridge regression \eqref{eqn-ridge} when $\cX = \RR^d$ and $K$ is the linear kernel in \Cref{example-linear}.
Although $\functionclass$ may have infinite dimensions in general, it was shown \citep{Wah90} that the optimal solution to \eqref{eqn-krr} has a finite-dimensional representation $\widehat{f} (\cdot) = \sum_{i=1}^{n} \widehat\alpha_i K (\bx_i , \cdot)$, with $\widehat\balpha = (\widehat\alpha_1,\cdots, \widehat\alpha_n)^{\top}$ being the unique solution to the quadratic program
\[
\min_{\balpha \in \RR^n} \bigg\{ 
\frac{1}{n} \| \bK \balpha - \by \|_2^2 + \lambda \balpha^{\top} \bK \balpha
\bigg\},
\]
with $\bK = [ K(\bx_i, \bx_j) ]_{n \times n}$ and $\by = (y_1,\cdots, y_n)^{\top}$.

For both ridge regression \eqref{eqn-ridge} and its kernelized version \eqref{eqn-krr}, choosing a large $\lambda$ shrinks the variance but inflates the bias. Quality outputs hinge on a balance between the above two effects.
In practice, tuning is often based on risk estimation using hold-out validation, cross-validation, and so on.
For instance, suppose that we have a finite set $\Lambda$ of candidate tuning parameters. Each $\lambda \in \Lambda$ is associated with a candidate model $\widehat{f}_{\lambda} \in \functionclass$ that solves the program \eqref{eqn-krr}.
If the responses $\{ y_{0i} \}_{i=1}^{n_0}$ in the target data were observed, we could estimate the out-of-sample risk of every $\widehat{f}_{\lambda} $ by its empirical version
\begin{align*}
\frac{1}{n_0}  \sum_{i=1}^{n_0} [ \widehat{f}_{\lambda}  (  \bx_{0i} ) - y_{0i} ]^2 ,
\end{align*}
and then select the one with the smallest empirical risk. Unfortunately, the missing responses in the target data create a visible obstacle. Although the source samples are labeled, they are not representative for the target distribution in the presence of covariate shift. 
The tuning parameter with the best predictive performance on the source could be sub-optimal for the target, as is shown in the numerical example in \Cref{sec-intro}.

Theoretical studies also shed lights on the impact of covariate shift. Below is one example.

\begin{example}\label{example-Sobolev-challenges}
Consider the Sobolev space in \Cref{example-Sobolev}, which has $\| f \|_{\functionclass} = \sqrt{ \int_{0}^{1} |f'(x)|^2 \rd x }$. Let $\source = \cU [0, 1]$ be the uniform distribution. Assume that the errors $\{ \varepsilon_i \}_{i=1}^n$ are i.i.d.~$N(0,\sigma^2)$.
\begin{itemize}
\item If $\target = \source$, then the minimax optimal rate of the prediction error over the ball $\{ f:~ \| f \|_{\functionclass} \leq R \}$ is achieved by KRR \eqref{eqn-krr} with $\lambda \asymp ( \frac{\sigma^2}{R^2 n} )^{2/3}$ \citep{Wai19}. 

\item 
If $\target = \delta_{x_0}$ is the pointmass at some $x_0 \in ( 0, 1 )$, i.e. we only want to predict the response at a single point, then KRR \eqref{eqn-krr} with $\lambda = \frac{\sigma^2}{R^2 n}$ achieves the smallest worst-case risk among estimators that are linear in $\by$ \citep{Spe79,Li82}.
In terms of the mean squared error, it is minimax optimal up to a factor of $5/4$ over all estimators \citep{Don94}. 
\end{itemize}
\end{example}

This example shows the dependence of optimal tuning on the target. When $\target$ is spread out, we have to estimate the whole function, need a large $\lambda$ to supress the variance, and can afford a fair amount of bias. 
When $\target$ is concentrated, we are faced with a different bias-variance trade-off. 
An ideal method should automatically adapt to the structures of $\source$ and $\target$.
In the next section, we will present a simple approach to tackle the aforementioned challenges.

\section{Methodology}\label{sec-methodology}

We consider the problem of kernel ridge regression under covariate shift (Assumption \ref{assumption-linear-model-rkhs}). As the most prominent issue is the lack of target labels, we develop a regression imputation method to generate synthetic labels.
The idea is to estimate the true regression function using part of the source data and then feed the obtained imputation model with the unlabeled target data. This is called \emph{pseudo-labeling} in the machine learning literature \citep{Lee13}. We describe the proposed approach in Algorithm \ref{alg}. 

\begin{algorithm}[h]
	{\bf Input:} Source data $\{ (\bx_i , y_i) \}_{i=1}^n$, target data $\{ \bx_{0i} \}_{i=1}^{n_0}$, sample size $n_1$ for training, a set of penalty parameters $\Lambda$ for training, a penalty parameter $\widetilde\lambda$ for pseudo-labeling.\\
	{\bf Step 1 (Data splitting).} Randomly partition the source data $\{ (\bx_i , y_i ) \}_{i=1}^n$ into two disjoint subsets $\cD_1$ and $\cD_2$ of sizes $n_1$ and $n_2 = n - n_1$, respectively. \\
{\bf Step 2 (Training).} Choose a finite collection of penalty parameters $\Lambda \subseteq (0 , +\infty)$ and another one $\widetilde{\lambda} > 0$.
Run KRR on $\cD_1$ to get candidate models $\{ \widehat{f}_{\lambda} \}_{\lambda \in \Lambda}$, where
\[
\widehat{f}_{\lambda} = \argmin_{ f \in \functionclass } 
\bigg\{  \frac{1}{ | \cD_1 | } \sum_{ (\bx, y) \in \cD_1 } |  f( \bx )  - y |^2 + \lambda \| f \|_{\functionclass}^2 
\bigg\} , \qquad \forall \lambda \in \Lambda .
\]
Run KRR on $\cD_2$ to get an imputation model
\[
\widetilde{f} = \argmin_{ f \in \functionclass } 
\bigg\{ 
\frac{1}{ | \cD_2 | } \sum_{ (\bx, y) \in \cD_2 } |  f( \bx )  - y |^2 + \widetilde\lambda \| f \|_{\functionclass}^2 
\bigg\} .
\]
{\bf Step 3 (Pseudo-labeling).} Generate pseudo-labels $\widetilde{y}_{0i} = \widetilde{f} ( \bx_{0i} ) $ for $i \in [n_0]$. \\
{\bf Step 4 (Model selection).} Select any
\begin{align}
	\widehat\lambda \in \argmin_{ \lambda \in \Lambda }  \bigg\{
	\frac{1}{n_0} \sum_{i=1}^{n_0} 
	| \widehat{f}_{\lambda} ( \bx_{0i} ) - \widetilde{y}_{0i} |^2 
	\bigg\} .
	\label{eqn-oracle-objective}
\end{align} \\
{\bf Output:} $\widehat{f} = \widehat{f}_{\widehat\lambda}$.
	\caption{Pseudo-labeling for KRR under covariate shift}
	\label{alg}
\end{algorithm}

The first two steps are conducted on the source data only, summarizing the information into candidate models and an imputation model. 
After that, the raw source data will no longer be used. 
Given a set of unlabeled data from the target distribution, we generate pseudo-labels using the imputation model and then select the candidate model that best fits them.
The method is computationally efficient because KRR is easily solvable and the other operations (pseudo-labeling and model selection) run even faster. 

The hyperparameters $n_1$, $\Lambda$ and $\widetilde{\lambda}$ need to be specified by the user. The first two are standard and also arise in penalized regression without covariate shift. The last one is the crux because it affects the qualities of the imputation model $\widetilde{f}$, the pseudo-labels $\{ \widetilde{y}_{0i} \}_{i=1}^{ n_0 }$, and thus the selected model $\widehat{f}$. 
In \Cref{sec-excess-risk-general}, we will provide practical guidance based on a theoretical study. Roughly speaking, 
\begin{itemize}
\item The proportion of training data $n_1 / n$ is bounded away from $0$ and $1$, such as $1/2$. 
\item The set of penalty parameters $\Lambda$ consists of a geometric sequence from $O(n^{-1})$ to $O(1)$, with $O( \log n )$ elements. 
The resulting candidate models span a wide spectrum from undersmoothed to oversmoothed ones. 
\item 
The penalty parameter $\widetilde{\lambda}$ is set to be $O(n^{-1})$.
In fact, the model selection accuracy depends on the bias and a variance proxy of the pseudo-label vector $\widetilde\by = (  \widetilde{y}_{01} , \cdots,  \widetilde{y}_{0 n_0} )^{\top}$ in a delicate way.
Our $\widetilde{\lambda}$ achieve a bias-variance tradeoff that minimizes the \emph{impact on model selection}, rather than the usual one for minimizing the mean squared prediction error.
\end{itemize}

To get ideas about why imperfect pseudo-labels may still lead to faithful model selection, imagine that the responses $\{ y_{0i} \}_{i=1}^{n_0}$ of the target data were observed. 
Then, we could estimate the risk of a candidate $\widehat{f}_{\lambda}$ by $\frac{1}{n_0} \sum_{i=1}^{n_0} | \widehat{f}_{\lambda} (\bx_{0i}) - y_{0i} |^2$ and then perform model selection accordingly. 
While $ y_{0i} $ is a noisy version of $f^{\star} (\bx_{0i})$ with zero bias and non-vanishing variance $\sigma^2$, such hold-out validation method works well in practice. One can anticipate that pseudo-labels with small bias and reasonable variance still do the job.

Ideally, we want to set the hyperparameters $\Lambda$ and $\widetilde{\lambda}$ without seeing the target data. According to our theories in \Cref{sec-excess-risk-general}, this is indeed possible in many common scenarios. 
As a result, we can train the candidate models $\{ \widehat{f}_{\lambda} \}_{\lambda \in \Lambda}$ plus the imputation model $\widetilde{f}$ solely on the labeled source data, and store them for adaptation to target tasks in the future. The latter stage is remarkably simple and does not require any special design tailored to the target task.

\begin{remark}[Comparison with Lepski-type methods]
A line of works on KRR without covariate shift \citep{BMM19,PGr21} apply the Lepski principle \citep{Lep91} to choose a regularization parameter that balances the bias and the variance. A recent work \citep{LRe24} developed a Lepski-type method for KRR in the transfer learning setting under model shift, using labeled data from both the source and the target. The above settings are different from ours. 
Those methods require an estimate of the variance, which is usually constructed using theoretical bounds. 
For the problem we consider, variance estimation would involve Gram matrices defined by the source and the target data. When adapting to another target distribution, the procedure needs to be repeated again. Compared to that, our method is conceptually simpler and computationally more efficient. Moreover, it is applicable to general scenarios where a theoretical bound on the variance is not available.
\end{remark}

\section{Theoretical guarantees on the adaptivity}\label{sec-excess-risk-general}

This section is devoted to theoretical guarantees for Algorithm \ref{alg}. We will show that the estimator adapts to both the structure of the target distribution and the covariate shift. In the analysis, we will introduce a notion of \emph{effective sample size} to gauge the value of labeled source data for the target regression task, providing a rigorous statement of \Cref{thm-intro}.

\subsection{Preparations}

We will study kernel ridge regression \eqref{eqn-krr} in the RKHS setting (Assumption \ref{assumption-linear-model-rkhs}), which covers ridge regression \eqref{eqn-ridge} in the Euclidean setting (Assumption \ref{assumption-linear-model}) as a special case. 
Our analysis is based upon mild regularity assumptions. 
The signal strength is assumed to be bounded; the noise $\varepsilon_i$ can be dependent on the covariates $\bx_i$, but it needs to have sub-Gaussian tails when conditioned on the latter. 

\begin{assumption}[Bounded signal and noise]\label{assumption-noise}
We have $\| \btheta^{\star} \|_{\HH} \leq R < \infty$. Conditioned on $\bx_i $, the noise variable $ \varepsilon_i $ is sub-Gaussian:
\[
\sup_{p \geq 1 } \{ p^{-1/2} \EE^{1/p}  ( |\varepsilon_i|^p  | \bx_i  ) \} \leq \sigma < \infty .
\]
\end{assumption}

On the other hand, the covariates either have bounded norms or satisfy a strong version of the sub-Gaussianity assumption.

\begin{assumption}[Bounded covariates]\label{assumption-covariates-bounded}
$\| \phi (  \bx_i ) \|_{\HH} \leq M$ and $\| \phi (  \bx_{0i} ) \|_{\HH} \leq M$ hold almost surely for some $M < \infty$.
\end{assumption}

\begin{assumption}[Strongly sub-Gaussian covariates]\label{assumption-covariates-subg}
$\EE \| \phi ( \bx_{i} ) \|_{\HH}^2 \leq M^2$ and $\EE \| \phi ( \bx_{0i} ) \|_{\HH}^2 \leq M^2$ hold for some $M < \infty$. In addition, there exists $\kappa \in [1, +\infty) $ such that $\| \langle  \phi ( \bx_{i} ) , \bv \rangle \|_{\psi_2}^2 \leq \kappa \EE | \langle  \phi ( \bx_{i} ) , \bv \rangle |^2 $ and $ \| \langle \phi ( \bx_{0i} ) , \bv \rangle \|_{\psi_2}^2 \leq \kappa \EE | \langle  \phi ( \bx_{0i} ) , \bv \rangle |^2 $ hold for all $\bv \in \HH$.
\end{assumption}

Assumption \ref{assumption-covariates-bounded} holds if the reproducing kernel $K$ is uniformly bounded over the domain. This includes Examples \ref{example-Gaussian} and \ref{example-Sobolev}, as well as Examples \ref{example-linear} and \ref{example-polynomial} with $\cX$ being a compact subset of $\RR^d$.
Assumption \ref{assumption-covariates-subg}, also known as the $\psi_2-L_2$ norm equivalence, is commonly used in high-dimensional statistics \citep{Ver10,KLo17}. The constant $\kappa$ is invariant under bounded linear transforms of $\phi ( \bx_i ) $ and $\phi ( \bx_{0i} )$.
Combined with Assumption \ref{assumption-linear-model-rkhs}, it implies that 
\[
\| \langle  \phi ( \bx_{i} ) , \bv \rangle \|_{\psi_2}^2 \leq \kappa \langle \bv , \bSigma \bv \rangle \leq \kappa \| \bSigma \| \cdot \| \bv \|_{\HH}^2
\]
and similarly, $\| \langle  \phi ( \bx_{0i} ) , \bv \rangle \|_{\psi_2}^2   \leq \kappa \| \bSigma_0 \| \cdot \| \bv \|_{\HH}^2.
$
Therefore, both $\phi(\bx_i)$ and $\phi ( \bx_{0i} )$ are sub-Gaussian. In the Euclidean case (\Cref{example-linear}) with $\cX = \HH = \RR^d$ and $\phi (\bx) = \bx$, Assumption \ref{assumption-covariates-subg} is equivalent to $\| (\bSigma^{\dagger})^{1/2} \bx_i \|_{\psi_2} \leq \kappa$ and $\| (\bSigma_0^{\dagger})^{1/2} \bx_{0i} \|_{\psi_2} \leq \kappa$, where the superscript $\dagger$ denotes the Moore-Penrose pseudoinverse. It holds when $\bx_{i} \sim N( \bm{0} , \bSigma )$ and $\bx_{0i} \sim N( \bm{0} , \bSigma_0 )$.

\subsection{Adaptivity to the unknown distribution shift}

We are ready to state our general guarantees for Algorithm \ref{alg}. \Cref{thm-excess-risk-0} below shows that the proposed method handles arbitrary covariate shift without prior information about it. Meanwhile, we need knowledge of $\sigma$, $R$, and $M$ (i.e.~upper bounds on the noise level, the signal strength, and the covariate vector length) for setting the hyperparameters $\Lambda$ and $\widetilde{\lambda}$, which is standard for theoretical analysis. We defer the proof to \Cref{sec-thm-excess-risk-0-proof}.

\begin{theorem}[Excess risk]\label{thm-excess-risk-0}
Let Assumptions \ref{assumption-linear-model-rkhs} and \ref{assumption-noise} hold. In addition, suppose that either Assumption \ref{assumption-covariates-bounded} or \ref{assumption-covariates-subg} holds, and $n$ is even. Denote by $\{ \mu_j \}_{j=1}^{\infty}$ the eigenvalues of $\bSigma_0$ sorted in descending order. Let $\mu^2 =  \max  \{ \sigma^2 / R^2, M^2 \}$ and 
\begin{align}
\effectivesamplesize 
= \sup  \{ t \leq n :~ t \bSigma_0 \preceq n \bSigma + 
\mu^2
\bI
\}.
\label{eqn-n-eff}
\end{align}
Consider Algorithm \ref{alg} with $n_1 = n/2$, $\widetilde\lambda = \mu^2  / n$, and $\Lambda = \{ 2^{j-1} \mu^2 / n  :~ 1 \leq j \leq 
\lceil \log_{2} n \rceil + 1
\}$. Choose any $\delta \in (0, 1/e]$. 
\begin{enumerate}
		\item Under Assumption \ref{assumption-covariates-bounded}, there exists a polylogarithmic function $\polylog$ of $(n, n_0, 1/\delta)$ that is independent of other parameters, such that with probability at least $1 - \delta$,
		\begin{align}
\risk ( \widehat{f}  ) - \risk ( f^{\star} ) 
\leq \polylog  \cdot 
\inf_{ \rho > 0 }
\bigg\{
R^2 \rho + \frac{ \sigma^2 }{ \effectivesamplesize } \sum_{j=1}^{\infty} \frac{\mu_j}{\mu_j + \rho}
\bigg\}
+ \zeta \cdot  R^2 \mu^2 \bigg(
\frac{1}{ \effectivesamplesize } + \frac{1}{n_0}
\bigg)
		.
		\label{eqn-excess-risk-0}
		\end{align}
		
		\item Under Assumption \ref{assumption-covariates-subg}, there exist a universal constant $C$ and a polylogarithmic function $\polylog$ of $(n, n_0, 1/\delta)$ determined by $\kappa$, such that when $n_0 \geq C \kappa^2 \log ( \delta^{-1} \log n )$, we have \eqref{eqn-excess-risk-0}. 
	\end{enumerate}
\end{theorem}

We considered an even split of data and a geometrically increasing sequence of penalty parameter with common ratio $2$ only for the ease of presentation. The same results continue to hold so long as the split ratio $n_1 / n$ is bounded away from $0$ and $1$ (e.g.,~0.8), and the common ratio is a constant bounded away from 1 (e.g.,~1.1).

The quantity $\effectivesamplesize$ defined in \eqref{eqn-n-eff} enters the final error bound as the \emph{effective sample size}. 
To see that, imagine running KRR with penalty $\rho $ on $k$ i.i.d.~labeled samples from the target distribution. Under mild regularity assumptions, the expected excess risk is bounded by
\begin{align}
R^2 \rho  +  \frac{\sigma^2}{k} \sum_{j = 1}^{\infty} \frac{\mu_j}{\mu_j + \rho}
	\label{eqn-intuition-2}
\end{align}
up to a constant factor \citep{Zha05}. The two summands correspond to the squared bias and the variance, respectively. Hence, the infimum term in \eqref{eqn-excess-risk-0} corresponds to the optimal bias-variance tradeoff for KRR on $\effectivesamplesize$ labeled target samples. Therefore, \Cref{thm-excess-risk-0} is a rigorous statement of \Cref{thm-intro}, confirming the interpretation of $\effectivesamplesize$. 

The effective sample size measures the information transfer from the source data to the target task through the relation between second-moment operators of $\source$ and $\target$. Below we briefly explain the intuition behind the definition \eqref{eqn-n-eff}. Without further specification, all probabilistic bounds hold up to polylogarithmic factors with high probability.
Define $\bS_{\lambda} = ( \bSigma + \lambda \bI )^{-1/2} \bSigma_0 ( \bSigma + \lambda \bI )^{-1/2}$ and $\bT_{\lambda} = ( \bSigma_0 + \lambda \bI )^{-1/2} \bSigma_0 ( \bSigma_0 + \lambda \bI )^{-1/2} $ for $\lambda > 0$. 
When $n_1 = n / 2$ and $\lambda \geq \mu^2 / n$, we can show that $\widehat{f}_{\lambda}$ in Algorithm \ref{alg} satisfies
\begin{align}
\risk (\widehat{f}_{\lambda}) - \risk (f^{\star})
\lesssim R^2  \lambda \|  \bS_{\lambda}  \|  + \frac{\sigma^2 }{n}  \Tr (  \bS_{\lambda} ).
\label{eqn-intuition-1}
\end{align}
The analysis requires proper concentration of the estimate $\widehat{f}_{\lambda}$ around its expectation, which results in the $O(n^{-1})$ lower bound on $\lambda$. Our choice of $\widetilde{\lambda}$ and the smallest value in $\Lambda$ are based upon this observation.
On the other hand, the infinite sum in \eqref{eqn-intuition-2} is equal to $\Tr (\bT_{\rho})$. To relate \eqref{eqn-intuition-1} to \eqref{eqn-intuition-2}, we will bound $\bS_{\lambda}$ using $\bT_{\rho}$ for some suitable $\rho$. For any $\lambda \in \Lambda$, we use the fact $\lambda \geq \mu^2 / n$ to derive that
\begin{align*}
& \bSigma + \lambda \bI 
\succeq \frac{1}{2} [ \bSigma + (\mu^2 / n) \bI ] + \frac{\lambda}{2}  \bI .
\end{align*}
Define $U = \| [\bSigma + (\mu^2 / n) \bI ]^{-1/2} \bSigma_0 [\bSigma + (\mu^2 / n) \bI ]^{-1/2}  \|$. Then, $\bSigma + (\mu^2 / n) \bI \succeq U^{-1} \bSigma_0$ and
\begin{align*}
\bSigma + \lambda \bI \succeq \frac{1}{2} ( U^{-1} \bSigma_0 + \lambda \bI )
= \frac{1}{2U} (  \bSigma_0 + U \lambda \bI ).
\end{align*}
As a result, we have $\| \bS_{\lambda} \| \leq 2  U \| \bT_{ U \lambda } \| \leq 2  U $ and $\Tr (  \bS_{\lambda} ) \leq 2  U \Tr ( \bT_{ U \lambda } )$. By \eqref{eqn-intuition-1},
\begin{align}
\risk (\widehat{f}_{\lambda}) - \risk (f^{\star})
\lesssim R^2   U \lambda     + \frac{\sigma^2 }{ n / U }  \Tr (  \bT_{ U \lambda} ), \qquad \forall \lambda \in \Lambda.
\end{align}
Comparing this with the bound \eqref{eqn-intuition-2} for KRR on labeled target data, we see that $ n / U $ and $U \lambda$ play the roles of $k$ and $\rho$, respectively. Hence, a candidate for the effective sample size is
\[
 n / U = 
\| (n \bSigma + \mu^2 \bI )^{-1/2} \bSigma_0 (n \bSigma + \mu^2 \bI )^{-1/2}  \|^{-1}
= \sup  \{ t :~ t \bSigma_0 \preceq n \bSigma + 
\mu^2
\bI
\}.
\]
Our final definition \eqref{eqn-n-eff} of $\effectivesamplesize$ truncates the above quantity at $n$ due to technical reasons in bridging the finite grid $\Lambda$ of penalty parameters and the infinite interval $(0, +\infty)$. 

Since $\| \bSigma_0 \| \leq M^2$ always hold, we have $\bSigma_0 \preceq M^2 \bI \preceq \mu^2 \bI$ and thus $1 \leq \effectivesamplesize \leq n$. If $\bSigma_0 \preceq B \bSigma$ for some $B \geq 1$, which is clearly true under the ``bounded likelihood ratio'' assumption $\frac{\rd \target}{\rd \source} \leq B$, then $\effectivesamplesize \geq n / B$. In \Cref{sec-examples}, we will provide concrete examples to further illustrate $\effectivesamplesize$, especially the importance of the additive term $\mu^2 \bI$ in the definition.

\begin{remark}[Effective sample size under covariate shift]
The concept of effective sample size has been used to gauge the transfer from the source to the target. It can be defined based on the coverage of the source distribution over the target one \citep{KMa21,PMW22,SZa22}, or the behaviors of the source and the target distributions near the optimal classification boundary \citep{HKp19,CWe21}. The former is suitable for learning H\"{o}lder smooth functions in classification and regression problems. 
The latter relies on the discrete nature of label in classification problems. 
\end{remark}

In addition to characterizing the transfer, \Cref{thm-excess-risk-0} also shows Algorithm \ref{alg}'s strong ability of model selection. Because the complexity of the function space plays a central role in the bias-variance tradeoff, the infimum term in \eqref{eqn-excess-risk-0} usually has the form $\effectivesamplesize^{-\beta}$ with $0 < \beta < 1$ (in nonparametric regression) or $D \effectivesamplesize^{-1}$ (in linear regression with dimension $D$). In contrast, the additional term in \eqref{eqn-excess-risk-0} is of order $O ( \effectivesamplesize^{-1} + n_0^{-1} )$ regardless of complexities of the function space. Hence, it is usually dominated by its preceding term, as we will see shortly. The size $n_0$ of the unlabeled target data only has a minor influence on the final error, which is negligible if $n_0 \geq \effectivesamplesize$. To verify this in common scenarios, we present a corollary of \Cref{thm-excess-risk-0}. See Appendix B.1 
for the proof.

\begin{corollary}\label{cor-r}
	Consider the setup in \Cref{thm-excess-risk-0}. In addition, assume that $\sigma$, $R$, $M$ and $\kappa$ are bounded from above and below by universal constants. We have the following results when $n$ and $n_0 / \log n$ are sufficiently large, where the inequalities only hide factors that are polylogarithmic in $n$ and $n_0$. Each of them holds with probability at least $1 - n^{-1}$.

	\begin{enumerate}
		\item\label{part-finite-rank-cor-r} If $\rank (\bSigma_0) \leq D < \infty$, then
		\[
		\risk ( \widehat{f} ) - \risk ( f^{\star} ) \lesssim
		D \effectivesamplesize^{-1} + n_0^{-1}.
		\]
		\item\label{part-exponential-cor-r} If $\mu_j \leq c_1 e^{-c_2 j  }$, $\forall j$ holds with some constants $c_1, c_2 > 0$, then
		\[
		\risk ( \widehat{f} ) - \risk ( f^{\star} ) \lesssim \effectivesamplesize^{-1} + n_0^{-1}.
		\]
		\item\label{part-polynomial-cor-r} If $\mu_j \leq c j^{- 2 \alpha}$, $\forall j$ holds with some constants $c > 0$ and $\alpha > 1/2$, then
		\[
		\risk ( \widehat{f} ) - \risk ( f^{\star} ) \lesssim
		\effectivesamplesize^{ - \frac{ 2 \alpha }{ 2 \alpha + 1  } } + n_0^{-1} .
		\]
	\end{enumerate}
\end{corollary}

\Cref{cor-r} presents excess risk bounds when the kernel has finite rank, or decaying eigenvalues with polynomial or exponential rates. A linear kernel in a Euclidean space satisfies the first condition.
Here are examples for the latter two. Let $\cX = [ 0 , 1 ]$ and $\target = \cU [0, 1]$. The first-order Sobolev kernel in \Cref{example-Sobolev} satisfies the polynomial decay property with $\alpha = 1$. The Gaussian kernel in \Cref{example-Gaussian} satisfies the exponential decay property. 

The excess risk bounds continue to hold if the exceptional probability is changed from $n^{-1}$ to $n^{-c}$ with any constant $c> 0$. The $O(n_0^{-1})$ dependence on the target sample size $n_0$ is surprisingly mild. When $n_0 \geq c n$ for some constant $c > 0$, the overhead caused by the finite target data is dominated by the preceding terms in the error bounds. 
Below we comment on the optimality of our results and compare them with the existing literature.

\begin{remark}[Optimality and adaptivity]
The bounds in \Cref{cor-r} are sharp. To see that, consider a simple case where $\frac{\rd \target}{\rd \source} \leq B$ holds for some $B \geq 1$. Then, we have $\effectivesamplesize \geq n / B$. \Cref{cor-r} yields error bounds $DB / n$, $B / n$ and $	(B/n)^{ - \frac{ 2 \alpha }{ 2 \alpha + 1  } }$ for the three scenarios. They are optimal according to the minimax lower bound in Theorem 2 in \cite{MPW22}, which holds for all estimators based on $n$ i.i.d.~labeled source samples plus full knowledge about $\target$ and $B$. Our Algorithm \ref{alg} achieves comparable performance with only $O(n)$ unlabeled samples from $\target$ and no information about $B$ or $\alpha$. This highlights the proposed algorithm's adaptivity to the unknown distribution shift.
In the next subsection, we will see through an example that \Cref{cor-r} can give the optimal error bound even when $\target$ is singular with respect to $\source$ .
\end{remark}

\begin{remark}[Comparisons with existing theories of kernel methods under covariate shift]
	\cite{MPW22} conducted an in-depth analysis of regression in an RKHS when the target and source covariate distributions have bounded likelihood ratio or $\chi^2$ divergence. 
	\begin{itemize}
		\item The former case assumes $\frac{\rd \target}{\rd \source} \leq B$. The authors showed that the minimax optimal error rate can be achieved (up to a logarithmic factor) by KRR, with a penalty parameter determined by the sample size $n$, the spectrum of $\bSigma_0$, and the upper bound $B$ on the likelihood ratio. However, it is generally hard to estimate $B$ given finite samples from $\source$ and $\target$. So, their method is not adaptive to the unknown covariate shift.
		Our self-tuning approach, on the other hand, does not require such knowledge while enjoying optimality guarantees.
		
		\item The latter case assumes $\EE_{x \sim \source} ( \frac{\rd \target}{\rd \source}  )^2 \leq B^2$, which is clearly weaker than the previous condition, plus a uniform bound on the eigenfunctions of the kernel. A reweighted version of KRR using truncated likelihood ratios was shown to be optimal. The $\chi^2$ bound on the distribution shift does not directly translate to our effective sample size $\effectivesamplesize$. Hence, it is not clear whether our method, or its variants, continues to be adaptive in this scenario. We leave this question for future research.
	\end{itemize}
	Both of those cases require $\target$ to be absolutely continuous with respect to $\source$. Such assumption is not needed for our theory, which is built upon the second-moment operators $\bSigma_0$ and $\bSigma$. As a result, we are able to obtain optimality guarantees in some singular cases.
	
	After our preprint had become available on arXiv, a new paper \cite{FHW24} extended the analysis in \cite{MPW22} to general kernel-based methods, including kernel logistic regression. We leave it for future research to study pseudo-labeling approaches beyond least squares regression.
\end{remark}

\subsection{Examples and discussions}\label{sec-examples}

We present two examples to illustrate the impact of covariate shift through explicit calculations of effective sample sizes. The first one is high-dimensional linear regression with $\source$ and $\target$ having different rates of spectral decay.

\begin{example}\label{example-linear-conversion}
Consider \Cref{example-linear} with $\cX= \RR^d$ and $K$ being the linear kernel. Let $R = \sigma = M = 1$, $\source = N( \bm{0}, \bSigma )$, $\target = N(\bm{0}, \bSigma_0)$, and $n_0 \geq n$. Define $\zeta(s) = \sum_{j=1}^{\infty} j^{-s}$ for $s > 1$. Consider the following scenarios with diagonal $\bSigma$ and $\bSigma_0$.
\begin{enumerate}
\item Let $\bSigma = \zeta(2)^{-1} \diag ( 1^{-2}, 2^{-2} ,\cdots, d^{-2} )$ and $\bSigma_0 = \zeta(4)^{-1} \diag ( 1^{-4}, 2^{-4} ,\cdots, d^{-4} )$. Then, $\bSigma_0 \preceq [ \zeta(2) / \zeta(4) ] \bSigma$ and $\effectivesamplesize \asymp n$. Part \ref{part-polynomial-cor-r} in \Cref{cor-r} shows an $\widetilde{O}( n^{-4/5} )$ excess risk bound. 

\item Let $\bSigma = \zeta(4)^{-1} \diag ( 1^{-4}, 2^{-4} ,\cdots, d^{-4} )$ and $\bSigma_0 = \zeta(2)^{-1} \diag ( 1^{-2}, 2^{-2} ,\cdots, d^{-2} )$. We have
\begin{align*}
\frac{n}{\effectivesamplesize}
& = \max \Big\{  \| (\bSigma + n^{-1} \bI)^{-1/2} \bSigma_0 (\bSigma + n^{-1} \bI)^{-1/2} \|
,~ 1 \Big\}
 \\
& \asymp \max_{j \in [d]} \frac{j^{-2}}{j^{-4} + n^{-1}} \leq \frac{1}{ \min_{j \geq 1} \{
	j^{-2} + j^2 / n
	\} }
\asymp \sqrt{n}
\end{align*}
and $\effectivesamplesize \gtrsim \sqrt{n}$. The rates are asymptotically sharp as $d \to \infty$. Part \ref{part-polynomial-cor-r} in \Cref{cor-r} yields an $\widetilde{O} ( n^{-1/3} )$ excess risk bound. 
\end{enumerate}
Thanks to the spectral decays, the above error bounds do not depend on the ambient dimension $d$ and continue to hold as $d \to \infty$. The difference between the convergence rates $n^{-4/5}$ and $n^{-1/3}$ highlights the asymmetric impact of covariate shift. In the first scenario, $\source$ is more spread out than $\target$. Such coverage helps gather rich information about the regression function for the target task. In the second scenario, $\source$ is less diverse than $\target$, which limits the information acquisition and increases the target error.
\end{example}

The second example concerns nonparametric regression, revisiting the challenge discussed in the second part of \Cref{example-Sobolev-challenges}.

\begin{example}\label{example-Sobolev-conversion}
Consider the Sobolev space in \Cref{example-Sobolev}. Let $R = \sigma = M = 1$, $\source = \cU[0, 1]$, and $\target = \delta_{x_0}$ with $x_0 = 1/2$.  Since $\target$ is a point mass, the unlabeled target dataset consists of $n_0$ repetitions of $x_0$. Regardless of $n_0$, Algorithm \ref{alg} always returns $\widehat{f}_{\widehat{\lambda}}$ with
\begin{align}
\widehat\lambda \in \argmin_{ \lambda \in \Lambda } 
| \widehat{f}_{\lambda} ( x_{0} ) - \widetilde{f} (x_0) |.
\label{eqn-alg-Sobolev}
\end{align} 
Hence, we can let $n_0$ be arbitrarily large without affecting the final outcome.  

Since $\rank ( \bSigma_0 ) = 1$, Part \ref{part-finite-rank-cor-r} of \Cref{cor-r} with $n_0 \to \infty$ implies that with probability at least $1 - n^{-1}$, we have
\begin{align}
| \widehat{f} (x_0) - f^{\star} (x_0) | =
\sqrt{ \risk ( \widehat{f} ) - \risk ( f^{\star} ) } \lesssim \effectivesamplesize^{-1/2}.
\label{eqn-alg-Sobolev-2}
\end{align}
This is consistent with the interpretation of $\effectivesamplesize$ as the effective sample size in \Cref{thm-intro}: under $\target$, the regression problem amounts to estimating a single mean parameter $f^{\star}(x_0)$ from i.i.d.~samples, whose optimal error rate is clearly the parametric rate in \eqref{eqn-alg-Sobolev-2}.

On the other hand, we can show that $\effectivesamplesize \asymp \sqrt{n}$. Consequently, $| \widehat{f} (x_0) - f^{\star} (x_0) | \lesssim n^{-1/4}$. This rate is minimax optimal. Hence, our Algorithm \ref{alg} adapts to the extreme covariate shift with $\source$ spread out and $\target$ being a point mass. See Appendix B.2 
for a proof of those claims.

In this example,  the imputation model $\widetilde{f}$ is trained on $\cD_2$ with penalty parameter $\widetilde{\lambda} = n^{-1}$. The pseudo-label $\widetilde{f}(x_0)$ itself is minimax rate-optimal for the linear functional $f^{\star} (x_0)$ of the unknown function $f^{\star}$, see the discussion in \Cref{example-Sobolev-challenges}. Since $\widetilde{\lambda} \in \Lambda$, the candidate $\widehat{f}_{\widetilde\lambda}(x_0)$ obtained from $\cD_1$ achieves the same error rate as $\widetilde{f}(x_0)$. The selection rule \eqref{eqn-alg-Sobolev} implies
\[
| \widehat{f}_{\widehat{\lambda}} (x_0) - \widetilde{f}(x_0) |
\leq 
| \widehat{f}_{\widetilde\lambda} (x_0) - \widetilde{f}(x_0) |
\]
and hence the optimality of the final estimate. The above analysis sheds light on the choice of $\widetilde{\lambda}$ for training the imputation model.
\end{example}

We conclude this section with remarks on the effective sample size and the adaptation to the target task.

\begin{remark}[Necessity of adaptation]\label{remark-adaptivity}\label{lem-LB}
In Algorithm \ref{alg}, we obtain candidate models $\{ \widehat{f}_{\lambda} \}_{\lambda \in \Lambda}$ by running KRR on $\cD_1$, and then select the final one by incorporating information about $\target$. A natural question is whether such adaptation is necessary. In other words, can we choose one model based solely on the source data and apply it to all possible target populations? The numerical example in \Cref{sec-intro} already suggests a negative answer. We use the Sobolev space in \Cref{example-Sobolev} to give a theoretical confirmation.

Let $\source = \cU[0, 1]$, $R = 1$ and assume that $\varepsilon_i \sim N(0, 1)$. \Cref{cor-r} and \Cref{example-Sobolev-conversion} show that when $\target = \cP$ or $\target = \delta_{1/2}$, Algorithm \ref{alg} satisfies the corresponding error bound $\widetilde{O} (n^{-2/3})$ or $\widetilde{O} (n^{-1/2})$, respectively.

Denote by $\risk_1$ or $\risk_2$ the population risk under $\target = \cP$ or $\target = \delta_{1/2}$. There exist universal constants $C_1$ and $C_2$ such that when $n \geq C_1$, we can find $f^{\star} \in \cF$ with $\| f^{\star} \|_{\cF} = 1$ such that
\[
\inf_{\lambda > 0}
\max \bigg\{
\frac{
	\EE [ \risk_1 (\widehat{f}_{\lambda}) - \risk_1 (f^{\star}) ]
}{ n^{-2/3} }
,~
\frac{
	\EE [ \risk_2 (\widehat{f}_{\lambda}) - \risk_2 (f^{\star}) ] 
}{ n^{-1/2} }
\bigg\} \geq C_2 n^{1/15} .
\]
See the proof in Appendix B.3. 
Therefore, any $\widehat{f}_{\lambda}$ is sub-optimal for at least one of the two possible target distributions. The lower bound shows the necessity of adaptation.
\end{remark}

\begin{remark}[Definition of $\effectivesamplesize$]\label{remark-n-eff}
In our definition \eqref{eqn-n-eff} of the effective sample size $\effectivesamplesize$, the additive term $\mu^2
	\bI$ is crucial. In fact, we have $\sup  \{ t :~ t \bSigma_0 \preceq \bSigma\} = 0$ for both the second scenario of \Cref{example-linear-conversion} (with $d \to \infty$) and \Cref{example-Sobolev-conversion}. The proof is deferred to Appendix B.4. 
Consequently, if we remove $\mu^2 \bI$ from the definition, then $\effectivesamplesize = 0$ and it fails to reflect the information transfer. 
\end{remark}

\begin{remark}[Limitations of direct analysis via existing KRR theory]
Define the source risk $\sourcerisk (f) = \EE
| f(\bx) - y |^2
$ with $(\bx, y)$ drawn from the source distribution, and let $\{ \nu_j \}_{j=1}^{\infty}$ be the eigenvalues of $\bSigma$. Since $\widehat{f}_{\lambda}$ is trained on $n$ i.i.d.~labeled samples from the source distribution, one can apply standard results for KRR \citep{Zha05} to show that
\begin{align*}
\min_{ \lambda \in \Lambda} \sourcerisk (\widehat{f}_{\lambda}) - \sourcerisk (f^{\star}) \lesssim 
\inf_{ \lambda > 0 } \bigg\{
R^2 \lambda + \frac{\sigma^2}{n} \sum_{j=1}^{\infty} \frac{\nu_j}{\nu_j + \lambda}
\bigg\}.
\end{align*}
If $\bSigma_0 \preceq B \bSigma$ for some $B$ (e.g.,~when $\frac{\rd \target}{\rd \source} \leq B$), we have
\[
\risk  ( f ) -  \risk  ( f^{\star} ) \leq B [ \sourcerisk (f) - \sourcerisk (f^{\star} ) ] .
\]
Combining the inequalities leads to a bound on the optimal excess target risk 
\begin{align}
	\min_{ \lambda \in \Lambda} \risk (\widehat{f}_{\lambda}) - \risk (f^{\star}) \lesssim 
	B \inf_{ \lambda > 0 } \bigg\{
	R^2 \lambda + \frac{\sigma^2}{n} \sum_{j=1}^{\infty} \frac{\nu_j}{\nu_j + \lambda}
	\bigg\}.
\label{eqn-risk-source}
\end{align}
It is expressed through the spectrum of $\bSigma$ rather than $\bSigma_0$ as in our results. This derivation based on existing KRR theory is quick and clean. Yet, it requires a strong assumption $\bSigma_0 \preceq B \bSigma$ that fails in our second scenario in \Cref{example-linear-conversion} and also \Cref{example-Sobolev-conversion}. Even when the assumption holds, the error bound \eqref{eqn-risk-source} can be loose. Consider the first scenario in \Cref{example-linear-conversion} where $\bSigma_0 \preceq [ \zeta(2) / \zeta(4) ] \bSigma$ and hence $B \asymp 1$. According to \eqref{eqn-risk-source}, the excess risk of the best candidate is bounded by
\[
 \inf_{ \lambda > 0 } \bigg\{
\lambda + \frac{1}{n} \sum_{j=1}^{\infty} \frac{ j^{-2} }{ j^{-2}  + \lambda}
\bigg\}.
\]
Taking $\lambda \asymp n^{-2/3}$ optimizes the right-hand side and gives an $O ( n^{-2/3} )$ error bound. As a comparison, in \Cref{example-linear-conversion} we have shown a tighter bound of order $\widetilde{O}( n^{-4/5} )$.

\end{remark}

\section{Demystifying the power of pseudo-labeling}\label{sec-oracle}

In this section, we will explain why the overhead incurred by imperfect pseudo-labels can be negligible compared to excess risk bounds for the candidate models.
We will present a bias-variance decomposition of this overhead for general settings, which may be of independent interest. The results will be illustrated through numerical simulations.

\subsection{A bias-variance decomposition}

The following oracle inequality is the cornerstone of our excess risk bounds in \Cref{thm-excess-risk-0}.
The proof is deferred to \Cref{sec-thm-oracle-proof}.

\begin{theorem}[Oracle inequality for pseudo-labeling]\label{thm-oracle}
Let Assumptions \ref{assumption-linear-model-rkhs} and \ref{assumption-noise} hold. In addition, suppose that either Assumption \ref{assumption-covariates-bounded} or \ref{assumption-covariates-subg} holds, and $n$ is even. Define $\mu^2$ and $\effectivesamplesize$ as in \Cref{thm-excess-risk-0}. Consider Algorithm \ref{alg} with $n_1 = n/2$, $\widetilde\lambda = \mu^2  / n$, and $\min_{ \lambda \in \Lambda} \lambda \geq \sigma^2 / (R^2 n )$. Choose any $\delta \in (0, 1/e]$ and $\gamma \in ( 0  , 1]$. 
\begin{enumerate}
\item Under Assumption \ref{assumption-covariates-bounded}, there exists a polylogarithmic function $\polylog$ of $(n, n_0, |\Lambda|, \delta^{-1} )$ which is independent of other parameters, such that with probability at least $1 - \delta$,
\begin{align*}
	& \sqrt{	\risk  (\widehat{f}) - \risk  (f^{\star}) }
	\leq  \min_{ \lambda \in \Lambda  } \sqrt{ \risk  ( \widehat{f}_{\lambda} ) - \risk  ( f^{\star} ) } 
	+ \zeta
	\sqrt{
	 R^2 \mu^2 ( \effectivesamplesize^{-1} + n_0^{-1} )
	}.
\end{align*}

\item Under Assumption \ref{assumption-covariates-subg}, there exist a universal constant $C$ and a polylogarithmic function $\polylog $ of $(n, n_0, |\Lambda|, \delta^{-1} )$ determined by $\kappa$, such that when $n_0 \geq ( C \kappa / \gamma )^2 \log ( |\Lambda| / \delta )$, the following inequality holds with probability at least $1 - \delta$:
\begin{align*}
& \sqrt{	\risk  (\widehat{f}) - \risk  (f^{\star}) }
\leq ( 1 + \gamma ) \min_{ \lambda \in \Lambda  } \sqrt{ \risk  ( \widehat{f}_{\lambda} ) - \risk  ( f^{\star} ) } 
	+ \zeta
	\sqrt{
		R^2 \mu^2 ( \effectivesamplesize^{-1} + n_0^{-1} )
	}.
\end{align*}
\end{enumerate}
\end{theorem}

\Cref{thm-oracle} implies that up to a polylogarithmic factor,
\begin{align*}
\risk  (\widehat{f}) - \risk  (f^{\star}) 
\lesssim \min_{ \lambda \in \Lambda  } \{ \risk  ( \widehat{f}_{\lambda} ) - \risk  ( f^{\star} ) \}
+		R^2 \mu^2 ( \effectivesamplesize^{-1} + n_0^{-1} ).
\end{align*}
This result leads to the error bound \eqref{eqn-excess-risk-0} in \Cref{thm-excess-risk-0}.
The additive term in the oracle inequality reflects the error of the imputation model $\widetilde{f}$. Its order $O(\effectivesamplesize^{-1})$ is at most comparable to the preceding term $\min_{ \lambda \in \Lambda  } \{ \risk  ( \widehat{f}_{\lambda} ) - \risk  ( f^{\star} ) \}$ in common scenarios, as we discussed before presenting \Cref{cor-r}.

We now demystify such phenomenon to explain the power of pseudo-labeling. 
Ideally, we want to find the candidate model in $\{ \widehat{f}_{\lambda} \}_{\lambda \in \Lambda}$ that minimizes the risk $\cR (\cdot)$. However, the objective is not directly computable as we only have finitely many unlabeled data from the target distribution. To overcome the challenge, we train an imputation model $\widetilde{f}$ to produce pseudo-labels $\{ \widetilde{y}_{0i} \}_{i=1}^{ n_0 }$ and select the final model $\widehat{f}$ based on \eqref{eqn-oracle-objective}.
The suboptimality of $ \widehat\lambda $ relative to $\argmin_{ \lambda \in \Lambda } \cR ( \widehat{f}_{\lambda} )$ is caused by $\mathrm{(i)}$ the errors of pseudo-labels, and $\mathrm{(ii)}$ the finite-sample approximation of the risk.

To get \Cref{thm-oracle}, we will first tease out the impact of pseudo-labels by studying fixed designs.
Define the in-sample risk on the target data
\begin{align}
	\empiricalrisk   (f) = 	
	\EE \bigg(
	\frac{1}{n_0} \sum_{i=1}^{n_0} 
	|
	f ( \bx_{0i} ) - y_{0i}
	|^2 
	\bigg| \{ \bx_{0i} \}_{i=1}^{n_0}
	\bigg) , \qquad \forall f \in \functionclass .
	\label{eqn-defn-risk-in}
\end{align}
An important property is
\begin{align}
\empiricalrisk   (f)  -	\empiricalrisk   (f^{\star}) = 	\frac{1}{n_0} \sum_{i=1}^{n_0} 
|
f ( \bx_{0i} ) - f^{\star} ( \bx_{0i} )
|^2 .
\label{eqn-in-decomposition}
\end{align}
Therefore, the objective function for model selection \eqref{eqn-oracle-objective} mimics the excess in-sample risk $	\empiricalrisk   ( \widehat{f}_{\lambda} )  -	\empiricalrisk   (f^{\star}) $ of the candidate $\widehat{f}_{\lambda} $. They become equal if the imputation model is perfect, i.e.~$\widetilde{f} = f^{\star}$. We will establish an oracle inequality
\begin{align}
\sqrt{ \empiricalrisk   (\widehat{f})  -	\empiricalrisk   (f^{\star})  }
\leq  \min_{ \lambda \in \Lambda} \sqrt{ \empiricalrisk   (\widehat{f}_{\lambda})  -	\empiricalrisk   (f^{\star})  }
+ \eta ,
\label{eqn-oracle-in-sample}
\end{align}
where $\eta$ quantifies how the discrepancies $\{ \widetilde{f} (\bx_{0i}) - f^{\star} (\bx_{0i}) \}_{i=1}^{n_0}$ influence the suboptimality. See \Cref{lem-in-sample} for a formal statement. Once that is done, we will bridge the in-sample and the out-of-sample excess risks to finish the proof of \Cref{thm-oracle}. See \Cref{sec-thm-oracle-proof} for details.

The in-sample oracle inequality \eqref{eqn-oracle-in-sample} is derived from the following \emph{bias-variance decomposition} of pseudo-labels' impact on model selection, whose proof is deferred to Appendix C. 
Suppose that from $m$ candidate models $\{ g_{j} \}_{j=1}^m$, we want to select one that best fits some unknown regression function $g^{\star}$. We have finitely many unlabeled sample $\{ z_i \}_{i=1}^n$ and use an imputation model $\widetilde{g}$ to fill the labels for model selection. Since $\widetilde{g}$ is often learned from data, we let it be random. All the other quantities are considered deterministic.

\begin{theorem}[Bias-variance decomposition]\label{thm-bias-variance}
Let $\{ z_i \}_{i=1}^n$ be deterministic elements in a set $\cZ$; $g^{\star}$ and $\{ g_{j} \}_{j=1}^m$ be deterministic functions on $\cZ$; $\widetilde{g}$ be a random function on $\cZ$. Define 
\[
\cL_n ( g ) = 
\bigg(
\frac{1}{n} \sum_{i=1}^n | g (z_i) - g^{\star} (z_i) |^2 
\bigg)^{1/2}
\]
for any function $g$ on $\cZ$. Assume that the random vector $\widetilde\by = ( \widetilde{g} (z_1) ,  \widetilde{g} (z_2), \cdots, \widetilde{g} (z_n) )^{\top}$ satisfies \mbox{$\|   \widetilde\by - \EE \widetilde\by   \|_{\psi_2} \leq V < \infty$}.
Choose any
	\begin{align*}
	& \widehat{j} \in \argmin_{ j \in [m] } \bigg\{
	\frac{1}{n} \sum_{i=1}^n | g_{j} (z_i) - \widetilde{g} (z_i) |^2 
	\bigg\} .
	\end{align*}
There exists a universal constant $c > 0$ such that
\begin{align*}
&
\PP \bigg(
 \cL_n ( g_{\widehat{j}}  ) 
	\leq 
\min_{ j \in [m]} \cL_n ( g_j ) 
+ 2 \cL_n ( \EE \widetilde{g} )
+
 c  V \sqrt{ \frac{ \log ( m / \delta ) }{n} } 
\bigg)
\geq 1 - \delta,\qquad  \forall \delta \in (0, 1]  , \\
&	\EE \cL_n^2 ( g_{\widehat{j}}  )
\leq 
\inf_{\gamma > 0}
\bigg\{
 (1+\gamma)  \min_{ j \in [m]} \cL_n^2 ( g_j ) + 2 (1 + \gamma^{-1}) 
	\bigg(
	4  \cL_n^2 ( \EE \widetilde{g} ) +
	\frac{c^2 V^2 ( 1 + \log m )}{n}
	\bigg)
\bigg\} 
	.
\end{align*}
\end{theorem}

\Cref{thm-bias-variance} is an oracle inequality for the in-sample risk of the selected model. The expected overhead consists of two terms, $\cL_n^2 ( \EE \widetilde{g} ) $ and $n^{-1} V^2 ( 1 + \log m ) $. The first term is equal to $n^{-1} \sum_{i=1}^n | \EE \widetilde{g} (z_i) - g^{\star} (z_i) |^2$, the mean squared bias of the pseudo-labels. The second term reflects their randomness, as $V^2$ is the variance proxy of the sub-Gaussian random vector $\widetilde\by$.
Therefore, \Cref{thm-bias-variance} can be viewed as a bias-variance decomposition of the overhead.

In contrast, the usual bias-variance decomposition of mean squared error gives
\begin{align*}
\EE \cL_n^2 ( \widetilde{g} ) = \cL_n^2 ( \EE \widetilde{g} ) + n^{-1} \EE \|  \widetilde\by - \EE  \widetilde\by \|_2^2 .
\end{align*}
The quantity $ \EE \|  \widetilde\by - \EE  \widetilde\by \|_2^2 $ can be as large as $n V^2$, whose implied upper bound on $\EE \cL_n^2 ( \widetilde{g} ) $ is $\cL_n^2 ( \EE \widetilde{g} ) + V^2$. This is much larger than the bound $\cL_n^2 ( \EE \widetilde{g} ) + n^{-1} V^2 (1 + \log  m ) $ on the overhead so long as $ \log m \ll n$.
Consequently, pseudo-labels with large mean squared error may still help select a decent model.

According to our error bound, the (squared) bias $\cL_n^2 ( \EE \widetilde{g} )$ of the pseudo-labels has a more significant impact on model selection compared to the variance proxy $V^2$. The theoretical investigation provides valuable practical guidance. It explains why we use a small regularization parameter $\widetilde{\lambda} \asymp n^{-1}$ and aim for an \emph{undersmoothed} imputation model.

\subsection{A numerical example}

We demonstrate the theoretical findings through a numerical simulation. The corresponding code is available on Github (\url{https://github.com/kw2934/KRR}). 

Consider the RKHS in defined in \Cref{example-Sobolev}. Let $f^{\star} (x) =  \cos ( 2 \pi x ) - 1$, $\nu_0 = \cU [0, 1/2]$ and $\nu_1 = \cU [1/2, 1]$. For any even integer $n \geq 2$, we follow the procedure below:
\begin{itemize}
	\item Let $B = n^{1/3}$, $\source = \frac{B}{B+1} \nu_0 + \frac{1}{B+1} \nu_1$ and $\target = \frac{1}{B+1} \nu_0 + \frac{B}{B+1} \nu_1$. Generate $n$ i.i.d.~source samples $\{ ( x_i, y_i) \}_{i=1}^n$ with $ x_i \sim \source$ and $y_i | x_i \sim N (  f^{\star} (  x_i )  , 1 )$. Generate $n_0 = n$ i.i.d.~unlabeled target samples $\cD_0 = \{ x_{0i} \}_{i=1}^{n_0}$ from $\target$.
	\item Run Algorithm \ref{alg} with $n_1 = n/2$, $\Lambda = \{  2^{k} / (10 n):~ k = 0, 1,  \cdots ,   \lceil \log_{2} ( 10  n ) \rceil   \}$, and $\widetilde{\lambda} = 1 / (10 n)$ to get an estimate $\widehat{f}$.
\end{itemize}

Our method selects a candidate model that minimizes the risk estimate $$f \mapsto \frac{1}{|\cD_0|} \sum_{ x \in \cD_0 } | f(x ) - \widetilde{f} (x) |^2 $$ based on pseudo-labels.
We compare it with an oracle selection approach whose risk estimate $\frac{1}{|\cD_0|} \sum_{ x \in \cD_0 } | f(x ) - f^{\star} (x) |^2 $ is constructed from noiseless labels, and a na\"{i}ve approach that uses the empirical risk $\frac{1}{|\cD_2|} \sum_{ (x, y) \in \cD_2 } | f(x) - y |^2 $ on $\cD_2$. 
The latter is the standard hold-out validation method that ignores the covariate shift.

\begin{figure}[t]
	\centering
	\includegraphics[width=0.5\linewidth]{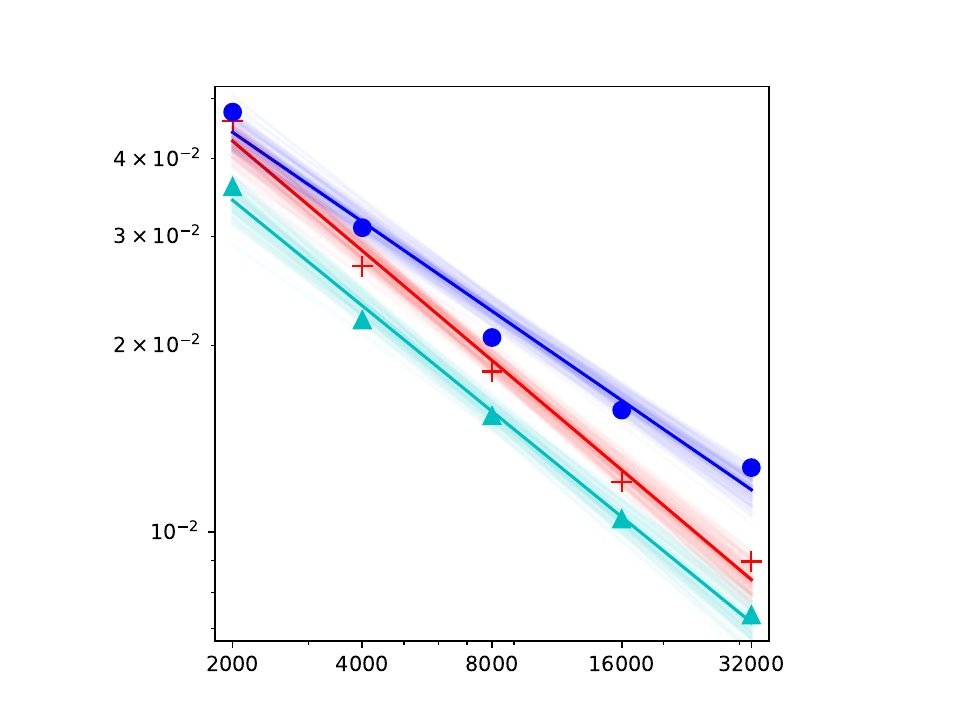}   
	\caption{Comparisons of three approaches on a log-log plot. $x$-axis: $n$. $y$-axis: mean squared error. Red crosses: pseudo-labeling. Cyan triangles: the oracle method. Blue circles: the na\"{i}ve method.}\label{fig-krr}
\end{figure}

For every $n \in \{ 2000, 4000, 8000, 16000, 32000 \}$, we conduct 100 independent runs of the experiment. \Cref{fig-krr} shows the performances of three model selection approaches: pseudo-labeling (red crosses), the oracle method (cyan triangles) and the na\"{i}ve method (blue circles). The $x$- and $y$-coordinates of each point are the sample size $n$ and the average of estimated excess risks over 100 independent runs, respectively. 
Since the points of each method are lined up on the log-log scale, the excess risk decays like $O( n^{-\alpha} )$ for some $\alpha > 0$. 
Linear regression yields the solid lines in \Cref{fig-krr} and estimates the $\alpha$'s for the pseudo-labeling, oracle, and na\"{i}ve methods as $0.587~(0.029)$, $
0.565~(0.034)$ and $0.478~(0.030)$, respectively. The numbers in the parentheses are standard errors estimated by cluster bootstrap \citep{CGM08}, which resamples the $100$ data points corresponding to each $n$ for $10000$ times.
Our proposed method does not differ significantly from the oracle one in terms of the error exponent. It is significantly better than the na\"{i}ve method, as the source data does not have a good coverage of the test cases. 
In addition, we also plot the linear fits of 100 bootstrap replicates to illustrate the stability, resulting in the thin shaded areas around the solid lines in \Cref{fig-krr}.

\section{Proofs of \Cref{thm-excess-risk-0,thm-oracle}}\label{sec-proofs-main}

We now prove our main results: the excess risk bounds (\Cref{thm-excess-risk-0}) and the oracle inequality for model selection (\Cref{thm-oracle}). We will first conduct a fixed-design analysis in \Cref{sec-proof-fixed} by treating the covariates and the data split as deterministic.
Based on that, we will analyze the selected model in \Cref{sec-selected}. Finally, we will prove \Cref{thm-oracle} in \Cref{sec-thm-oracle-proof}, followed by \Cref{thm-excess-risk-0} in \Cref{sec-thm-excess-risk-0-proof}.

\subsection{Fixed-design analysis}\label{sec-proof-fixed}

To begin with, we tease out the impact of random errors $\{ \varepsilon_i \}_{i=1}^n$ on our estimator $\widehat{f}$ by conditioning on the covariates $\{ \bx_i \}_{i=1}^{n}$, $\{ \bx_{0i} \}_{i=1}^{n_0}$ and the split $\cD_1 \cup \cD_2$. 
As a result, we will study the in-sample risk $\empiricalrisk$ defined in \eqref{eqn-defn-risk-in} and provide excess risk bounds as well as oracle inequalities.

To facilitate the analysis, we introduce some notations.
When $\cX = \RR^d$ for some $d < \infty$ and $K(\bz, \bw) = \bz^{\top} \bw$, $\phi$ is the identity mapping. We can construct design matrices $\bX = (\bx_1,\cdots, \bx_n)^{\top} \in \RR^{n \times d}$, $\bX_0 = (\bx_{01},\cdots, \bx_{0 n_0})^{\top} \in \RR^{n_0 \times d}$ and the response vector $\by = (y_1 , \cdots , y_n)^{\top} \in \RR^n$. Define the index set $\cT = \{ i \in [n] :~ (\bx_i, y_i) \in \cD_1  \}$ of the data for training candidate models.
In addition, let $n_1 = |\cT|$ and $n_2 = n - n_1$. Denote by $\bX_1 \in \RR^{ |\cT| \times d} $ and $\bX_2 \in \RR^{ (n - |\cT|) \times d }$ the sub-matrices of $\bX$ by selecting rows whose indices belong to $\cT$ and $[n] \backslash \cT$, respectively. Similarly, denote by $\by_1 \in \RR^{ |\cT| }$ and $\by_2 \in \RR^{ n - |\cT| }$ the sub-vectors of $\by$ induced by those index sets.

The design matrices can be generalized to operators when $( \cX,  K )$ defines a general RKHS. In that case, $\bX$ and $\bX_0$ are bounded linear operators defined through
\begin{align*}
& \bX:~ \HH \to \RR^n,~~ \btheta \mapsto ( \langle \phi( \bx_1 ) , \btheta \rangle , ~ \cdots, ~
\langle  \phi( \bx_n ) , \btheta \rangle
)^{\top} ,\\
& \bX_0:~ \HH \to \RR^{n_0} ,~~ \btheta \mapsto ( \langle \phi( \bx_{01} ) , \btheta \rangle , ~ \cdots, ~
\langle  \phi( \bx_{0 n_0} ) , \btheta \rangle
)^{\top} .
\end{align*}
We can define $\bX_1$ and $\bX_2$ similarly. 
With slight abuse of notation, we use $\bX^{\top}$ to refer to the adjoint of $\bX$, i.e. $\bX^{\top}:~ \RR^n \to \HH$, $\bu \mapsto \sum_{i=1}^{n} u_i \phi (\bx_i)$. Similarly, we can define $\bX_0^{\top}$, $\bX_1^{\top}$ and $\bX_2^{\top}$.
Let $\widehat\bSigma_0 = \frac{1}{n_0} \sum_{i=1}^{n_0} \phi (\bx_{0i}) \otimes \phi ( \bx_{0i} )$, $\widehat\bSigma_1 = \frac{1}{n_1} \sum_{i \in \cT} \phi (\bx_{i}) \otimes \phi ( \bx_{i} )$ and $\widehat\bSigma_2 = \frac{1}{n_2} \sum_{i \in [n] \backslash \cT} \phi (\bx_{i}) \otimes \phi ( \bx_{i} )$. We have $\widehat\bSigma_j = n_j^{-1} \bX_j^{\top} \bX_j$ for all $j \in \{  0 , 1 , 2  \}$. Moreover, $\widehat{f}_{\lambda} ( \cdot ) = \langle \phi (\cdot ) , \widehat{\btheta}_{\lambda} \rangle $ holds for $ \widehat\btheta_{\lambda} = (
\widehat\bSigma_1 + \lambda \bI)^{-1} ( n_1^{-1} \bX_1^{\top} \by_1 )$, and 
$\widetilde{f} ( \cdot ) = \langle \phi (\cdot ) , \widetilde{\btheta}  \rangle $ holds for $ \widetilde\btheta  = ( \widehat\bSigma_2 + \widetilde\lambda \bI)^{-1} (n_2^{-1} \bX_2^{\top} \by_2 ) $.

Below we present a lemma on kernel ridge regression that plays a major role throughout the analysis. See Appendix A.1 
for its proof.

\begin{lemma}[Fixed-design KRR]\label{lem-error-bounds}
	Let Assumptions \ref{assumption-linear-model-rkhs} and \ref{assumption-noise} hold, except that all covariates are deterministic and the data split is fixed.
	Choose any $\lambda > 0$ and define $\bar\btheta_{\lambda} = \EE  \widehat\btheta_{\lambda} $.
		We have $\| \bar{\btheta}_{\lambda} - \btheta^{\star} \|_{\HH} \leq R$, $\| \widehat{\btheta}_{\lambda} - \bar{\btheta}_{\lambda} \|_{\psi_2} \leq   \sigma / \sqrt{n_1 \lambda} $,
\begin{align*}
&  \| \bX_0 ( \bar\btheta_{\lambda} - \btheta^{\star} ) \|_2^2 
\leq \lambda^2 R^2 \| \bX_0 (\widehat{\bSigma}_1 + \lambda \bI)^{-2} \bX_0^{\top} \| , \\
&  \|
	\bX_0 (  \widehat\btheta_{\lambda} - \bar\btheta_{\lambda}  ) 
\|_{\psi_2}^2
\leq \frac{ \sigma^2}{ n_1} \|  \bX_0 (\widehat{\bSigma}_1 + \lambda \bI)^{-1} \widehat{\bSigma}_1 (\widehat{\bSigma}_1 + \lambda \bI)^{-1}  \bX_0^{\top}  \| .
\end{align*}
On the other hand, choose any positive semi-definite trace-class linear operator $\bQ:~\HH \to \HH$ and define $\widehat\bS_{\lambda} =  (\widehat{\bSigma}_1 + \lambda \bI)^{-1/2} \bQ (\widehat{\bSigma}_1 + \lambda \bI)^{-1/2} $. There is a universal constant $C$ such that
\[
\PP \bigg(
\| \widehat\btheta_{\lambda} - \btheta^{\star} \|_{\bQ}^2
\leq 
2 \lambda R^2 \| \widehat\bS_{\lambda} \|  + \frac{ C \sigma^2 \Tr( \widehat\bS_{\lambda} ) \log(1/\delta) }{n_1} 
\bigg)
\geq 1 - \delta
,\qquad \forall  \delta \in (0, 1/e].
\]
\end{lemma}

\Cref{lem-error-bounds} yields error bounds that involve empirical second-moment operators. The following lemma relates them to the population versions, whose proof is in Appendix A.2. 

\begin{lemma}[Concentration of second-moment operators]\label{lem-cov-intermediate}
	Choose any $\delta \in (0, 1/e]$ and let Assumption \ref{assumption-linear-model-rkhs} hold. Define $\bS_{\lambda} =   ( \bSigma + \lambda \bI)^{-1/2} \bSigma_0 ( \bSigma +  \lambda \bI)^{-1/2}$ for $\lambda > 0$.
	\begin{enumerate}
		\item Let Assumption \ref{assumption-covariates-bounded} hold and define $\mu_j = \frac{ 16 M^2 \log ( 14 n_j / \delta)  }{ n_j }$ for $j \in \{ 0, 1, 2 \}$. We have
		\begin{align*}
			& \PP \bigg(
			\widehat\bSigma_0 \preceq \frac{3}{2} 
			\bSigma_0 +\frac{\mu_0}{2}  \bI
			\bigg) \geq 1 - \delta 
			, \\
			& \PP \bigg(
			\widehat\bSigma_j + \lambda \bI \succeq	\frac{\min \{ \lambda / \mu_j, 1 \} }{2} ( \bSigma + \lambda \bI ) , ~ \forall \lambda \geq 0
			\bigg) \geq 1 - \delta,\qquad j \in \{ 1, 2 \} .
		\end{align*}
		
		\item Let Assumption \ref{assumption-covariates-subg} hold. Define $\mu_j = \frac{ \Tr (\bSigma) \log (1/ \delta)   }{ n_j } $ for $j \in \{ 1, 2 \}$ and $\mu_0 = \frac{ \Tr (\bSigma_0) \log (1 / \delta)   }{ n_0 } $. There exists a constant $C \geq 1$ determined by $\kappa$ such that
		\begin{align*}
			& \PP \bigg(
			\widehat\bSigma_0 \preceq \frac{3}{2} 
			\bSigma_0 + C \mu_0 \bI
			\bigg) \geq 1 - \delta 
			, \\
			& \PP \bigg(
			\widehat\bSigma_j + \lambda \bI \succeq	C^{-1} \min \{ \lambda / \mu_j, 1 \}  ( \bSigma + \lambda \bI ) , ~ \forall \lambda \geq 0
			\bigg) \geq 1 - \delta,\qquad j \in \{ 1, 2 \} .
		\end{align*}
	\end{enumerate}
\end{lemma}

\subsection{Analysis of the selected candidate}\label{sec-selected}

From \Cref{lem-error-bounds,lem-cov-intermediate}, we obtain an excess risk bound for the best candidate in $\{ \widehat{f}_{\lambda} \}_{\lambda \in \Lambda}$.

\begin{lemma}[Excess risk of the best candidate]\label{lem-excess-risk}
	Let Assumptions \ref{assumption-linear-model-rkhs} and \ref{assumption-noise} hold. In addition, suppose that either Assumption \ref{assumption-covariates-bounded} or \ref{assumption-covariates-subg} hold. Choose any $\delta \in (0, 1/e]$. For any $\lambda > 0$, define $\bS_{\lambda} =   ( \bSigma + \lambda \bI)^{-1/2} \bSigma_0 ( \bSigma +  \lambda \bI)^{-1/2}$ and
	\begin{align*}
		& 	\cE (\lambda , \delta ) = 
		\max \bigg\{ 
		\frac{M^2 \log(n_1 /\delta) }{n_1 \lambda}  , 1  \bigg\} 
		\bigg(
		\lambda R^2 \| \bS_{\lambda} \|  + \frac{ \sigma^2 \Tr( \bS_{\lambda} ) \log(|\Lambda|/\delta) }{n_1} 
		\bigg) .
	\end{align*}
	There exists a universal constant $C$ such that with probability at least $1 - \delta$,
	\begin{align*}
		\min_{ \lambda \in \Lambda}  	\risk ( \widehat{f}_{\lambda} ) - \risk ( f^{\star} )
		\leq C \min_{ \lambda \in \Lambda }	\cE (\lambda , \delta ) .
	\end{align*}
\end{lemma}

\begin{proof}[\bf Proof of \Cref{lem-excess-risk}]	
	Choose any $\delta \in (0, 1/e] $ and $\gamma \in (0, 1 ]$. To analyze $	\risk ( \widehat{f}_{\lambda} ) - \risk ( f^{\star} ) $, we apply \Cref{lem-error-bounds} (with $\bQ = \bSigma_0$) and union bounds to get
	\[
	\PP \bigg(
	\risk ( \widehat{f}_{\lambda} ) - \risk ( f^{\star} ) 
	\leq 
	2 \lambda R^2 \| \widehat\bT_{\lambda} \|  + \frac{ C_1 \sigma^2 \Tr( \widehat\bT_{\lambda} ) \log(2 |\Lambda|/\delta) }{n_1} ,~\forall \lambda \in \Lambda
	\bigg)
	\geq 1 - \delta/2 ,
	\]
	where $\widehat\bT_{\lambda} =  (\widehat{\bSigma}_1 + \lambda \bI)^{-1/2} \bSigma_0 (\widehat{\bSigma}_1 + \lambda \bI)^{-1/2} $, and $C_1$ is a universal constant.
	
	Define $\mu_1 = \frac{ 16 M^2 \log ( 28 n_1 / \delta)  }{ n_1 }$. \Cref{lem-cov-intermediate} implies that with probability at least $1 - \delta/2$,
	\[
	(\widehat\bSigma_1 + \lambda \bI)^{-1} \preceq	2 \max \{ \mu_1 / \lambda  , 1  \}  ( \bSigma + \lambda \bI )^{-1} , ~ \forall \lambda > 0.
	\]
	On that event,
	\begin{align*}
		\| \widehat\bT_{\lambda} \| & = \|  \bSigma_0^{1/2} (\widehat{\bSigma}_1 + \lambda \bI)^{-1} \bSigma_0^{1/2} \| 
		\leq 2 \max \{ \mu_1 / \lambda  , 1  \}
		\|  \bSigma_0^{1/2} (\bSigma + \lambda \bI)^{-1} \bSigma_0^{1/2} \| \\
		& 
		\lesssim \max \bigg\{ 
		\frac{M^2 \log(n_1 /\delta) }{n_1 \lambda}  , 1  \bigg\} 
		\|  \bS_{\lambda} \| 
	\end{align*}
	and similarly,
	\[
	\Tr (\widehat\bT_{\lambda} ) \lesssim \max \bigg\{ 
	\frac{M^2 \log(n_1 /\delta) }{n_1 \lambda}  , 1  \bigg\}
	\Tr (   \bS_{\lambda} ),
	\]
	where $\lesssim$'s only hide universal constants. Hence, there exists a universal constant $C$ such that
	\begin{align*}
		\PP \bigg(
		\risk ( \widehat{f}_{\lambda} ) - \risk ( f^{\star} ) 
		\leq C 	\cE (\lambda , \delta ) 
		,~\forall \lambda \in \Lambda
		\bigg)
		\geq 1 - \delta.
	\end{align*}
	On that event, we have the desired excess risk bound.
\end{proof}

We now come to our selected model $\widehat{f}$. To relate it to the best candidate, we first obtain from \Cref{lem-error-bounds} and \Cref{thm-bias-variance} an oracle inequality of the form \eqref{eqn-oracle-in-sample} on the in-sample risk. 
The proof is deferred to Appendix A.3. 

\begin{lemma}[Oracle inequality for in-sample risk]\label{lem-in-sample}
	Let Assumptions \ref{assumption-linear-model-rkhs} and \ref{assumption-noise} hold, except that all covariates are deterministic and the data split is fixed. Suppose that $\widetilde\lambda > 0$.
	For any $\lambda > 0$ and $\delta \in (0, 1/e]$, define 
\[
\widehat\xi  ( \widetilde\lambda , \delta ) = 
\log(|\Lambda|/\delta) 
\max \{ \sigma^2 / n_2 ,    R^2  \widetilde\lambda \}
\|  (\widehat{\bSigma}_2 + \widetilde\lambda \bI)^{-1/2} \widehat\bSigma_0 (\widehat{\bSigma}_2 + \widetilde\lambda \bI)^{-1/2}  \|.
\]
	There exists a universal constant $C$ such that with probability at least $1 - \delta$, we have
	\begin{align*}
&
\sqrt{	\empiricalrisk  (\widehat{f}) - \empiricalrisk  (f^{\star}) }
	\leq   \min_{ \lambda \in \Lambda  } \sqrt{ \empiricalrisk  ( \widehat{f}_{\lambda} ) - \empiricalrisk  ( f^{\star} ) } +   C \sqrt{ \widehat\xi  ( \widetilde{\lambda} , \delta ) } .
	\end{align*}
\end{lemma}

Next, we bridge the in-sample risk $\empiricalrisk$ and the out-of-sample risk $\risk$ by incorporating the randomness of the unlabeled target data $\{ \bx_{0i} \}_{i=1}^{n_0}$. See Appendix A.4 
for the proof.

\begin{lemma}[Bridging in-sample and out-of-sample risks]\label{lem-sandwich}
Let Assumptions \ref{assumption-linear-model-rkhs} and \ref{assumption-noise} hold, except that the source covariates $\{ \bx_{i} \}_{i=1}^{n}$ are deterministic and the data split is fixed.
	Choose any $\delta \in ( 0  , 1/e ]$ and $\gamma \in (0 , 1 ]$. 
	\begin{enumerate}
		\item Suppose that Assumption \ref{assumption-covariates-bounded} holds. Define $\bar{R} = \max \{ R, \sigma /  \sqrt{n_1 \min_{ \lambda \in \Lambda  } \lambda} \} $. There exists a universal constant $C$ such that with probability at least $1 - \delta$, 
		\begin{align*}
		\max_{\lambda \in \Lambda}
		\Big|
		\sqrt{ \risk ( \widehat{f}_{\lambda} ) - \risk ( f^{\star} ) } - 
		\sqrt{	\empiricalrisk ( \widehat{f}_{\lambda} ) - \empiricalrisk ( f^{\star} ) } 
		\Big|
		&		\leq C M \bar{R} 
		\sqrt{
			\frac{
				\log ( n_0 |\Lambda| / \delta ) \log (|\Lambda| / \delta)
			}{n_0}
		}.
		\end{align*}

		\item Suppose that Assumption \ref{assumption-covariates-subg} holds. There exists a universal constant $C$ such that when $n_0 \geq ( C \kappa / \gamma )^2 \log ( 2 |\Lambda| / \delta )$, the following inequality holds with probability at least $1 - \delta$:
		\begin{align*}
		( 1 - \gamma) [ \risk ( \widehat{f}_{\lambda} ) - \risk ( f^{\star} ) ]
		\leq  	\empiricalrisk ( \widehat{f}_{\lambda} ) - \empiricalrisk ( f^{\star} )
		\leq ( 1 +  \gamma) [ \risk ( \widehat{f}_{\lambda} ) - \risk ( f^{\star} ) ] , \qquad \forall \lambda \in \Lambda.
		\end{align*}
	\end{enumerate}
\end{lemma}

From \Cref{lem-in-sample} and \Cref{lem-sandwich}, we immediately get the following oracle inequalities for out-of-sample risks. The proof is omitted.

\begin{corollary}[Oracle inequalities for out-of-sample risks]\label{cor-oos}
Let Assumptions \ref{assumption-linear-model-rkhs} and \ref{assumption-noise} hold, except that the covariates are deterministic and the data split is fixed. Let $\widehat\xi$ be defined as in \Cref{lem-in-sample}.
	Choose any $\delta \in ( 0  , 1/e ]$ and $\gamma \in (0 , 1 ]$. 
	\begin{enumerate}
		\item Suppose that Assumption \ref{assumption-covariates-bounded} holds. Define $\bar{R} = \max \{ R, \sigma /  \sqrt{n_1 \min_{ \lambda \in \Lambda  } \lambda} \} $. There exists a universal constant $C$ such that with probability at least $1 - 2 \delta$, 
		\begin{align*}
		\sqrt{	\risk  (\widehat{f}) - \risk  (f^{\star}) }
		& \leq   \min_{ \lambda \in \Lambda  } \sqrt{ \risk  ( \widehat{f}_{\lambda} ) - \risk  ( f^{\star} ) } \\
		&~~~~ +   C
		\bigg(
		 \sqrt{ \widehat\xi  ( \widetilde{\lambda} , \delta ) } 
		+ M \bar{R} 
		\sqrt{
			\frac{
				\log ( n_0 |\Lambda| / \delta ) \log (|\Lambda| / \delta)
			}{n_0}
		}
	\bigg)
		.
		\end{align*}
		
		\item Suppose that Assumption \ref{assumption-covariates-subg} holds. There exist universal constants $C_1$ and $C_2$ such that when $n_0 \geq ( C_1 \kappa / \gamma )^2 \log ( 2 |\Lambda| / \delta )$, the following inequality holds with probability at least $1 - 2 \delta$:
	\begin{align*}
&
 \sqrt{	\risk  (\widehat{f}) - \risk  (f^{\star}) }
\leq (  1 + \gamma ) \min_{ \lambda \in \Lambda  } \sqrt{ \risk  ( \widehat{f}_{\lambda} ) - \risk  ( f^{\star} ) } +  C_2
\sqrt{ \widehat\xi  ( \widetilde{\lambda} , \delta ) } 
.
\end{align*}
	\end{enumerate}
\end{corollary}

The above result can be viewed as an empirical version of \Cref{thm-oracle}, since the additive terms in the oracle inequalities are functions of the random covariates. Equipped with \Cref{cor-oos} and \Cref{lem-cov-intermediate}, we are ready to tackle \Cref{thm-oracle}. After finishing that, we will derive \Cref{thm-excess-risk-0} from \Cref{thm-oracle} (oracle inequality) and \Cref{lem-excess-risk} (excess risk bound for the best candidate).

\subsection{Proof of \Cref{thm-oracle}}\label{sec-thm-oracle-proof}

In this proof, we use $\lesssim$ to hide polylogarithmic factors. Suppose that Assumption \ref{assumption-covariates-bounded} holds. \Cref{cor-oos} asserts that with probability at least $1 - \delta/2$, 
\begin{align}
\sqrt{	\risk  (\widehat{f}) - \risk  (f^{\star}) }
-  \min_{ \lambda \in \Lambda  } \sqrt{ \risk  ( \widehat{f}_{\lambda} ) - \risk  ( f^{\star} ) } 
& \lesssim
\sqrt{ \widehat\xi  ( \widetilde{\lambda} , \delta / 4 ) } 
+ 
	\frac{
MR
	}{\sqrt{n_0}
}
,
\label{eqn-thm-oracle-01}
\end{align}
We now work on $\widehat\xi  ( \widetilde{\lambda} , \delta )$. Define $\mu_j = \frac{ 16 M^2 \log ( 56 n_j / \delta)  }{ n_j }$ for $j \in \{ 0, 2 \}$. Part 1 of \Cref{lem-cov-intermediate} implies that with probability at least $1 - \delta/2$, we have
\begin{align}
\widehat\bSigma_0 \preceq \frac{3}{2} 
\bSigma_0 +\frac{\mu_0}{2}  \bI
\qquad\text{and}\qquad
\widehat\bSigma_2 + \widetilde\lambda \bI \succeq	\frac{\min \{ \widetilde\lambda / \mu_2, 1 \} }{2} ( \bSigma + \widetilde\lambda \bI ) .
\label{eqn-thm-oracle-0}
\end{align}
Since $\widetilde\lambda / \mu_2  \gtrsim \frac{1}{\log (n / \delta) }$, on the event \eqref{eqn-thm-oracle-0}, we have
\begin{align*}
& \|   (\widehat{\bSigma}_2 + \widetilde\lambda \bI)^{-1/2} \widehat\bSigma_0 (\widehat{\bSigma}_2 +  \widetilde\lambda \bI)^{-1/2} \| 
 = \| \widehat\bSigma_0^{1/2}  (\widehat{\bSigma}_2 + \widetilde\lambda \bI)^{-1} \widehat\bSigma_0^{1/2} \| \\
& \lesssim  \|  \widehat\bSigma_0^{1/2}  (\bSigma + \widetilde\lambda \bI)^{-1} \widehat\bSigma_0^{1/2}  \| 
= \|   (\bSigma + \widetilde\lambda \bI)^{-1/2} \widehat\bSigma_0 
 (\bSigma + \widetilde\lambda \bI)^{-1/2} 
 \| \\
& \lesssim  \|   (\bSigma + \widetilde\lambda \bI)^{-1/2} 
 (
\bSigma_0 + \mu_0 \bI
 )
(\bSigma + \widetilde\lambda \bI)^{-1/2} 
 \| \\
& =   \|  (\bSigma + \widetilde\lambda \bI)^{-1/2} 
\bSigma_0 
(\bSigma + \widetilde\lambda \bI)^{-1/2}  \| + \mu_0 / \widetilde\lambda   \\
& = \|  \bSigma_0^{1/2} (\bSigma + \widetilde\lambda \bI)^{-1} \bSigma_0^{1/2}  \| + \mu_0 / \widetilde\lambda 
 = n / \effectivesamplesize + \mu_0 / \widetilde\lambda
\lesssim n \bigg( \frac{1}{\effectivesamplesize} + \frac{1}{n_0} \bigg)
\end{align*}
and then, $\widehat\xi  ( \widetilde\lambda , \delta/ 4 )  \lesssim R^2 \mu^2 ( \effectivesamplesize^{-1} + n_0^{-1} )$. Plugging this into \eqref{eqn-thm-oracle-01} finishes Part 1.

The proof of Part 2 is similar to the above. We will do it briefly. Let Assumption \ref{assumption-covariates-subg} hold. Choose any $\delta,\gamma \in ( 0  , 1 ]$. According to Part 2 of \Cref{cor-oos}, there exist universal constants $C$ and $C_1$ such that when $n_0 \geq ( C \kappa / \gamma )^2 \log ( |\Lambda| / \delta )$, the following inequality holds with probability at least $1 -  \delta/2$:
\begin{align}
\sqrt{	\risk  (\widehat{f}) - \risk  (f^{\star}) }
\leq (  1 + \gamma ) \min_{ \lambda \in \Lambda  } \sqrt{ \risk  ( \widehat{f}_{\lambda} ) - \risk  ( f^{\star} ) } +  C_1
\sqrt{ \widehat\xi  ( \widetilde{\lambda} , \delta / 4 ) } 
.
\label{eqn-thm-oracle-02}
\end{align}
Using Part 2 of \Cref{lem-cov-intermediate}, it is easy to show that $\widehat\xi  ( \widetilde\lambda , \delta/ 4 ) \lesssim R^2 \mu^2 ( \effectivesamplesize^{-1} + n_0^{-1} )$ holds with probability at least $1-\delta/2$, where $\lesssim$ hides a polylogarithmic factor determined by $\kappa$. 
The proof is completed by combining that and \eqref{eqn-thm-oracle-02}.

\subsection{Proof of \Cref{thm-excess-risk-0}}\label{sec-thm-excess-risk-0-proof}

Throughout the proof, we use $\lesssim$ to hide polylogarithmic factors. 
Define $\conversionrate = \effectivesamplesize / n$. 
Sort the elements in $\Lambda$ in ascending order and get $\{ \lambda_j \}_{j=1}^m$.
Let $\cE (\lambda , \delta ) $ be defined as in \Cref{lem-excess-risk}.
Since $\effectivesamplesize \bSigma_0 \preceq  n \bSigma + \mu^2 \bI $, the followings hold for any $\lambda \geq \mu^2 / n$: $\bSigma + \lambda \bI \succeq \conversionrate \bSigma_0$, $\| \bS_{\lambda} \| \leq \conversionrate^{-1}$, and
\begin{align*}
& 2 (\bSigma + \lambda \bI) \succeq [ \bSigma + ( \mu^2 / n) \bI ] + \lambda \bI
\succeq \conversionrate \bSigma_0 + \lambda \bI
= \conversionrate ( \bSigma_0 + \conversionrate^{-1} \lambda \bI ) ,\\
& \Tr ( \bS_{\lambda} ) \preceq 2 \conversionrate^{-1} \Tr [ ( \bSigma_0 + \conversionrate^{-1} \lambda \bI )^{-1} \bSigma_0 ].
\end{align*}
Using the facts that $n_1 = n_2 = n/2$, $ \widetilde\lambda = \lambda_1 = \mu^2 / n$ and $|\Lambda| = \lceil \log_2 n \rceil + 1$, we get
\begin{align*}
& \cE (\lambda , \delta ) 
\lesssim   \conversionrate^{-1} \lambda R^2 + \frac{ \sigma^2 \Tr [ ( \bSigma_0 + \conversionrate^{-1} \lambda \bI )^{-1} \bSigma_0 ]  }{ \effectivesamplesize }  
\end{align*}
for all $\lambda \geq \mu^2 / n$. Then, \Cref{lem-excess-risk} implies that with probability at least $1 - \delta$,
\begin{align*}
\min_{ \lambda \in \Lambda}  	\risk ( \widehat{f}_{\lambda} ) - \risk ( f^{\star} )
\lesssim \effectivesamplesize^{-1}  \min_{ \lambda \in \Lambda }
\Big\{
n R^2  \lambda +  \sigma^2 \Tr [ ( \bSigma_0 + \conversionrate^{-1} \lambda \bI )^{-1} \bSigma_0 ]  
\Big\}
 .
\end{align*}

Meanwhile, \Cref{thm-oracle} with $\gamma = 1$ implies that with probability at least $1 - \delta$,
\begin{align*}
	\risk ( \widehat{f} ) - \risk ( f^{\star} )
	& \lesssim \min_{ \lambda \in \Lambda  } \{  \risk ( \widehat{f}_{\lambda} ) - \risk ( f^{\star} ) \} 
	+  R^2 \mu^2   (
\effectivesamplesize^{-1} +  n_0^{-1}
	 )
	.
\end{align*}
By applying union bounds and adjusting $\delta$, we obtain that with probability at least $1 - \delta$,
\begin{align*}
	\risk ( \widehat{f}  ) - \risk ( f^{\star} )
&	\lesssim  
\effectivesamplesize^{-1}  \min_{ \lambda \in \Lambda }
\Big\{
\underbrace{
n R^2  \lambda
}_{A(\lambda)}
+ 
\underbrace{
 \sigma^2 \Tr [ ( \bSigma_0 + \conversionrate^{-1} \lambda \bI )^{-1} \bSigma_0 ]  
}_{B(\lambda)}
\Big\}
	+ R^2 \mu^2  ( \effectivesamplesize^{-1} + n_0^{-1}
)
	.
\end{align*}
We will finish the proof by establishing the following connection between the finite grid and the infinite interval:
\begin{align}
	\min_{ \lambda \in \Lambda } \{ A (\lambda) + B (\lambda) \} \lesssim  \inf_{ \lambda > 0 } \{ A (\lambda) + B (\lambda) \} +  R^2 \mu^2  .
\label{eqn-grid-interval}
\end{align}

The functions $A$ and $B$ are increasing and decreasing, respectively. From $\lambda_1 = \mu^2 / n$ we get $A(\lambda_1) = R^2 \mu^2 $ and
\begin{align}
	\inf_{ 0 < \lambda \leq \lambda_1 } \{  A (\lambda) + B (\lambda) \} +  R^2 \mu^2  
	\geq B(\lambda_1) + A(\lambda_1).
	\label{eqn-lem-excess-risk-1}
\end{align}
From $\mu^2 \leq \lambda_m \leq 2 \mu^2$ we get
\begin{align*}
	\inf_{\lambda \geq \lambda_m} \{ A(\lambda) + B(\lambda) \}
	& \geq A(\lambda_m) \asymp n R^2 \mu^2 \geq n \sigma^2.
\end{align*}
Meanwhile,
\begin{align*}
	B(\lambda_m)
	& =  \sigma^2 
	\Tr [ ( \bSigma_0 + \conversionrate^{-1} \lambda_m \bI )^{-1} \bSigma_0 ] \leq 
	\sigma^2 \cdot
	\frac{ \Tr (\bSigma_0) }{  \conversionrate^{-1} \lambda_m } 	
	\lesssim  \sigma^2 M^2 / \mu^2 \leq \sigma^2. 
\end{align*}
Therefore, $B(\lambda_m) \lesssim A(\lambda_m) $ and
\begin{align}
	& \inf_{\lambda \geq \lambda_m} \{ A(\lambda) + B(\lambda) \}
	\gtrsim A(\lambda_m) + B(\lambda_m).
	\label{eqn-lem-excess-risk-2}
\end{align}
Next, we bridge the grid $\Lambda$ and the interval $[\lambda_1, \lambda_m]$. For any $\lambda \in [ \lambda_j, \lambda_{j+1} ]$, we have $A( \lambda) \geq A(\lambda_j)$ and $ \bSigma_0 +  \conversionrate^{-1}  \lambda \bI \preceq (\lambda / \lambda_j) ( \bSigma_0 +  \conversionrate^{-1} \lambda_j  \bI)$. Hence,
\[
A(\lambda) + B(\lambda) \geq A(\lambda_j) + \frac{\lambda_j}{\lambda} B(\lambda_j) \geq \frac{ A(\lambda_j) + B(\lambda_j) }{2}.
\]
We get 
\begin{align}
	\inf_{\lambda_1 \leq \lambda \leq \lambda_m} \{ A(\lambda) + B(\lambda) \}
	\geq \frac{1}{2} \min_{\lambda \in \Lambda} \{  A(\lambda) + B(\lambda) \}.
	\label{eqn-lem-excess-risk-3}
\end{align}
The claim \eqref{eqn-grid-interval} follows from \eqref{eqn-lem-excess-risk-1}, \eqref{eqn-lem-excess-risk-2} and \eqref{eqn-lem-excess-risk-3}.

\section{Discussion}\label{sec-discussion}

We developed a simple approach for kernel ridge regression under covariate shift. 
A key component is an imputation model that generates pseudo-labels for model selection. The final estimator is shown to adapt to the covariate shift in many common scenarios.
We hope that our work can spur further research toward a systematic solution to covariate shift problems.
Lots of open questions are worth studying. 
For instance, it would be interesting to connect our notion of effective sample size to various discrepancy measures between probability distributions, such as divergences or Wasserstein metrics. 
Our two-stage approach with separate uses of source and target data may need modification in certain scenarios, as the candidate models may better be trained using some weighted loss based on the target data.
Another direction is statistical learning over general function classes. The current paper only considered kernel ridge regression with well-specified model. Analytical expressions of the fitted models greatly facilitated our theoretical analysis. In future work, we would like to go beyond that and develop more versatile methods.

\section*{Acknowledgement}
Kaizheng Wang's research is supported by an NSF grant DMS-2210907, a Data Science Institute seed grant SF-181 and a startup grant at Columbia University. We acknowledge computing resources from Columbia University's Shared Research Computing Facility project, which is supported by NIH Research Facility Improvement Grant 1G20RR030893-01, and associated funds from the New York State Empire State Development, Division of Science Technology and Innovation (NYSTAR) Contract C090171, both awarded April 15, 2010. 

\newpage 
\appendix

\section{Proofs of Section 6}

\subsection{Proof of \Cref{lem-error-bounds}}\label{sec-lem-error-bounds-proof}

Define $\bvarepsilon_1 = \by_1 - \bX_1 \btheta^{\star}$. Under Assumption \ref{assumption-linear-model-rkhs},
\begin{align}
	& \bar\btheta_{\lambda} = \EE [ (\widehat{\bSigma}_1 + \lambda \bI)^{-1} (n_1^{-1} \bX_1^{\top} \by_1 ) ] = (\widehat{\bSigma}_1 + \lambda \bI)^{-1}  \widehat{\bSigma}_1 \btheta^{\star} , \notag  \\
	& \bar\btheta_{\lambda} - \btheta^{\star} = - \lambda (\widehat{\bSigma}_1 + \lambda \bI)^{-1}   \btheta^{\star} , \label{eqn-lem-error-bounds-1} \\
	& \widehat{\btheta}_{\lambda} - \bar\btheta_{\lambda} = (\widehat{\bSigma}_1 + \lambda \bI)^{-1} (n_1^{-1} \bX_1^{\top} \bvarepsilon_1 ). \label{eqn-lem-error-bounds-2}
\end{align}
On the one hand,
\begin{align}
	\| \bar\btheta_{\lambda} - \btheta^{\star} \|_{\HH} = \lambda \| (\widehat{\bSigma}_1 + \lambda \bI)^{-1}   \btheta^{\star}  \|_{\HH} \leq \| \btheta^{\star} \|_{\HH} \leq R.
	\label{eqn-lem-error-bounds-2.1}
\end{align}
On the other hand, note that $( \widehat{\bSigma}_1 + \lambda \bI)^{-1} (n_1^{-2} \bX_1^{\top} ) :~ \RR^n \to \HH$ is a bounded linear operator whose adjoint is $(n_1^{-2} \bX_1 ) ( \widehat{\bSigma}_1 + \lambda \bI)^{-1} $. Hence,
\begin{align*}
	& \| ( \widehat{\bSigma}_1 + \lambda \bI)^{-1} (n_1^{-1} \bX_1^{\top} ) \|^2 = \| ( \widehat{\bSigma}_1 + \lambda \bI)^{-1} (n_1^{-1} \bX_1^{\top} ) ( n_1^{-1} \bX_1 ) ( \widehat{\bSigma}_1 + \lambda \bI)^{-1} \| \notag \\
	&= n_1^{-1} \| ( \widehat{\bSigma}_1 + \lambda \bI)^{-1}   \widehat{\bSigma}_1 ( \widehat{\bSigma}_1 + \lambda \bI)^{-1} \| \leq (n_1 \lambda)^{-1} . 
\end{align*}
By Assumption \ref{assumption-noise},
\begin{align*}
	\| \widehat\btheta_{\lambda} - \bar\btheta_{\lambda} \|_{\psi_2} 
	= \| (\widehat{\bSigma}_1 + \lambda \bI)^{-1} (n_1^{-1} \bX_1^{\top} ) \bvarepsilon_1  \|_{\psi_2}
	\leq \frac{ \| \bvarepsilon_1  \|_{\psi_2} }{ \sqrt{n_1 \lambda} } 
	\leq  \frac{ \sigma}{ \sqrt{n_1 \lambda} } .
\end{align*}

By \eqref{eqn-lem-error-bounds-1},
\begin{align*}
	\| \bX_0 ( \bar\btheta_{\lambda} - \btheta^{\star} ) \|_2^2 
	& = \lambda^2 \| \bX_0 (\widehat{\bSigma}_1 + \lambda \bI)^{-1}   \btheta^{\star} \|_2^2 
	=\lambda^2 R^2 \| \bX_0 (\widehat{\bSigma}_1 + \lambda \bI)^{-2} \bX_0^{\top} \|
.
\end{align*}	
To study $\|
	n_0^{-1/2}		\bX_0 (  \widehat\btheta_{\lambda} -  \bar\btheta_{\lambda}  ) 
	\|_{\psi_2} $, define $\bA = \bX_0 (\widehat{\bSigma}_1 + \lambda \bI)^{-1} n_1^{-1} \bX_1^{\top} \in \RR^{n\times n}$. By \eqref{eqn-lem-error-bounds-2}, we have $\bX_0 ( \widehat{\btheta}_{\lambda} - \bar\btheta_{\lambda} ) = \bA \bvarepsilon_1 $. By Assumption \ref{assumption-noise},
\begin{align*}
	\| 
	\bX_0 ( \widehat{\btheta}_{\lambda} - \bar\btheta_{\lambda} ) 
	\|_{\psi_2}
	\leq \| \bA \|_2 \| \bvarepsilon_1 \|_{\psi_2}
	\leq   \sigma  \| \bA \|_2 .
\end{align*}
Note that $\bA \bA^{\top} = n_1^{-1} \bX_0 (\widehat{\bSigma}_1 + \lambda \bI)^{-1} \widehat{\bSigma}_1 (\widehat{\bSigma}_1 + \lambda \bI)^{-1}  \bX_0^{\top} $. Consequently,
\begin{align*}
\| 
\bX_0 ( \widehat{\btheta}_{\lambda} - \bar\btheta_{\lambda} ) 
\|_{\psi_2}^2 
& \leq  \frac{ \sigma^2}{n_1} \|  \bX_0 (\widehat{\bSigma}_1 + \lambda \bI)^{-1} \widehat{\bSigma}_1 (\widehat{\bSigma}_1 + \lambda \bI)^{-1}  \bX_0^{\top}  \| 
.
\end{align*}

Next, we come to $	\| \widehat\btheta_{\lambda} - \btheta^{\star} \|_{\bQ}^2$. Note that
\begin{align*}
	\| \widehat{\btheta}_{\lambda}  - \btheta^{\star} \|_{\bQ}^2 \leq  [ \|  \widehat\btheta_{\lambda} - \bar\btheta_{\lambda}   \|_{\bQ} + \|  \bar\btheta_{\lambda} - \btheta^{\star} \|_{\bQ} ]^2
	\leq 2 \|  \widehat\btheta_{\lambda} - \bar\btheta_{\lambda}   \|_{\bQ}^2 +2  \|  \bar\btheta_{\lambda} - \btheta^{\star}  \|_{\bQ}^2 .
\end{align*}
By \eqref{eqn-lem-error-bounds-1},
\begin{align*}
	& \| \bar\btheta_{\lambda} - \btheta^{\star} \|_{\bQ}^2  = \lambda^2  \| (\widehat{\bSigma}_1 + \lambda \bI)^{-1}   \btheta^{\star} \|_{\bQ}^2 
	\leq  \lambda^2 R^2 \| (\widehat{\bSigma}_1 + \lambda \bI)^{-1}   \bQ (\widehat{\bSigma}_1 + \lambda \bI)^{-1} \| 
	 .
\end{align*}
By \eqref{eqn-lem-error-bounds-2},
\begin{align*}
	\| \widehat\btheta_{\lambda} - \bar\btheta_{\lambda} \|_{\bQ}^2 & =
	\| (\widehat{\bSigma}_1 + \lambda \bI)^{-1} (n_1^{-1} \bX_1^{\top} \bvarepsilon_1 ) \|_{\bQ}^2 \\
	& = \langle \bvarepsilon_1 , 
	(n_1^{-1} \bX_1 ) (\widehat{\bSigma}_1 + \lambda \bI)^{-1} \bQ (\widehat{\bSigma}_1 + \lambda \bI)^{-1} (n_1^{-1} \bX_1^{\top} )
	\bvarepsilon_1 \rangle .
\end{align*}
Note that 
\begin{align*}
	& \Tr [(n_1^{-1} \bX_1 ) (\widehat{\bSigma}_1 + \lambda \bI)^{-1} \bQ (\widehat{\bSigma}_1 + \lambda \bI)^{-1} (n_1^{-1} \bX_1^{\top} )] \\
	& = \Tr [ (\widehat{\bSigma}_1 + \lambda \bI)^{-1} \bQ (\widehat{\bSigma}_1 + \lambda \bI)^{-1} (n_1^{-1} \bX_1^{\top} ) (n_1^{-1} \bX_1 )] \\
	& = n_1^{-1} \Tr [ (\widehat{\bSigma}_1 + \lambda \bI)^{-1} \bQ (\widehat{\bSigma}_1 + \lambda \bI)^{-1} \widehat{\bSigma}_1] 
	= n_1^{-1} \Tr [  \widehat{\bSigma}_1 (\widehat{\bSigma}_1 + \lambda \bI)^{-2} \bQ ] .
\end{align*}
By Assumption \ref{assumption-noise} and \Cref{lem-subg-norm},
\begin{align*}
	\PP \bigg(
	\|  \widehat\btheta_{\lambda} - \bar\btheta_{\lambda}   \|_{\bQ}^{2} \leq  
	\frac{C \sigma^2 t}{ n_1 } \Tr [  \widehat{\bSigma}_1 (\widehat{\bSigma}_1 + \lambda \bI)^{-2} \bQ ] 
	\bigg) \geq 1 - e^{- t }, \qquad \forall t \geq 1.
\end{align*}
Here $C$ is the universal constant in \Cref{lem-subg-norm}. Therefore, given any $t \geq 1$, with probability at least $1 - e^{-t}$ we have
\begin{align*}
\frac{1}{2}\| \widehat\btheta_{\lambda} - \btheta^{\star} \|_{\bQ}^2
& \leq \|  \widehat\btheta_{\lambda} - \bar\btheta_{\lambda}   \|_{\bQ}^2 +  \|  \bar\btheta_{\lambda} - \btheta^{\star}  \|_{\bQ}^2 \\
& \leq 
\lambda^2 R^2 \| (\widehat{\bSigma}_1 + \lambda \bI)^{-1}   \bQ (\widehat{\bSigma}_1 + \lambda \bI)^{-1} \| 
+ \frac{C \sigma^2 t}{ n_1 } \Tr [  \widehat{\bSigma}_1 (\widehat{\bSigma}_1 + \lambda \bI)^{-2} \bQ ] 
\\ &
\leq 
\lambda R^2 \| (\widehat{\bSigma}_1 + \lambda \bI)^{-1/2}   \bQ (\widehat{\bSigma}_1 + \lambda \bI)^{-1/2} \| 
+ \frac{C \sigma^2 t}{ n_1 } \Tr [  (\widehat{\bSigma}_1 + \lambda \bI)^{-1} \bQ ] \\
& = \lambda R^2 \| \widehat\bS_{\lambda} \|  + \frac{C \sigma^2 \Tr( \widehat\bS_{\lambda} )  }{n_1} t.
\end{align*}
The proof is completed by choosing any $\delta \in (0, 1 / e]$, letting $t = \log(1/\delta)$, and redefining $C$.

	\subsection{Proof of \Cref{lem-cov-intermediate}}\label{sec-lem-cov-intermediate-proof}

	Choose any $\delta \in (0, 1/e]$. To prove Part 1, we let Assumption \ref{assumption-covariates-bounded} hold. 
	On the one hand, we apply the first part of \Cref{cor-cov-rkhs} to $\{ \phi (  \bx_{0i} ) \}_{i=1}^{n_0}$, with $\gamma$ set to be $1/2$. This leads to
	\begin{align*}
		\PP \bigg(
		\frac{1}{2} ( \bSigma_0 + \mu_0 \bI)  \preceq
		\widehat\bSigma_0 + \mu_0 \bI \preceq \frac{3}{2} ( \bSigma_0 + \mu_0 \bI)  
		\bigg)
		\geq 1 - \delta .
	\end{align*}
	Therefore, 
	\begin{align*}
		\PP\bigg(
		\widehat\bSigma_0   \preceq \frac{3}{2} 
		\bSigma_0 + \frac{\mu_0}{2}  \bI
		\bigg) \geq 1 - \delta  .
	\end{align*}
	
	On the other hand, we apply the first part of \Cref{cor-cov-rkhs} to $\{ \phi (  \bx_i ) \}_{i \in \cT}$, with $\gamma$ set to be $1/2$. This yields
	\begin{align*}
		\PP \bigg(
		\frac{1}{2} ( \bSigma + \mu_1 \bI ) \preceq  \widehat\bSigma_1 + \mu_1 \bI \preceq 
		\frac{3}{2} ( \bSigma + \mu_1 \bI ) 
		\bigg)
		\geq 1 - \delta .
	\end{align*}
	On the above event, for any $\lambda \geq \mu_1$ we have
	\begin{align*}
		\widehat\bSigma_1 + \lambda \bI =
		(\widehat\bSigma_1 + \mu_1 \bI) + (\lambda - \mu_1) \bI 
		\succeq \frac{1}{2} ( \bSigma + \mu_1 \bI ) + (\lambda - \mu_1) \bI 
		\succeq \frac{1}{2} ( \bSigma + \lambda \bI );
	\end{align*}
	for any $0 \leq \lambda \leq \mu_1$ we have
	\begin{align*}
		\widehat\bSigma_1 + \lambda \bI \succeq
		(\lambda / \mu_1) (\widehat\bSigma_1 + \mu_1 \bI) 
		\succeq (\lambda / \mu_1) \cdot \frac{1}{2} ( \bSigma + \mu_1 \bI ) 
		\succeq \frac{\lambda / \mu_1}{2} ( \bSigma + \lambda \bI ).
	\end{align*}
	Therefore, 
	\begin{align}
		\PP \bigg(
		\widehat\bSigma_1 + \lambda \bI \succeq 
		\frac{\min \{ \lambda / \mu_1, 1 \} }{2} ( \bSigma + \lambda \bI ) , ~~ \forall \lambda \geq 0
		\bigg) \geq 1 - \delta 
		.
		\label{eqn-lem-cov-intermediate-1-2}
	\end{align}
	Similarly, we apply the first part of \Cref{cor-cov-rkhs} to $\{  \phi ( \bx_i ) \}_{i \in [n] \backslash \cT}$ and get
	\begin{align}
		\PP \bigg(
		\widehat\bSigma_2 + \lambda \bI \succeq 
		\frac{\min \{ \lambda / \mu_2, 1 \} }{2} ( \bSigma + \lambda \bI ) , ~~ \forall \lambda \geq 0
		\bigg) \geq 1 - \delta .
		\label{eqn-lem-cov-intermediate-1-3}
	\end{align}
	This proves Part 1. The proof of Part 2 is very similar to the above (using the second part of \Cref{cor-cov-rkhs} instead) and is thus omitted.

\subsection{Proof of \Cref{lem-in-sample}}\label{sec-lem-in-sample-proof}

Define
\begin{align*}
& \by^{\star} =  \bX_0 \btheta^{\star}  =  ( f^{\star} ( \bx_{01} ) ,~\cdots, f^{\star} ( \bx_{0n_0} )  )^{\top}, \\
& \widetilde\by =  \bX_0 \widetilde{\btheta}  = ( \widetilde{f} ( \bx_{01} ) ,~\cdots, \widetilde{f} ( \bx_{0n_0} )  )^{\top},
\end{align*}
and $\bar\btheta = \EE ( \widetilde\btheta | \bX_2 )$. It is easily seen from the decomposition \eqref{eqn-in-decomposition} that
\[
\empiricalrisk ( \widehat{f}_{\lambda} ) - \empiricalrisk ( f^{\star} ) = \| \widehat\btheta_{\lambda} - \btheta^{\star} \|_{\widehat\bSigma_0}^2 .
\] 
Choose $\delta \in (0, 1]$. \Cref{thm-bias-variance} implies that with probability at least $1 - \delta$,
\begin{align*}
& \sqrt{ \empiricalrisk  (\widehat{f}) - \empiricalrisk  (f^{\star}) } -
\min_{ \lambda \in \Lambda  } \sqrt{ \empiricalrisk  ( \widehat{f}_{\lambda} ) - \empiricalrisk  ( f^{\star} ) }
\\& 
\lesssim
n_0^{-1/2} \| \EE \widetilde\by - \by^{\star} \|_2
+
\| \widetilde\by - \EE \widetilde\by \|_{\psi_2} \sqrt{ \frac{ \log ( |\Lambda|  / \delta ) }{n_0} } ,
\end{align*}
where $\lesssim$ only hides a universal constant. Applying the first half of \Cref{lem-error-bounds} to $\widetilde\btheta$ shows that 
\begin{align*}
& 	
\| \EE \widetilde\by - \by^{\star} \|_2^2
\leq \widetilde\lambda^2 R^2 \| \bX_0 (\widehat{\bSigma}_2 + \widetilde\lambda \bI)^{-2} \bX_0^{\top} \|, \\
&	\| \widetilde\by - \EE \widetilde\by \|_{\psi_2}^2 
\leq
\frac{ \sigma^2}{ n_2} \|  \bX_0 (\widehat{\bSigma}_2 + \widetilde\lambda \bI)^{-1} \widehat{\bSigma}_2 (\widehat{\bSigma}_2 + \widetilde\lambda \bI)^{-1}  \bX_0^{\top}  \| .
\end{align*}
Therefore,
\begin{align}
& 
n_0 \bigg|
\sqrt{ \empiricalrisk  (\widehat{f}) - \empiricalrisk  (f^{\star}) } -
\min_{ \lambda \in \Lambda  } \sqrt{ \empiricalrisk  ( \widehat{f}_{\lambda} ) - \empiricalrisk  ( f^{\star} ) }
\bigg|^2
\notag \\& 
\lesssim
\widetilde\lambda^2 R^2 \| \bX_0 (\widehat{\bSigma}_2 + \widetilde\lambda \bI)^{-2} \bX_0^{\top} \|
+
\frac{ \sigma^2}{ n_2} \|  \bX_0 (\widehat{\bSigma}_2 + \widetilde\lambda \bI)^{-1} \widehat{\bSigma}_2 (\widehat{\bSigma}_2 + \widetilde\lambda \bI)^{-1}  \bX_0^{\top}  \| \cdot \log(2|\Lambda|/\delta) \notag \\
& \leq 2  \log(2|\Lambda|/\delta)
\bigg\| \bX_0 
\bigg(
\widetilde\lambda^2 R^2 (\widehat{\bSigma}_2 + \widetilde\lambda \bI)^{-2} 
+ \frac{ \sigma^2}{ n_2}
(\widehat{\bSigma}_2 + \widetilde\lambda \bI)^{-1} \widehat{\bSigma}_2 (\widehat{\bSigma}_2 + \widetilde\lambda \bI)^{-1}  
\bigg)
\bX_0^{\top} \bigg\| \notag \\
& = \frac{2 \sigma^2 \log(2|\Lambda|/\delta) }{ n_2 }
\| \bX_0 
[
(\widehat{\bSigma}_2 + \widetilde\lambda \bI)^{-1}
(
\widehat{\bSigma}_2 + n_2  R^2 \sigma^{-2} \widetilde\lambda^2 \bI
)
(\widehat{\bSigma}_2 + \widetilde\lambda \bI)^{-1}
]
\bX_0^{\top} \|
.
\label{eqn-lem-in-sample-1}
\end{align}
It is easily seen that
\[
\widehat{\bSigma}_2 + n_2  R^2 \sigma^{-2} \widetilde\lambda^2 \bI
\preceq \max \{ 1 ,  n_2  R^2 \sigma^{-2} \widetilde\lambda \} (\widehat{\bSigma}_2 + \widetilde\lambda \bI)
\]
and thus,
\begin{align}
& \| \bX_0 
[
(\widehat{\bSigma}_2 + \widetilde\lambda \bI)^{-1}
(
\widehat{\bSigma}_2 + n_2  R^2 \sigma^{-2} \widetilde\lambda^2 \bI
)
(\widehat{\bSigma}_2 + \widetilde\lambda \bI)^{-1}
]
\bX_0^{\top} \|
\notag \\&
\leq \max \{ 1 ,  n_2  R^2 \sigma^{-2} \widetilde\lambda \}
\| \bX_0 
(\widehat{\bSigma}_2 + \widetilde\lambda \bI)^{-1}
\bX_0^{\top} \| 
\notag \\&
= n_0 \max \{ 1 ,  n_2  R^2 \sigma^{-2} \widetilde\lambda \}
\| (\widehat{\bSigma}_2 + \widetilde\lambda \bI)^{-1/2}
\widehat\bSigma_0
(\widehat{\bSigma}_2 + \widetilde\lambda \bI)^{-1/2}
 \|.
\label{eqn-lem-in-sample-2}
\end{align}
Combining \eqref{eqn-lem-in-sample-1} and \eqref{eqn-lem-in-sample-2} completes the proof.

\subsection{Proof of \Cref{lem-sandwich}}\label{sec-lem-sandwich-proof}

\subsubsection{Part 1}

We will invoke Part 1 of \Cref{lem-L2} to relate the in-sample risk to the out-of-sample one. Some preparations are needed.

\begin{claim}\label{claim-subg-norm}
There exists a universal constant $C$ such that for any $\lambda \in \Lambda$,
	\begin{align*}
& \PP \bigg(
 | \langle \phi ( \bx_{01} ) , \widehat{\btheta}_{\lambda} - \btheta^{\star} \rangle |
>
C M \bar{R}
 \sqrt{  \log ( 1 / \varepsilon ) }
\bigg) \leq  \varepsilon , \qquad \forall \varepsilon \in ( 0 , 1 / 2 ]
		  , \\
& \EE^{1/4}  | \langle \phi ( \bx_{01} ) , \widehat{\btheta}_{\lambda} - \btheta^{\star} \rangle |^4
\leq C M \bar{R}.
	\end{align*}
\end{claim}

\begin{proof}[\bf Proof of Claim \ref{claim-subg-norm}]
On the one hand, \Cref{lem-error-bounds} asserts that
\begin{align*}
& \| \bar{\btheta}_{\lambda} - \btheta^{\star} \|_{\HH} \leq \| \btheta^{\star} \|_{\HH} \leq R  
\qquad\text{and}\qquad
\| \widehat{\btheta}_{\lambda} - \bar{\btheta}_{\lambda} \|_{\psi_2} \lesssim  \sigma /  \sqrt{n_1 \lambda}  .
\end{align*}	
Then, $\|  \widehat{\btheta}_{\lambda} -  \btheta^{\star}  \|_{\psi_2} \lesssim  \bar{R}$. Assumption \ref{assumption-covariates-bounded} leads to $\| \langle \phi ( \bx_{01} ) , \widehat{\btheta}_{\lambda} -  \btheta^{\star} \rangle \|_{\psi_2} \lesssim  M \bar{R}$. Hence, there exists a universal constant $c$ such that
	\begin{align*}
		\PP \bigg(
		| \langle \phi ( \bx_{01} ) , \widehat{\btheta}_{\lambda} - \btheta^{\star} \rangle |
		> c M \bar{R} \sqrt{  \log ( 1 / \varepsilon ) }
		\bigg) \leq  \varepsilon , \qquad \forall \varepsilon \in (0, 1/2].
	\end{align*}
Meanwhile,
\begin{align*}
 \EE^{1/4}  | \langle \phi ( \bx_{01} ) , \widehat{\btheta}_{\lambda} - \btheta^{\star} \rangle |^4
& \leq 
\sqrt{4} \cdot \| \langle \phi ( \bx_{01} ) , \widehat{\btheta}_{\lambda} -  \btheta^{\star} \rangle \|_{\psi_2} 
 \lesssim M \bar{R}.
\end{align*}
The proof is finished by choosing a sufficiently large constant $C$.
\end{proof}

Choose any $\lambda \in \Lambda $ and $\varepsilon \in (0, 1/2]$. Let $\cZ = \cX$, $f(\cdot) = \langle \phi ( \cdot ) , \widehat{\btheta}_{\lambda} - \btheta^{\star} \rangle  $, $U = C M \bar{R}$ and $r = C M \bar{R} \sqrt{ \log ( 1 / \varepsilon ) } $. Claim \ref{claim-subg-norm} asserts that for $z \sim \target$,
	\begin{align*}
	& \PP ( | f(z) | > r ) \leq \varepsilon
	\qquad\text{and}\qquad
	\EE | f(z) |^4 \leq U^4 
\end{align*}
Based on the above, Part 1 of \Cref{lem-L2} implies that for any $\delta_1 \in (0, 1)$,
	\begin{align*}
&	\PP \bigg[
	\Big|
\| \widehat{\btheta}_{\lambda} - \btheta^{\star} \|_{\widehat\bSigma_0}
 - \| \widehat{\btheta}_{\lambda} - \btheta^{\star} \|_{\bSigma_0}
	\Big|
	\leq 
C M \bar{R}
\bigg(
 \sqrt{
		\frac{
			7 \log ( 1 / \varepsilon ) \log (2 / \delta_1)
		}{n_0}
	}
	+ \varepsilon^{1/4}
	\bigg)
\bigg] \\
&	\geq 1 - n_0 \varepsilon - \delta_1.
\end{align*}
By union bounds, the event
\[
\max_{\lambda \in \Lambda}
	\Big|
\| \widehat{\btheta}_{\lambda} - \btheta^{\star} \|_{\widehat\bSigma_0}
- \| \widehat{\btheta}_{\lambda} - \btheta^{\star} \|_{\bSigma_0}
\Big|
\leq 
C M \bar{R}
\bigg(
\sqrt{
	\frac{
		7 \log ( 1 / \varepsilon ) \log (2 / \delta_1)
	}{n_0}
}
+ \varepsilon^{1/4}
\bigg)
\]
holds with probability at least $1 - ( n_0 \varepsilon + \delta_1 ) | \Lambda |$.

Choose any $\delta \in (0, 1/e]$. Let $\varepsilon = [ \delta / ( 2 n_0 |\Lambda| ) ]^2$ and $\delta_1 = \delta / (2 |\Lambda|)$. We have
\[
( n_0 \varepsilon + \delta_1 ) | \Lambda |
= \bigg[
n_0 \cdot 
\bigg(
\frac{\zeta}{2 n_0 |\Lambda|} 
\bigg)^2
+ \frac{\zeta}{2 |\Lambda|} 
\bigg] |\Lambda| \leq \delta.
\]
Hence, with probability at least $1 - \delta$,
\begin{align*}
& \max_{\lambda \in \Lambda}
\Big|
\sqrt{ \risk ( \widehat{f}_{\lambda} ) - \risk ( f^{\star} ) } - 
\sqrt{	\empiricalrisk ( \widehat{f}_{\lambda} ) - \empiricalrisk ( f^{\star} ) } 
\Big|
 =
\max_{\lambda \in \Lambda}
\Big|
\| \widehat{\btheta}_{\lambda} - \btheta^{\star} \|_{\widehat\bSigma_0}
- \| \widehat{\btheta}_{\lambda} - \btheta^{\star} \|_{\bSigma_0}
\Big| \\
& \lesssim
M \bar{R}
\bigg(
\sqrt{
	\frac{
		\log ( n_0 |\Lambda| / \delta ) \log (|\Lambda| / \delta)
	}{n_0}
}
+ \sqrt{
	\frac{\zeta}{2 n_0 |\Lambda|}
}
\bigg) 
 \lesssim M \bar{R}
\sqrt{
	\frac{
		\log ( n_0 |\Lambda| / \delta ) \log (|\Lambda| / \delta)
	}{n_0}
}.
\end{align*}
The proof is finished by redefining $C$.

\subsubsection{Part 2}

We will invoke Part 2 of \Cref{lem-L2}. Choose any $\lambda \in \Lambda $ and $\delta \in (0, 1/e]$. Let $\cZ = \cX$, $f(\cdot) = \langle \phi ( \cdot ) , \widehat{\btheta}_{\lambda} - \btheta^{\star} \rangle  $, and $z_i = \bx_{0i}$. We have
\[
\| f \|_{L^2}^2 = \EE \Big[ | \langle  \phi ( \bx_{0i} ) , \widehat{\btheta}_{\lambda} - \btheta^{\star} \rangle |^2 \Big| \widehat{\btheta}_{\lambda} \Big].
\]
Thanks to Assumption \ref{assumption-covariates-subg}, given $f$, the $\psi_2$ norm of $f(z_1) / \| f \|_{L^2}$ is bounded by $\sqrt{\kappa}$. Then, Part 2 of \Cref{lem-L2} implies the existence of a constant $C \geq 1$ such that when $\delta \in (0, 1)$ and $n \geq C^2 \kappa^2 \log(2/\delta) $, we have
\begin{align*}
\PP \bigg(
\Big| \| f \|_n^2 -  \| f \|_{L^2}^2 \Big| \leq 
\| f \|_{L^2}^2 \cdot C \kappa \sqrt{ \frac{ \log (2/\delta) }{ n_0 } }
\bigg)  \geq 1 -  \delta .
\end{align*}
By union bounds, if $\delta \in (0, 1]$ and $n_0 \geq  (C \kappa)^2 \log ( 2 |\Lambda| / \delta )$, then
\begin{align*}
\PP \bigg(
\Big| \| \widehat{\btheta}_{\lambda} - \btheta^{\star} \|_{\widehat{\bSigma}_0}^2 -  \| \widehat{\btheta}_{\lambda} - \btheta^{\star} \|_{\bSigma_0}^2 \Big| \leq 
\| \widehat{\btheta}_{\lambda} - \btheta^{\star} \|_{\bSigma_0}^2 
\cdot C \kappa \sqrt{ \frac{ \log (2 |\Lambda| /\delta) }{ n_0 } }
,~\forall \lambda \in \Lambda
\bigg)  \geq 1 -  \delta .
\end{align*}
Choose any $\gamma \in (0, 1]$. If $n_0 \geq  (C \kappa / \gamma)^2 \log ( 2 |\Lambda| / \delta )$, then the above concentration inequality forces
\begin{align*}
\PP \bigg(
\Big| \| \widehat{\btheta}_{\lambda} - \btheta^{\star} \|_{\widehat{\bSigma}_0}^2 -  \| \widehat{\btheta}_{\lambda} - \btheta^{\star} \|_{\bSigma_0}^2 \Big| \leq 
\gamma \| \widehat{\btheta}_{\lambda} - \btheta^{\star} \|_{\bSigma_0}^2 
,~\forall \lambda \in \Lambda
\bigg)  \geq 1 -  \delta .
\end{align*}

\section{Proofs of Section 4}

\subsection{Proof of \Cref{cor-r}}\label{sec-cor-r-proof}

Throughout the proof, we use $\lesssim$ to hide logarithmic factors. \Cref{thm-excess-risk-0} implies that when $n$ and $n_0 / \log n$ are sufficiently large, the following inequality holds with probability at least $1 - n^{-1}$:
\begin{align*}
	\risk ( \widehat{f}  ) - \risk ( f^{\star} ) 
	&	\lesssim
	\inf_{  \rho > 0 }
	\bigg\{
	\rho + \frac{ 1 }{ \effectivesamplesize } \sum_{j=1}^{\infty} \min \{ 1 , \mu_j / \rho \}
	\bigg\}
	+ 
	\frac{1}{\effectivesamplesize} +
	\frac{
		1
	}{n_0}
	.
\end{align*}
Denote by $A(\rho)$ the function inside the infimum. 
When $\rank (\bSigma_0) \leq D$, we have $$\sum_{j=1}^{\infty} \min \{ 1 , \mu_j / \rho \} \leq D.$$ Then, the fact
\[
\inf_{  \rho > 0 } A(\rho)
\leq 
\inf_{  \rho > 0 } 
\bigg\{
\rho + \frac{ D }{ \effectivesamplesize }  
\bigg\} = \frac{D}{ \effectivesamplesize}
\]
yields Part \ref{part-finite-rank-cor-r}.

Suppose that $\mu_j \leq c_1 e^{-c_2 j  }$, $\forall j$ holds with some constants $c_1, c_2 > 0$. We have
\begin{align*}
	A(\rho) 
	& \lesssim \rho + \frac{ 1 }{ \effectivesamplesize }  \bigg(
	k + \sum_{j=k + 1}^{\infty} e^{-c_2 j} / \rho
	\bigg)  \lesssim \rho + \frac{ k + e^{-c_2 k}  / \rho }{ \effectivesamplesize } ,\qquad\forall k \in \ZZ_+.
\end{align*}
Taking $\rho = \effectivesamplesize^{-1}$ and $k = \lceil c_2^{-1} \log(1/  \rho ) \rceil$ yields $A(\rho ) \lesssim \effectivesamplesize^{-1}$ and proves Part \ref{part-exponential-cor-r}.

Finally, suppose that $\mu_j \leq c j^{- 2 \alpha}$, $\forall j$ holds with some constants $c > 0$ and $\alpha > 1/2$. We have
\begin{align*}
	A(\rho) 
	& \lesssim \rho + \frac{ 1 }{ \effectivesamplesize }  \bigg(
	k + \sum_{j=k + 1}^{\infty} j^{-2 \alpha} / \rho
	\bigg)  \lesssim \rho + \frac{ k + k^{1-2\alpha} / \rho }{ \effectivesamplesize } ,\qquad\forall k \in \ZZ_+.
\end{align*}
Taking $\rho = \effectivesamplesize^{- 2 \alpha / (2 \alpha + 1) }$ and $k = \lceil \rho^{-1/(2\alpha)} \rceil$ yields $A( \rho ) \lesssim \effectivesamplesize^{- 2 \alpha / (2 \alpha + 1)}$ and proves Part \ref{part-polynomial-cor-r}.

\subsection{Proof of \Cref{example-Sobolev-conversion}}\label{sec-example-Sobolev-conversion-proof}

For $j \in \ZZ_+$ and $x \in [0, 1]$, define
\[
v_j (x) = \sin \bigg( \frac{ (2j-1) \pi x }{2} \bigg)
\qquad\text{and}\qquad
\sigma_j = \bigg( \frac{2}{ (2j-1) \pi } \bigg)^2.
\]
Consider a Hilbert space $\HH = \ell^2 = \{ ( z_1, z_2, \cdots ) :~ z_j \in \RR,~ \sum_{j=1}^{\infty} z_j^2 < \infty \}$ with $\langle \bz, \bw \rangle = \sum_{j=1}^{\infty} z_j w_j$, and define
\[
\phi:~ \cX \to \HH ,~x \mapsto ( \sqrt{\sigma_1} v_1 (x), \sqrt{\sigma_2} v_2 (x) , \cdots  ).
\]
We have $K(x, z) = \langle \phi(x), \phi(z) \rangle$, $\forall x, z \in \cX$; see Example 12.23 in \cite{Wai19}. In addition, we have $ \EE_{z \sim \source} [ K(\cdot, z) v_j(z) ] = \int_{0}^{1} K(\cdot, z) v_j (z) \rd z = \sigma_j v_j(\cdot)$ for all $j \in \ZZ_+$. Hence,
\begin{align*}
	& \bSigma = \EE_{x \sim \source} [ \phi(x) \otimes \phi(x) ] = \diag ( \sigma_1, \sigma_2, \cdots  ) , \\
	& \bSigma_0 
	= \phi(x_0) \otimes \phi(x_0)
	= ( \sqrt{\sigma_1} v_1 (x_0), \sqrt{\sigma_2} v_2 (x_0) , \cdots  )^{\otimes 2},
\end{align*}

First, we will calculate $\effectivesamplesize$. 
Note that $(\bSigma + n^{-1} \bI)^{-1/2}$ is a bounded linear transform on $\HH$. Define $\xi = (\bSigma + n^{-1} \bI)^{-1/2} \phi (x_0) $. Since $x_0 = 1/2$, we have $|v_j (x_0)|^2 = 1/2$ and
\begin{align*}
	\| \xi \|_{\HH}^2 = \sum_{j=1}^{\infty } \frac{\sigma_j}{ \sigma_j + n^{-1} }  |v_j(x_0)|^2
	\asymp \sum_{j=1}^{\infty } \frac{j^{-2}}{ j^{-2} + n^{-1} }
	\asymp \sqrt{n}.
\end{align*}
Therefore, $\|  (\bSigma + n^{-1} \bI)^{-1/2} \bSigma_0 (\bSigma + n^{-1} \bI)^{-1/2} \| \asymp \sqrt{n}$, which implies $\effectivesamplesize \asymp \sqrt{n}$.

To prove the minimax rate-optimality of our estimate $\widehat{f}(x_0)$, we resort to Le Cam's method \citep{Le73}. The key is to construct two similar hypotheses. Recall that $\{ x_i \}_{i=1}^n$ are i.i.d.~from $\source = \cU [0, 1]$. Suppose that conditioned on them, $\{  y_i \}_{i=1}^n$ are independent and $y_i \sim N( f^{\star} (x_i) , 1 )$. Define $g_0, g_1 \in \cF$ with $g_0 \equiv 0$ and
\[
g_1 (x) = \begin{cases}
	0 &, \mbox{ if } x \in  [ 0, x_0 - \frac{1}{2 \sqrt{n}}  ] \cup  ( x_0 + \frac{1}{2 \sqrt{n}} , 1  ] \\
	n^{-1/4} /2 + n^{1/4} ( x - x_0 ) &, \mbox{ if } x \in  ( x_0 - \frac{1}{2 \sqrt{n}} , x_0  ]  \\
	n^{-1/4} /2 -  n^{1/4} ( x - x_0 ) &, \mbox{ if } x \in  ( x_0 ,  x_0 + \frac{1}{2 \sqrt{n}}  ] 
\end{cases},
\]
We have $\| g_0 \|_{\cF}^2 = 0$ and $\| g_1 \|_{\cF}^2 = \int_{0}^{1} | g_1'(x) |^2 \rd x = 1$. Applying Le Cam's method to the hypothesis $H_0$: $f^{\star} = g_0$ and $H_1$: $f^{\star} = g_1$ yields a minimax lower bound of order $n^{-1/4}$. Here we provide some intuitions. The two functions $g_0$ and $g_1$ are different within an interval of length $n^{-1/2}$, which only includes around $\sqrt{n}$ samples. No method can consistently distinguish the two hypotheses that are merely $(\sqrt{n})^{-1/2}  = n^{-1/4}$ apart.

\subsection{Proof of \Cref{lem-LB}}\label{sec-lem-LB-proof}

Define $\bX_1$, $\widehat{\btheta}_{\lambda}$ and $\widehat{\bSigma}_{1}$ as in \Cref{lem-error-bounds}; $v_j$, $\sigma_j$, and $\phi$ as in \Cref{sec-example-Sobolev-conversion-proof}. For simplicity, below we suppress the subscript in $\widehat\bSigma_1$. By \eqref{eqn-lem-error-bounds-1} and \eqref{eqn-lem-error-bounds-2},
\begin{align*}
	& \EE [ \risk (\widehat{f}_{\lambda}) - \risk (f^{\star}) | ~ \bX_1 ]
	= \EE \Big( \| \widehat{\btheta}_{\lambda} - \btheta^{\star} \|_{\bSigma_0}^2
	\Big|~ \bX_1
	\Big) \notag \\
	& = \lambda^2 \| (\widehat{\bSigma} + \lambda \bI)^{-1}   \btheta^{\star} \|_{\bSigma_0}^2
	+ \frac{\sigma^2}{n_1} \Tr [ \bSigma_0 (\widehat{\bSigma} + \lambda \bI)^{-2} \widehat{\bSigma} ].
\end{align*}
Since $\sigma = 1$ and $n_1 = n/2$, we have
\begin{align*}
	& \EE [ \risk_1 (\widehat{f}_{\lambda}) - \risk_1 (f^{\star}) ]
	\geq  \frac{2}{n} \EE 
	\{
	\Tr [ \bSigma (\widehat{\bSigma} + \lambda \bI)^{-2} \widehat{\bSigma} ]
	\}
	, \\	
	& \EE [ \risk_2 (\widehat{f}_{\lambda}) - \risk_2 (f^{\star}) ]
	\geq  
	\lambda^2 \EE | \langle \phi (x_0) , (\widehat{\bSigma} + \lambda \bI)^{-1}   \btheta^{\star} \rangle |^2 .
\end{align*}

Let $\nu = n^{-4/5}$. We will show that when $n$ is sufficiently large,
\begin{align}
	& \inf_{ 0 < \lambda \leq \nu } \EE 
	\{
	\Tr [ \bSigma (\widehat{\bSigma} + \lambda \bI)^{-2} \widehat{\bSigma} ]
	\}
	\gtrsim n^{2/5}, 
	\label{eqn-LB-01}
	\\
	& \inf_{ \lambda > \nu } 	
	\{
	\lambda^2 \EE | \langle \phi (x_0) , (\widehat{\bSigma} + \lambda \bI)^{-1}   \btheta^{\star} \rangle |^2 
	\}
	\gtrsim n^{-2/5}.
	\label{eqn-LB-02}
\end{align}
From there, we can immediately get the desired lower bound.

\subsubsection{Proof of \eqref{eqn-LB-01}}

We present a useful result on the concentration of $\widehat{\bSigma}$.

\begin{claim}\label{claim-concentration}
	Let $\bDelta = ( \bSigma + \nu \bI)^{-1/2} ( \widehat\bSigma - \bSigma ) ( \bSigma + \nu \bI)^{-1/2}$. There exists a universal constant $C$ such that
	\begin{align*}
		\PP \bigg(
		\| \bDelta \|
		\leq \frac{ C \log n }{ 2 n^{1/2} \nu^{1/4}  }
		\bigg) \geq 1 / 2.
	\end{align*}
\end{claim}

\begin{proof}
	We invoke \Cref{lem-cov} to bound $\| \bDelta \|$. Let $\widetilde{\bx}_i = ( \bSigma + \nu \bI)^{-1/2} \bx_i$ for $i \in [n]$, and $\widetilde{\bSigma} = \EE ( \widetilde{\bx}_i \otimes \widetilde{\bx}_i ) = ( \bSigma + \nu \bI)^{-1/2}  \bSigma ( \bSigma + \nu \bI)^{-1/2}$. Then, we have
	\begin{align*}
		& \| \widetilde{\bx}_i \|_{\HH}^2 
		= \sum_{j=1}^{\infty} \frac{ \sigma_j | v_j (x_i) |^2 }{ \sigma_j + \nu }
		\lesssim \sum_{j=1}^{\infty} \frac{ j^{-2} }{ j^{-2} + \nu }
		\asymp 
		\sum_{j=1}^{\lceil
			\nu^{-1/2}
			\rceil} 1 + \sum_{j=\lceil
			\nu^{-1/2}
			\rceil + 1 }^{\infty} j^{-2} / \nu 
		\lesssim \nu^{-1/2}
		, \\
		& \Tr (\widetilde{\bSigma}) =  \sum_{j=1}^{\infty} \frac{ \sigma_j }{ \sigma_j + \nu }  
		\lesssim \nu^{-1/2},
	\end{align*}
	and $\| \widetilde{\bSigma} \| \leq 1$. Hence, there exists a universal constant $c$ such that $\| \widetilde{\bx}_i \|_{\HH}^2 \leq c\nu^{-1/2} $ and $\Tr (\widetilde{\bSigma}) \leq c\nu^{-1/2}$. Applying \Cref{lem-cov} to $\{ \widetilde{\bx}_i \}_{i \in \cD_1}$, with parameters $\delta = 1/2$, $v = 1$ and $M^2 = r = c \nu^{-1/2}$, we see that
	\[
	\PP \bigg(
	\| \bDelta \|
	\leq C_1
	\max \bigg\{
	\sqrt{ \frac{ \log n }{n \nu^{1/2} } }
	,~ \frac{ \log n }{n \nu^{1/2} }
	\bigg\}
	\bigg)
	\geq 1/2
	.
	\]
	Here $C_1$ is a universal constant. The claim follows from the fact that $\nu \geq n^{-1}$.
\end{proof}

We now come to \eqref{eqn-LB-01}. For any $\lambda \in (0, \nu]$, we have $(\widehat{\bSigma} + \lambda \bI)^{-2} \widehat{\bSigma} \succeq (\widehat{\bSigma} + \nu \bI)^{-2} \widehat{\bSigma} $. Define $\bSigma_{\nu} = \bSigma + \nu \bI$ and $\widehat\bSigma_{\nu} = \widehat\bSigma + \nu \bI$. 
It suffices to show that when $n$ is sufficiently large,
\begin{align}
	\PP \bigg(
	\Tr (
	\bSigma \widehat{\bSigma}_{\nu}^{-2} \widehat{\bSigma}
	)
	\geq c_1 n^{2/5}	
	\bigg) \geq 1/2,
	\label{eqn-LB-01-1}
\end{align}
where $c_1$ is a universal constant. Note that
\begin{align*}
	& \Tr (  \bSigma \widehat{\bSigma}_{\nu}^{-2} \widehat{\bSigma} )
	= \langle  \bSigma , \widehat{\bSigma}_{\nu}^{-1}  \bSigma_{\nu}^{1/2}
	( \bSigma_{\nu}^{-1/2} \widehat{\bSigma}  \bSigma_{\nu}^{-1/2} ) \bSigma_{\nu}^{1/2} \widehat{\bSigma}_{\nu}^{-1} \rangle  
	.
\end{align*}
Let $C$ be the constant defined in Claim \ref{claim-concentration} and
\[
\zeta =  \frac{ C \log n }{ 2 n^{1/2} \nu^{1/4}  }
=   \frac{ C \log n }{ 2 n^{3/10}  }.
\]
With probability at least $1/2$,
\begin{align}
	\| \bSigma_{\nu}^{-1/2} \widehat{\bSigma}  \bSigma_{\nu}^{-1/2} - 
	\bSigma_{\nu}^{-1/2} \bSigma  \bSigma_{\nu}^{-1/2} \| \leq \zeta .
	\label{eqn-concentration-event}
\end{align}
and hence, $(1 - \zeta)  \bSigma_{\nu}
\preceq  \widehat{\bSigma}_{\nu}  \preceq (1 + \zeta)  \bSigma_{\nu} $. When the event \eqref{eqn-concentration-event} happens,
\begin{align}
	\Tr (  \bSigma \widehat{\bSigma}_{\nu}^{-2} \widehat{\bSigma} )
	& \geq  \langle  \bSigma , \widehat{\bSigma}_{\nu}^{-1}  \bSigma_{\nu}^{1/2}
	( \bSigma_{\nu}^{-1/2} \bSigma  \bSigma_{\nu}^{-1/2} ) \bSigma_{\nu}^{1/2} \widehat{\bSigma}_{\nu}^{-1} \rangle 
	- \zeta  \langle  \bSigma , \widehat{\bSigma}_{\nu}^{-1}  \bSigma_{\nu}^{1/2} \bSigma_{\nu}^{1/2} \widehat{\bSigma}_{\nu}^{-1} \rangle \notag  \\
	& = \langle  \bSigma , \widehat{\bSigma}_{\nu}^{-1}  \bSigma  \widehat{\bSigma}_{\nu}^{-1} \rangle 
	- \zeta  \langle  \bSigma , \widehat{\bSigma}_{\nu}^{-1}  \bSigma_{\nu} \widehat{\bSigma}_{\nu}^{-1} \rangle 
	\label{eqn-concentration-1}
\end{align}
When $n$ is sufficiently large, we have $\zeta \leq 1/2$ and
\begin{align}
	\langle  \bSigma , \widehat{\bSigma}_{\nu}^{-1}  \bSigma  \widehat{\bSigma}_{\nu}^{-1} \rangle 
	=  \langle  \widehat{\bSigma}_{\nu}^{-1},  \bSigma  \widehat{\bSigma}_{\nu}^{-1} \bSigma \rangle 
	\geq   (1+\zeta)^{-2}
	\langle \bSigma_{\nu}^{-1} ,  \bSigma  \bSigma_{\nu}^{-1} \bSigma \rangle 
	= (1+\zeta)^{-2}
	\langle \bSigma_{\nu}^{-2} ,  \bSigma^2  \rangle .
	\label{eqn-concentration-2}
\end{align}
Meanwhile,
\begin{align}
	& \langle  \bSigma , \widehat{\bSigma}_{\nu}^{-1}  \bSigma_{\nu} \widehat{\bSigma}_{\nu}^{-1} \rangle 
	\leq 
	\langle  \widehat{\bSigma}_{\nu}^{-1/2}  \bSigma  \widehat{\bSigma}_{\nu}^{-1/2} 
	, \widehat{\bSigma}_{\nu}^{-1/2}  \bSigma_{\nu} \widehat{\bSigma}_{\nu}^{-1/2} \rangle 
	\notag\\
	& \leq \Tr ( \widehat{\bSigma}_{\nu}^{-1/2}  \bSigma  \widehat{\bSigma}_{\nu}^{-1/2} ) \| \widehat{\bSigma}_{\nu}^{-1/2}  \bSigma_{\nu} \widehat{\bSigma}_{\nu}^{-1/2} \| 
	= \langle \bSigma , \widehat{\bSigma}_{\nu}^{-1} \rangle \| \widehat{\bSigma}_{\nu}^{-1/2}  \bSigma_{\nu} \widehat{\bSigma}_{\nu}^{-1/2} \| 
	\notag\\
	&\leq 
	\frac{  \langle \bSigma , \bSigma_{\nu}^{-1} \rangle 
	}{1-\zeta}
	\cdot \frac{1}{1-\zeta} = (1-\zeta)^{-2}  \langle \bSigma , \bSigma_{\nu}^{-1} \rangle 
	.
	\label{eqn-concentration-3}
\end{align}

By \eqref{eqn-concentration-1}, \eqref{eqn-concentration-2}, \eqref{eqn-concentration-3} and the fact that $\zeta \in [0, 1/2]$, on the event \eqref{eqn-concentration-event} we have
\begin{align*}
	\Tr (  \bSigma \widehat{\bSigma}_{\nu}^{-2} \widehat{\bSigma} )
	\geq  
	\frac{ \langle \bSigma_{\nu}^{-2} ,  \bSigma^2  \rangle }{(1+\zeta)^{2}}
	- \frac{\zeta}{(1-\zeta)^2}  \langle \bSigma , \bSigma_{\nu}^{-1} \rangle 
	\geq \frac{4}{9}  \langle \bSigma_{\nu}^{-2} ,  \bSigma^2  \rangle 
	- 2 \zeta  \langle \bSigma , \bSigma_{\nu}^{-1} \rangle  .
\end{align*}
Note that $ \langle \bSigma_{\nu}^{-2} ,  \bSigma^2  \rangle  \asymp \nu^{-1/2}$ and $ \langle \bSigma_{\nu}^{-1} ,  \bSigma  \rangle  \asymp \nu^{-1/2}$. Hence, for sufficiently large $n$, the event \eqref{eqn-concentration-event} implies $\Tr (  \bSigma \widehat{\bSigma}_{\nu}^{-2} \widehat{\bSigma} ) \gtrsim n^{2/5}$. This proves \eqref{eqn-LB-01-1}.

\subsubsection{Proof of \eqref{eqn-LB-02}}

Let $\bphi = \phi (x_0) $ and $\btheta^{\star} = ( \bSigma + \nu \bI )^{-1} \bphi / \| ( \bSigma + \nu \bI )^{-1} \bphi \|_{\HH}$. It suffices to prove that for sufficiently large $n$ and any $\lambda \geq \nu$, we have
\begin{align}
	\PP \bigg(
	\lambda^2 | \langle \bphi  , (\widehat{\bSigma} + \lambda \bI)^{-1}   \btheta^{\star} \rangle |^2 
	\geq c_2 n^{-2/5}
	\bigg) \geq 1/2,
	\label{eqn-LB-bias}
\end{align}
where $c_2$ is a universal constant. Choose any $\lambda \geq \nu = n^{-4/5}$. 

\begin{claim}
	$| \langle \bphi , (\bSigma + \lambda \bI)^{-1}   \btheta^{\star} \rangle | \gtrsim \lambda^{-1} \nu^{1/4} = \lambda^{-1} n^{-1/5}$.
\end{claim}
\begin{proof}
	Note that
	\begin{align*}
		\langle \bphi , (\bSigma + \lambda \bI)^{-1}   \btheta^{\star} \rangle
		& = \langle \bphi , (\bSigma + \lambda \bI)^{-1}  ( \bSigma + \nu \bI )^{-1} \bphi \rangle / \| ( \bSigma + \nu \bI )^{-1} \bphi \|_{\HH} .
	\end{align*}
	On the one hand,
	\begin{align*}
		&\| ( \bSigma + \nu \bI )^{-1} \bphi \|_{\HH}^2 
		\asymp \sum_{j=1}^{\infty} \frac{j^{-2}}{  ( j^{-2} + \nu )^2 }
		\asymp \sum_{j=1}^{\lceil \nu^{-1/2} \rceil } \frac{j^{-2}}{  ( j^{-2} )^2 } +  \sum_{j = \lceil \nu^{-1/2} \rceil + 1 }^{\infty} \frac{j^{-2}}{  \nu^2 } \notag \\
		& =  \sum_{j=1}^{\lceil \nu^{-1/2} \rceil } j^2 +   \sum_{j = \lceil \nu^{-1/2} \rceil + 1 }^{\infty} \frac{ j^{-2} }{\nu^2}
		\asymp \lceil \nu^{-1/2} \rceil^3 + \nu^{-2}  \lceil \nu^{-1/2} \rceil^{-1}
		\asymp  \nu^{-3/2} .
	\end{align*}
	On the other hand, for $\lambda \geq \nu$,
	\begin{align*}
		& \langle \bphi , (\bSigma + \lambda \bI)^{-1}  ( \bSigma + \nu \bI )^{-1} \bphi \rangle 
		\asymp \sum_{j=1}^{\infty} \frac{j^{-2}}{ (j^{-2} + \lambda ) (j^{-2} + \nu ) } \\
		& \gtrsim \sum_{j= \lceil \nu^{-1/2} \rceil }^{\infty} \frac{j^{-2}}{ \lambda \nu }
		\asymp (\lambda \nu)^{-1}  \lceil \nu^{-1/2} \rceil^{-1} \asymp \lambda^{-1} \nu^{-1/2}.
	\end{align*}
	The claimed result becomes obvious.
\end{proof}

\begin{claim}
	There exists a universal constant $c$ such that for sufficiently large $n$,
	\[
	\PP \bigg(
	| \langle \bphi , ( \widehat\bSigma_1 + \lambda \bI)^{-1}   \btheta^{\star} \rangle | - |\langle \bphi , (\bSigma + \lambda \bI)^{-1}   \btheta^{\star} \rangle
	|
	\leq 
	\frac{ c \log n }{ \lambda n^{3/10} } 
	\bigg)
	\geq \frac{1}{2}
	.
	\]
\end{claim}
\begin{proof}
	By direct calculation,
	\begin{align*}
		& | \langle \bphi , ( \widehat\bSigma + \lambda \bI)^{-1}   \btheta^{\star} \rangle | - |\langle \bphi , (\bSigma + \lambda \bI)^{-1}   \btheta^{\star} \rangle
		|
		\leq \| \bphi \|_{\HH} \| \btheta^{\star} \|_{\HH}  \|  ( \widehat\bSigma + \lambda \bI)^{-1}    -  ( \bSigma + \lambda \bI)^{-1}  \| \notag \\
		& \leq  \| ( \bSigma + \lambda \bI)^{-1/2} 
		\{
		[ ( \bSigma + \lambda \bI)^{-1/2} ( \widehat\bSigma + \lambda \bI) ( \bSigma + \lambda \bI)^{-1/2} ]^{-1}  - \bI \}  ( \bSigma + \lambda \bI)^{-1/2}  \| \notag\\
		& \leq \lambda^{-1}  \| [ ( \bSigma + \lambda \bI)^{-1/2} ( \widehat\bSigma + \lambda \bI) ( \bSigma + \lambda \bI)^{-1/2} ]^{-1}  - \bI \|
		= \lambda^{-1}
		\| (	\bI + \bDelta
		)^{-1}  - \bI \| .
	\end{align*}
	Claim \ref{claim-concentration} implies that when $n$ is large, with probability $1/2$, we have
	\[
	\| (	\bI + \bDelta
	)^{-1}  - \bI \| \lesssim \| \bDelta \|
	\lesssim n^{-3/10} \log n.
	\]
	This proves the claim.
\end{proof}
Finally,  \eqref{eqn-LB-bias} follows from the above two claims.

\subsection{Proof of \Cref{remark-n-eff}}\label{sec-remark-n-eff-proof}

Let $\conversionrate = \sup  \{ t :~ t \bSigma_0 \preceq \bSigma\}$. Since $\bSigma \succeq 0$, we have $\conversionrate \geq 0$. For the second scenario of \Cref{example-linear-conversion} with $d \to \infty$, we have $\langle \be_j , \bSigma_0 \be_j \rangle \asymp j^{-2}$ while $\langle \be_j , \bSigma \be_j \rangle \asymp j^{-4}$. This immediately shows $\conversionrate = 0$. To prove $\conversionrate = 0$ for \Cref{example-Sobolev-conversion}, we define $\HH$, $\sigma_j$ and $v_j$ as in \Cref{sec-example-Sobolev-conversion-proof}. Let
\[
\psi_j = ( \underbrace{ 0 , \cdots, 0 }_{(j-1) \text{ entries}} , \sqrt{\sigma_j} v_j (x_0), \sqrt{\sigma_{j+1}} v_{j+1} (x_0) , \cdots  ) , \qquad \forall j \in \ZZ_+.
\]
We have $\psi_j \in \HH$. The fact $\conversionrate \bSigma_0 \preceq \bSigma$ yields
\begin{align}
	\conversionrate | \langle \psi_j , \phi (x_0) \rangle |^2  \leq \langle \psi_j , \bSigma \psi_j \rangle , \qquad \forall j \in \ZZ_+.
	\label{eqn-example-Sobolev-1}
\end{align}
Since $x_0 = 1/2$, we have $|v_j (x_0)|^2 = 1/2$ and
\begin{align*}
	& \langle \psi_j , \phi (x_0) \rangle = \sum_{k=j}^{\infty} \sigma_k |v_k(x_0)|^2 
	=  \sum_{k=j}^{\infty} \bigg( \frac{2}{ (2k-1) \pi } \bigg)^2 \cdot \frac{1}{2}
	\asymp j^{-1} , \\
	& \langle \psi_j , \bSigma \psi_j \rangle = \sum_{k=j}^{\infty} \sigma_k^2 |v_k(x_0)|^2 
	=\sum_{k=j}^{\infty} \bigg( \frac{2}{ (2k-1) \pi } \bigg)^4 \cdot \frac{1}{2}
	\asymp j^{-3} .
\end{align*}
Then, the inequality \eqref{eqn-example-Sobolev-1} forces $\conversionrate = 0$.

\section{Proof of Theorem 5.2}\label{sec-thm-bias-variance-proof}

We introduce a deterministic oracle inequality for pseudo-labeling.
\begin{lemma}\label{lem-selection-regret}
	Let $\{ \by_{\lambda} \}_{ \lambda \in \Lambda} $ be a collection of vectors in $\RR^n$ and $\by^{\star} , \widetilde\by \in \RR^n$. Choose any $\widehat\lambda \in \argmin_{ \lambda \in \Lambda } \| \by_{\lambda} -  \widetilde\by   \|_2^2$. Define
	\[
	\overhead = \sup_{\lambda, \lambda' \in \Lambda} \bigg\langle
	\frac{
		\by_{\lambda} - \by_{\lambda'} 
	}{
		\| \by_{\lambda} - \by_{\lambda'}  \|_2
	}
	,  \widetilde\by  - \by^{\star}  \bigg\rangle,
	\]
	with the convention $\bm{0} / 0 = \bm{0}$. We have
\begin{align*}
& \| \by_{\widehat\lambda} - \by^{\star} \|_2 \leq \inf_{\lambda \in \Lambda }   \| \by_{\lambda} - \by^{\star} \|_2 + 2  \overhead .
\end{align*}
\end{lemma}

\begin{proof}[\bf Proof of \Cref{lem-selection-regret}]
	Choose any $\lambda \in \Lambda$. By direct calculation, we have
	\begin{align*}
	&\| \by_{\widehat\lambda} - \by^{\star} \|_2^2 - \| \by_{{\lambda} } - \by^{\star} \|_2^2 = \langle \by_{\widehat\lambda} - \by_{{\lambda} } , \by_{\widehat\lambda} + \by_{{\lambda} } - 2 \by^{\star} \rangle  .
	\end{align*}
	By the assumption $\widehat\lambda \in \argmin_{ \lambda \in \Lambda } \| \by_{{\lambda} } -  \widetilde\by  \|_2^2 $ and the equality above,
	\begin{align*}
	0 & \geq \| \by_{\widehat\lambda}  -  \widetilde\by  \|_2^2 - \| \by_{{\lambda} } -  \widetilde\by  \|_2^2 = \langle \by_{\widehat\lambda}  - \by_{{\lambda} } , \by_{\widehat\lambda} + \by_{{\lambda} } - 2  \widetilde\by  \rangle \\
	& = \| \by_{\widehat\lambda} - \by^{\star} \|_2^2 - \| \by_{{\lambda} } - \by^{\star} \|_2^2 - 2 \langle \by_{\widehat\lambda}  - \by_{{\lambda} } ,  \widetilde\by  - \by^{\star} \rangle .
	\end{align*}
	Hence,
	\begin{align*}
	\| \by_{\widehat\lambda} - \by^{\star} \|_2^2 & \leq \| \by_{{\lambda} } - \by^{\star} \|_2^2 +  2 \langle \by_{\widehat\lambda}  - \by_{{\lambda} } ,  \widetilde\by  - \by^{\star} \rangle \\
	&= \| \by_{{\lambda} } - \by^{\star} \|_2^2 +  2 \| \by_{\widehat\lambda}  - \by_{{\lambda} } \|_2 \bigg\langle \frac{ \by_{\widehat\lambda}  - \by_{{\lambda} } }{ \| \by_{\widehat\lambda}  - \by_{{\lambda} } \|_2 } ,  \widetilde\by  - \by^{\star} \bigg\rangle \\
	& \leq  \| \by_{{\lambda} } - \by^{\star} \|_2^2 +  2 \| \by_{\widehat\lambda}  - \by_{{\lambda} } \|_2 \overhead \\
	& \leq \| \by_{{\lambda} } - \by^{\star} \|_2^2 +  2 (  \| \by_{\widehat\lambda}  - \by^{\star}   \|_2
	+ \| \by^{\star} - \by_{{\lambda}}  \|_2
	) \overhead ,
	\end{align*}
	where the last inequality follows from the fact $\overhead \geq 0$ and the triangle's inequality. Rearranging the terms, we get
	\begin{align*}
	& \| \by_{\widehat\lambda} - \by^{\star} \|_2^2 - 2  \| \by_{\widehat\lambda}  - \by^{\star}   \|_2 \overhead  \leq \| \by_{{\lambda} } - \by^{\star} \|_2^2 +  2  \| \by^{\star} - \by_{{\lambda}}  \|_2 \overhead 
	\end{align*}
	and thus $ ( \| \by_{\widehat\lambda} - \by^{\star} \|_2 - \overhead )^2 \leq ( \| \by_{{\lambda} } - \by^{\star} \|_2 + \overhead )^2$. As a result,
	\begin{align*}
	& \| \by_{\widehat\lambda} - \by^{\star} \|_2  \leq \| \by_{{\lambda} } - \by^{\star} \|_2 + 2 \overhead , \qquad \forall \lambda \in \Lambda.
	\end{align*}
\end{proof}

To prove \Cref{thm-bias-variance}, we define $\by^{\star} = ( g^{\star} (z_1) , \cdots, g^{\star} (z_n) )^{\top}$ and
\[
\by_j = ( g_j (z_1) , \cdots, g_j (z_n) )^{\top}, \qquad \forall j \in [m].
\]
\Cref{lem-selection-regret} leads to
\begin{align*}
& \cL_n ( g_{\widehat{j}}  )  \leq \min_{ j \in [m] }   \cL_n ( g_j )  + 2  \overhead / \sqrt{n},
\end{align*}
where
\[
\overhead =  \max_{j, j' \in [m]} \bigg\langle
\frac{
	\by_{j} - \by_{j'} 
}{
	\| \by_{j} - \by_{j'}  \|_2
}
,  \widetilde\by  - \by^{\star}  \bigg\rangle  ,
\]
with the convention $\bm{0} / 0 = \bm{0}$. It is easily seen that
\begin{align*}
& 0 \leq \overhead 
\leq \| \EE \widetilde\by -  \by^{\star}  \|_2 + 
\max_{j, j' \in [m]} \bigg\langle
\frac{
	\by_{j} - \by_{j'} 
}{
	\| \by_{j} - \by_{j'}  \|_2
}
, \widetilde\by  -  \EE \widetilde\by  \bigg\rangle .
\end{align*}
On the one hand, $n^{-1/2} \| \EE \widetilde\by -  \by^{\star}  \|_2 = \cL_n ( \EE \widetilde{g} )$.
On the other hand,
\[
\bigg\|
\bigg\langle
\frac{
	\by_{j} - \by_{j'} 
}{
	\| \by_{j} - \by_{j'}  \|_2
}
, \widetilde\by  -  \EE \widetilde\by  \bigg\rangle
\bigg\|_{\psi_2} \leq V ,\qquad \forall j , j' \in [m] .
\]
The sub-Gaussian concentration and union bounds imply that with probability $1 - \delta$,
\[
\max_{j, j' \in [m]} 
\bigg|
\bigg\langle
\frac{
	\by_{j} - \by_{j'} 
}{
	\| \by_{j} - \by_{j'}  \|_2
}
, \widetilde\by  -  \EE \widetilde\by  \bigg\rangle
\bigg| \leq c_1 V \sqrt{ \log (m / \delta) } ,
\]
where $c_1$ is a universal constant. We get the desired concentration inequality in \Cref{thm-bias-variance} by combing all the above estimates and taking $c = 2 c_1$. To prove the bound on the expectation, fix any $\gamma > 0$ and define
\[
W = 
\bigg[
\cL_n ( g_{\widehat{j}}  )
- \bigg(
 \min_{ j \in [m]} \cL_n ( g_j )
+   2
\cL_n ( \EE \widetilde{g} ) 
\bigg)
\bigg]_+
.
\]
Since $(a+b)^2 = a^2 + 2ab + b^2 \leq (1 + \gamma) a^2 + (1 + \gamma^{-1}) b^2$ holds for any $a, b \geq 0$ and $\gamma > 0$. We have
\begin{align*}
 \cL_n^2 ( g_{\widehat{j}}  )
	& \leq   \bigg( \min_{ j \in [m]} \cL_n ( g_j ) 
	+   2
	\cL_n ( \EE \widetilde{g} ) 
	+ W
	\bigg)^2 \\
& \leq (1 + \gamma) \min_{ j \in [m]} \cL_n^2 ( g_j ) 
+ (1 + \gamma^{-1}) [2
\cL_n ( \EE \widetilde{g} ) 
+ W]^2 \\
& \leq (1 + \gamma) \min_{ j \in [m]} \cL_n^2 ( g_j ) 
+ (1 + \gamma^{-1})[
2 \cdot 4 \cL_n^2 ( \EE \widetilde{g} ) 
+ 2 \cdot W^2
].
\end{align*}
Hence,
\begin{align}
\EE \cL_n^2 ( g_{\widehat{j}}  )
& \leq  (1 + \gamma) \min_{ j \in [m]} \cL_n^2 ( g_j ) 
+ 2 (1 + \gamma^{-1}) [
4 \cL_n^2 ( \EE \widetilde{g} ) 
+  \EE W^2
].
\label{eqn-thm-bias-variance-1}
\end{align}

To bound $\EE W^2$, note that $W \geq 0$ and there exists a constant $c$ such that
\[
\PP \bigg(
W > c V \sqrt{  \frac{ 
 \log ( m / \delta )
 }{n}
}
\bigg) \leq \delta , \qquad \forall \delta \in (0, 1].
\]
Let $t = \log ( m / \delta )$. Then, $\delta = m e^{-t}$. We have
\[
\PP \bigg(
W^2 >   \frac{ 
	c^2 V^2 }{n} t
\bigg) \leq m e^{-t} , \qquad \forall t \geq \log m .
\]
As a result, we get
\begin{align*}
& \bigg(
\frac{	c^2 V^2 }{n} 
\bigg)^{-1} \EE W^2 
 = \int_{0}^{\infty} \PP \bigg(
W^2 >   \frac{ 
c^2 V^2 }{n} t
\bigg) \rd t
 \leq  \int_{0}^{\log m} 1 \rd t +  \int_{\log m}^{\infty} m e^{-t} \rd t =
\log m + 1.
\end{align*}
Plugging this into \eqref{eqn-thm-bias-variance-1} yields
\begin{align*}
	\EE \cL_n^2 ( g_{\widehat{j}}  )
	& \leq (1+\gamma)  \min_{ j \in [m]} \cL_n^2 ( g_j ) + 2 (1 + \gamma^{-1}) 
\bigg(
	4  \cL_n^2 ( \EE \widetilde{g} ) +
\frac{c^2 V^2 ( 1 + \log m )}{n}
\bigg)
	.
\end{align*}
The proof is finished by taking the infimum over $\gamma > 0$.

\section{Technical lemmas}

\begin{lemma}\label{lem-subg-norm}
	Suppose that $\bx \in \RR^d$ is a zero-mean random vector with $\| \bx \|_{\psi_2} \leq 1$. There exists a universal constant $C > 0$ such that for any symmetric and positive semi-definite matrix $\bSigma \in \RR^{d\times d}$, 
	\begin{align*}
	\PP\Big(
	\bx^{\top} \bSigma \bx  \leq C 
	\Tr(\bSigma) t
	\Big) \geq 1 - e^{- r(\bSigma) t }, \qquad \forall t \geq 1.
	\end{align*}
	Here $r(\bSigma) = \Tr(\bSigma) / \| \bSigma \|_2$ is the effective rank of $\bSigma$.
\end{lemma}

\begin{proof}[\bf Proof of \Cref{lem-subg-norm}]
This is a direct corollary of Lemma J.4 in \cite{DDW21}.
\end{proof}


\begin{lemma}\label{lem-cov}
Let $\{ \bx_i \}_{i=1}^n$ be i.i.d.~random elements in a separable Hilbert space $\HH$ with $\bSigma = \EE ( \bx_i \otimes \bx_i )$ being trace class. Define $\widehat{\bSigma} = \frac{1}{n} \sum_{i=1}^n  \bx_i \otimes \bx_i $.
\begin{enumerate}
\item If $\|   \bx_i \|_{\HH} \leq M$ holds almost surely for some constant $M$, then for any $v^2 \geq   \| \bSigma \|$, $r \geq \Tr (\bSigma) / v^2$ and $\delta \in (0, r]$,
\[
\PP \bigg(
\| \widehat\bSigma - \bSigma  \|
\leq  
\max \bigg\{
2 v \sqrt{ \frac{ M^2\log (14 r / \delta ) }{n} }
,~ \frac{8M^2 \log (14 r / \delta )}{3n}
\bigg\}
\bigg)
\geq 1 - \delta .
\]
\item If $\| \langle   \bx_{i} , \bv \rangle \|_{\psi_2}^2 \leq \kappa \EE | \langle  \bx_{i} , \bv \rangle |^2 $ holds for all $\bv \in \HH$ and some constant $\kappa$, then there exists a constant $C$ determined by $\kappa$ such that
\begin{align*}
\PP \bigg(
\| \widehat\bSigma  - \bSigma \| \leq C \| \bSigma \| \max\bigg\{
\sqrt{ \frac{ r  }{ n } } ,~ \frac{ r  }{n},~ \sqrt{ \frac{ \log(1/\delta) }{n} } ,~ \frac{ \log(1/\delta) }{n}
\bigg\}
\bigg) \geq 1 - \delta 
\end{align*}
holds for all $\delta \in (0, 1/e]$. Here $r = \Tr(\bSigma) / \| \bSigma \|_2$ is the effective rank of $\bSigma$.
\end{enumerate}
\end{lemma}

\begin{proof}[\bf Proof of \Cref{lem-cov}]
The second part directly follows from Theorem 9 in \cite{KLo17}.
To prove the first part, we invoke a Bernstein inequality for self-adjoint operators. The Euclidean case is an intermediate result in the proof of Theorem 3.1 in \cite{Min17}. Section 3.2 of the same paper extends the result to the Hilbert setting.

\begin{lemma}\label{lem-matrix-bernstein}
Let $\HH$ be a separable Hilbert space, and $\{ \bX_i \}_{i=1}^n$ be independent self-adjoint, Hilbert-Schmidt operators satisfying $\EE \bX_i = \bm{0}$ and $\| \sum_{i=1}^n \EE \bX_i^2 \|_2 \leq \sigma^2$. Assume that $\| \bX_i \|_{\HH} \leq U$ holds almost surely for all $i$ and some $U > 0$. If $t^2 \geq \sigma^2 + U t / 3$, then
\begin{align*}
\PP \bigg( \bigg\| \sum_{i=1}^n \bX_i \bigg\| \geq t  \bigg) \leq  \frac{ 14 \Tr ( \sum_{i=1}^n \EE \bX_i^2 ) }{ \sigma^2 } \exp \bigg(
- \frac{ t^2 / 2 }{ \sigma^2 + U t / 3 }
\bigg) .
\end{align*}
\end{lemma}

Let $\bX_i =  \bx_i \otimes  \bx_i - \bSigma$. Under Assumptions \ref{assumption-linear-model-rkhs} and \ref{assumption-covariates-bounded},
\begin{align*}
& \| \bX_i \| \leq \| \bx_i \otimes  \bx_i   \| + \| \EE (\bx_i \otimes  \bx_i ) \| \leq 2  M^2 , \\
& \EE \bX_i^2 = \EE (\| \bx_i \|_{\HH}^2 \bx_i \otimes  \bx_i) - \bSigma^2
\preceq \EE (\| \bx_i \|_{\HH}^2 \bx_i \otimes  \bx_i) \preceq M^2 \bSigma \preceq M^2 v^2 \bI.
\end{align*}
Choose any ${r} \geq \Tr (\bSigma) / v^2$ and $\delta \in (0, r]$. Let $U = 2 M^2$, $\sigma^2 = n M^2 v^2$ and
\[
t = \max \bigg\{
2 v \sqrt{ n M^2\log (14 r / \delta ) }
,~ \frac{8M^2 \log (14 r / \delta )}{3}
\bigg\}.
\]
It is easily seen that $t \geq \max \{ \sqrt{2} \sigma , 2 U / 3 \}$ and thus $t^2 \geq \sigma^2 + U t / 3$. \Cref{lem-matrix-bernstein} implies that
\begin{align*}
& \PP  ( n \| \widehat\bSigma - \bSigma  \| \geq t  )  \leq 
\frac{ 14 n M^2 \Tr (\bSigma) }{ \sigma^2 } \exp \bigg(
- \frac{ t^2 / 2 }{ n M^2 v^2 + 2 M^2 t / 3 }
\bigg) \\
& \leq
\frac{ 14 n M^2 \Tr (\bSigma) }{ n M^2 v^2 } \exp \bigg(
- \frac{ t^2 / 2 }{ n M^2 v^2 + 2 M^2 t / 3 }
\bigg) \leq
14 {r} \exp \bigg(
- \frac{ t^2 / 2 }{ n M^2 v^2 + 2 M^2 t / 3 }
\bigg) \\
& \leq 
14 {r} \exp \bigg(
- \frac{ t^2 / 2 }{ 2 \max \{ n M^2 v^2 , 2 M^2 t / 3 \} }
\bigg) \\
& = 14 {r} \exp \bigg(
- \min \bigg\{
\frac{ t^2 }{ 4 n M^2 v^2  }
,
\frac{ t }{ 8 M^2 / 3  }
\bigg\}
\bigg)
= \delta.
\end{align*}
\end{proof}

\begin{corollary}\label{cor-cov-rkhs}
	Let $\{ \bx_i \}_{i=1}^n$ be i.i.d.~random elements in a separable Hilbert space $\HH$ with $\bSigma = \EE ( \bx_i \otimes \bx_i )$ being trace class. Define $\widehat{\bSigma} = \frac{1}{n} \sum_{i=1}^n  \bx_i \otimes \bx_i $.
	Choose any constant $\gamma \in (0, 1)$ and define an event $\cA = \{  (1 - \gamma) ( \bSigma + \lambda \bI) \preceq  \widehat\bSigma + \lambda \bI \preceq (1 + \gamma) ( \bSigma + \lambda \bI)  \}$. 
	\begin{enumerate}
		\item If $\|   \bx_i \|_{\HH} \leq M$ almost surely, $\delta \in (0, 1/e]$, $\gamma \in (0, 1]$ and $\lambda \geq \frac{ 4 M^2 \log (14 n / \delta)  }{ \gamma^2 n } $, then $\PP ( \cA ) \geq 1 - \delta$.
		\item If $\| \langle   \bx_{i} , \bv \rangle \|_{\psi_2}^2 \leq \kappa \EE | \langle  \bx_{i} , \bv \rangle |^2 $ holds for all $\bv \in \HH$ and some constant $\kappa$, then there exists a constant $C $ determined by $\kappa$ such that $\PP ( \cA ) \geq 1 - \delta$ holds so long as $\delta \in (0, 1 / e]$ and $\lambda \geq \frac{  C \Tr (\bSigma) \log (1 / \delta)   }{ \gamma^2 n }  $.
	\end{enumerate}
\end{corollary}

\begin{proof}[\bf Proof of \Cref{cor-cov-rkhs}]
We will apply \Cref{lem-cov} to $\widetilde\bx_{i} = (\bSigma + \lambda \bI)^{-1/2} \bx_{i}$ and $\widetilde\bSigma =  (\bSigma + \lambda \bI)^{-1} \bSigma$. 

Suppose that $\|   \bx_i \|_{\HH} \leq M$ holds almost surely for some constant $M$.
We have $\| \widetilde\bx_{i} \|_{\HH} \leq M / \sqrt{\lambda}$, $\Tr (\widetilde\bSigma) \leq \Tr (\bSigma) / \lambda \leq M^2 / \lambda$ and $\| \widetilde \bSigma \| \leq 1$. 
Hence, the $M$ and $v^2$ in \Cref{lem-cov} can be set to be $M / \sqrt{\lambda}$ and $1$, respectively. 
Choose any $ \delta \in (0, 1/e]$. For any $\lambda \geq \frac{ M^2 \log(n/\delta) }{n}$, we have
\[
\Tr(\widetilde\bSigma) / v^2 \leq M^2 / \lambda \leq \frac{n}{ \log(n/\delta) } \leq n .
\]
Taking $r = n$ in \Cref{lem-cov}, we see that with probability at least $1 - \delta$,
\begin{align*}
\| (\bSigma + \lambda \bI )^{-1/2} ( \widehat\bSigma - \bSigma ) (\bSigma + \lambda \bI )^{-1/2}  \|
& \leq  
\max \bigg\{
2 \sqrt{ \frac{ M^2\log (14 n / \delta ) }{n\lambda} }
,~ \frac{8M^2 \log (14 n / \delta )}{3n\lambda}
\bigg\} .
\end{align*}

Choose any $\gamma \in (0, 1)$. When $\lambda \geq \frac{ 4 M^2\log (14 n / \delta ) }{\gamma^2 n} $, we have
\begin{align*}
& \PP \bigg(
\| (\bSigma + \lambda \bI )^{-1/2} ( \widehat\bSigma - \bSigma ) (\bSigma + \lambda \bI )^{-1/2}  \|
\leq  \gamma  
\bigg) \geq 1 - \delta .
\end{align*}
This proves the first part of \Cref{cor-cov-rkhs}.

Now, suppose that $\| \langle   \bx_{i} , \bv \rangle \|_{\psi_2}^2 \leq \kappa \EE | \langle  \bx_{i} , \bv \rangle |^2 $ holds for all $\bv \in \HH$ and some constant $\kappa$. Then, the same property holds for $\widetilde\bx_{i} $. The second part of \Cref{lem-cov} implies that for any $\delta \in (0, 1/e]$,
\begin{align*}
& \PP \bigg(
\| (\bSigma + \lambda \bI )^{-1/2} ( \widehat\bSigma - \bSigma ) (\bSigma + \lambda \bI )^{-1/2}  \|
\\
&~~~~
\leq C \| \widetilde\bSigma \| \max\bigg\{
\sqrt{ \frac{ \widetilde{r}  }{ n } } ,~ \frac{ \widetilde{r} }{n},~ \sqrt{ \frac{ \log(1/\delta) }{n} } ,~ \frac{ \log(1/\delta) }{n}
\bigg\}
\bigg) \geq 1 - \delta .
\end{align*}
Here $\widetilde{r} = \Tr(\widetilde\bSigma) / \| \widetilde\bSigma \|_2$.
Note that $\| \widetilde\bSigma \| \leq 1$ and $\| \widetilde\bSigma \| \leq \Tr (\widetilde\bSigma) \leq \Tr (\bSigma) / \lambda$. We have
\begin{align*}
& \| \widetilde\bSigma \|  \sqrt{ \frac{ \widetilde{r} }{n} } = \sqrt{ \Tr (\widetilde\bSigma) \| \widetilde\bSigma \| / n } \leq \sqrt{ \Tr (\widetilde\bSigma)  / n } \leq  \sqrt{ \frac{\Tr(\bSigma)}{n \lambda } } , \\
& \| \widetilde\bSigma \| \frac{ \widetilde{r} }{n} =   \Tr (\widetilde\bSigma) / n \leq \frac{\Tr(\bSigma)}{n \lambda } , \\
& \| \widetilde\bSigma \|  \sqrt{ \frac{ \log (1/\delta) }{n} } =  \sqrt{ \frac{ \| \widetilde\bSigma \|^2 \log (1/\delta) }{n} }  \leq \sqrt{ \frac{ \Tr ( \bSigma ) \log (1/\delta) }{n \lambda } } , \\
&  \| \widetilde\bSigma \|  \frac{ \log (1/\delta) }{n}  \leq  \frac{ \Tr ( \bSigma ) \log (1/\delta) }{n \lambda } 
\end{align*}
Choose any $\delta \in (0, 1/e]$ and $\gamma \in (0, 1]$. When $\lambda \geq  \frac{ C^2 \Tr (\bSigma) \log (1/\delta) }{ \gamma^2 n } $, we have
\begin{align*}
\| \widetilde\bSigma \| \max\bigg\{
\sqrt{ \frac{ \widetilde{r}  }{ n } } ,~ \frac{ \widetilde{r} }{n},~ \sqrt{ \frac{ \log(1/\delta) }{n} } ,~ \frac{ \log(1/\delta) }{n}
\bigg\} \leq \sqrt{ \frac{ \Tr ( \bSigma ) \log (1/\delta) }{n \lambda } } \leq \gamma / C 
\end{align*}
and thus,
\begin{align*}
\PP \bigg(
\| (\bSigma + \lambda \bI )^{-1/2} ( \widehat\bSigma - \bSigma ) (\bSigma + \lambda \bI )^{-1/2}  \|
\leq \gamma 
\bigg) \geq 1 - \delta .
\end{align*}
The proof is completed by redefining $C$.
\end{proof}

The following lemmas are crucial for connecting the empirical (in-sample) and the population (out-of-sample) squared errors.

\begin{lemma}\label{lem-one-sided-Bernstein}
	Let $\{ X_i \}_{i=1}^n$ be independent random variables. If $X_i \leq b$ almost surely and $\EE X_i^2 \leq v^2$, then
	\begin{align*}
	\PP \bigg(
	\frac{1}{n} \sum_{i=1}^n (X_i - \EE X_i) \leq 
	\max \bigg\{
	2  \sqrt{ \frac{ v^2 \log (1/\delta) }{ n } } ,~  \frac{ 4 b \log (1/\delta) }{ 3 n }
	\bigg\}
	\bigg) \geq 1 - \delta , \qquad \forall \delta \in (0, 1).
	\end{align*}
	Consequently, if $X_i \geq 0$ almost surely and $\EE X_i^2 \leq v^2$, then
	\begin{align*}
	\PP \bigg(
	\frac{1}{n} \sum_{i=1}^n (X_i - \EE X_i) \geq - 2  \sqrt{ \frac{ v^2 \log (1/\delta) }{ n } }
	\bigg) \geq 1 - \delta , \qquad \forall \delta \in (0, 1).
	\end{align*}
\end{lemma}

\begin{proof}[\bf Proof of \Cref{lem-one-sided-Bernstein}]
	Suppose that $X_i \leq b$ almost surely and $\EE X_i^2 \leq v^2$. Proposition 2.14 in \cite{Wai19} implies that for any $t \in \RR$,
	\begin{align*}
	& \PP \bigg(
	\frac{1}{n} \sum_{i=1}^n (X_i - \EE X_i) \geq t
	\bigg)  \leq \exp \bigg(
	- \frac{n t^2 / 2}{ v^2 + b t / 3 } 
	\bigg)   \leq \exp \bigg(
	- \frac{n t^2 / 2}{ 2 \max \{  v^2 ,~ b t / 3 \} } 
	\bigg) \\
	& = \max \bigg\{ 
	\exp \bigg(
	- \frac{n t^2}{ 4  v^2 } 
	\bigg) ,~
	\exp \bigg(
	- \frac{n t^2}{ 4 b t / 3 } 
	\bigg)
	\bigg\} 
	=  \max \bigg\{ 
	\exp \bigg(
	- \frac{n t^2}{ 4  v^2 } 
	\bigg) ,~
	\exp \bigg(
	- \frac{3 n t }{ 4 b  } 
	\bigg)
	\bigg\} .
	\end{align*}
	For $\delta \in (0, 1)$ and $t = \max  \{
	2 v \sqrt{ \frac{ \log (1/\delta) }{ n } } ,  \frac{ 4 b \log (1/\delta) }{ 3 n }
	\}$, we have $\frac{n t^2}{ 4  v^2 } \geq \log (1/\delta)$ and $\frac{3 n t }{ 4 b  }  \geq \log (1/\delta)$. Then, $\PP  (
	\frac{1}{n} \sum_{i=1}^n (X_i - \EE X_i) \geq t
	)  \leq \delta$. This proves the first part of the lemma. The second part can be obtained by applying the first part to $-X_i$ and setting $b = 0$.
\end{proof}

\begin{lemma}\label{lem-L2}
	Let $\{ z_i \}_{i=1}^n$ be i.i.d.~samples in a space $\cZ$. For any $g:~\cZ \to \RR$, define $\| g \|_n = (
	n^{-1} \sum_{i=1}^{n} g^2 (z_i) 
	)^{1/2}$ and $\| g \|_{ L^2 } = \sqrt{ \EE g^2 (z_1) }$. Let $f:~\cZ \to \RR$ be a random function that is independent of $z_i$'s.
	\begin{enumerate}
		\item Suppose that $\PP ( | f(z_1) | > r ) \leq \varepsilon$ and $\EE | f(z_1) |^4 \leq U^4 $ hold for some deterministic $r, U \geq 0$. Then,
		\begin{align*}
		\PP \bigg(
		\Big|
		\| f \|_n - \| f \|_{L^2 }
		\Big|
		\leq 
		r\sqrt{
			\frac{
				7  \log (2 / \delta)
			}{n}
		}
		+ U \varepsilon^{1/4}
		\bigg)
		\geq 1 - n \varepsilon - \delta, \qquad \forall \delta \in (0, 1).
		\end{align*}
		\item Define $h (\cdot) = f(\cdot) / \| f \|_{ L^2 }$ and suppose that conditioned on $f$, $ h(z_i)$ is sub-Gaussian in the sense that
		\[
		\sup_{p \geq 1 } \{ p^{-1/2} \EE^{1/p}  [ |h(z_i)|^p  | f  ] \} \leq K \in [1, +\infty) .
		\]
There exists a constant $C \geq 1$ such that when $\delta \in (0, 1)$ and $n \geq C^2 K^4 \log(2/\delta) $, we have
\begin{align*}
\PP \bigg(
\Big| \| f \|_n^2 -  \| f \|_{L^2}^2 \Big| \leq 
 \| f \|_{L^2}^2 \cdot C K^2 \sqrt{ \frac{ \log (2/\delta) }{ n } }
\bigg)  \geq 1 -  \delta .
\end{align*}
	\end{enumerate}
\end{lemma}

\begin{proof}[\bf Proof of \Cref{lem-L2}]
	We will prove Part 1 by a truncation argument. Define $\bar{f}(z) = f(z) \bm{1}_{ \{ |f(z)| \leq r \} }$, $\forall z \in \cZ$. We have
	\begin{align}
	\| f \|_n - \| f \|_{L^2}
	= (  \| f \|_n - \| \bar{f} \|_n )
	+  (  \| \bar{f} \|_n - \| \bar{f} \|_{L^2} )
	+  ( \| \bar{f} \|_{L^2} - \| f \|_{L^2} ).
	\label{eqn-lem-L2-0}
	\end{align}
	By union bounds, the event $\cap_{i=1}^n \{ \bar{f} (z_i) = f (z_i)  \}$ happens with probability at least $1 - n \varepsilon$. Hence,
	\begin{align}
	\PP (
	\| f \|_n = \| \bar{f} \|_n 
	) \geq 1 - n \varepsilon.
	\label{eqn-lem-L2-1}
	\end{align}
	
	Next, we analyze $ \| \bar{f} \|_n - \| \bar{f} \|_{L^2}$. Choose any deterministic function $g$ that satisfies $|g(z)| \leq r$, $\forall z$. Then, $0 \leq g^2(z_i) \leq r^2$ and $\EE | g^2 (z_i) |^2 \leq r^2  \| g \|_{L^2}^2$. Choose any $\delta \in (0, 1)$. By \Cref{lem-one-sided-Bernstein}, with probability at least $1 - \delta$,
	\begin{align*}
	\bigg| \frac{1}{n} \sum_{i=1}^{n} [ g^2(z_i) - \EE g^2(z_i) ] \bigg|
	\leq 2 \sqrt{ \frac{
			r^2  \| g \|_{L^2}^2 \cdot \log (2 / \delta)
		}{n} }
	+ \frac{
		4 \cdot r^2 \cdot  \log (2 / \delta)
	}{3n}
	.
	\end{align*}
	On that event, we have
	\begin{align}
	& \| g \|_n^2
	\leq \| g \|_{L^2}^2 +
	2  \| g \|_{L^2} \sqrt{ \frac{
			r^2   \log (2 / \delta)
		}{n} }
	+ \frac{
		4  r^2  \log (2 / \delta)
	}{3n} 
	\notag	\\&
	\leq 
	\bigg(
	\| g \|_{L^2} + 
	\sqrt{ \frac{
			4 r^2  \log (2 / \delta)
		}{3 n} }
	\bigg)^2
	\label{eqn-lem-L2-2-1}
	\end{align}
	and
	\begin{align*}
	\| g \|_n^2
	&	\geq \| g \|_{L^2}^2 -
	2  \| g \|_{L^2} \sqrt{ \frac{
			r^2 \log (2 / \delta)
		}{n} }
	- \frac{
		4  r^2  \log (2 / \delta)
	}{3n} 
	\\&
	= \bigg(
	\| g \|_{L^2}  -
	\sqrt{ \frac{
			r^2  \log (2 / \delta)
		}{n} }
	\bigg)^2 - 
	\frac{
		7  r^2  \log (2 / \delta)
	}{3n}.
	\end{align*}
	The latter inequality leads to
	\begin{align}
	\bigg(
	\| g \|_{L^2} -
	\sqrt{ \frac{
			r^2  \log (2 / \delta)
		}{n} }
	\bigg)^2 
	& \leq 	\| g \|_n^2
	+
	\frac{
		7  r^2  \log (2 / \delta)
	}{3n}
	\leq 
	\bigg(
	\| g \|_n + 
	\sqrt{
		\frac{
			7  r^2  \log (2 / \delta)
		}{3n}
	}
	\bigg)^2
	.
	\label{eqn-lem-L2-2-2}
	\end{align}
	From \eqref{eqn-lem-L2-2-1} and \eqref{eqn-lem-L2-2-2} we obtain that
	\begin{align*}
	\Big| 	\| g \|_n - \| g \|_{L^2}
	\Big| \leq 
	(1 + \sqrt{7/3})
	\sqrt{
		\frac{
			r^2  \log (2 / \delta)
		}{n}
	}
	\leq 	r\sqrt{
		\frac{
			7  \log (2 / \delta)
		}{n}
	}
	.
	\end{align*}
	Since $ \bar{f} $ is independent of $z_i$'s and takes values in $[-r, r]$, the above tail bound yields
	\begin{align}
	\PP \bigg(
	\Big| 	\| \bar{f} \|_n - \| \bar{f} \|_{L^2}
	\Big| 
	\leq 	r\sqrt{
		\frac{
			7  \log (2 / \delta)
		}{n}
	}
	\bigg) \geq 1 - \delta
	.
	\label{eqn-lem-L2-2}
	\end{align}
	
	We finally come to $\| \bar{f} \|_{L^2} - \| f \|_{L^2} $. It is easily seen that
	\begin{align}
	\Big|
	\| \bar{f} \|_{L^2} - \| f \|_{L^2} 
	\Big|
	& \leq 
	\| \bar{f} - f \|_{L^2} 
	= \sqrt{
		\EE [ f^2(z_1) \bm{1}_{ \{ 
			|f(z_1)| > r
			\} }
		]
	} 
	\notag \\&
	\leq \bigg(
	\sqrt{ \EE f^4 (z_1) } \cdot 
	\sqrt{
		\PP (
		|f(z_1)| > r
		)
	}
	\bigg)^{1/2}
	\leq U \varepsilon^{1/4}.
	\label{eqn-lem-L2-3}
	\end{align}
	Part 1 follows from \eqref{eqn-lem-L2-0}, \eqref{eqn-lem-L2-1}, \eqref{eqn-lem-L2-2} and \eqref{eqn-lem-L2-3}.

To prove Part 2, note that conditioned on $f$, $\{ h^2 (z_i) \}_{i=1}^n$ are i.i.d., and the $\psi_1$ norm of $h^2(z_i) $ is $O( K^2 )$. By Proposition 5.17 in \cite{Ver10} (a Bernstein-type inequality), 
\begin{align*}
&
\PP \bigg( \Big| \| h \|_n^2 -  \| h \|_{L^2}^2 \Big| \geq t 
~\bigg|~f
\bigg)  
\leq 2 \exp \bigg[ - c n \min \bigg\{
\bigg(
\frac{t}{ K^2 }
\bigg)^2 ,~
\frac{t}{ K^2 }
\bigg\} \bigg].
\end{align*}
Here $c$ is a universal constant. Choose any $\delta \in (0, 1)$ and set
\[
t = K^2  \max\bigg\{ 
\sqrt{ \frac{\log(2/\delta)}{ c n } } , ~  \frac{\log(2/\delta)}{ c n }
\bigg\}.
\]
We get
\[
\PP \bigg( \Big| \| h \|_n^2 -  \| h \|_{L^2}^2 \Big| \geq t 
\bigg)   \geq 1 -  \delta.
\] 
Note that $\| h \|_{L^2} = 1$ and $\| h \|_n^2 = \| f \|_n^2 / \| f \|_{L^2}^2$. Let $C = \max \{ 1 / \sqrt{c} , 1 / c , 1 \}$. Then,
\[
\PP \bigg(
\Big| \| f \|_n^2 -  \| f \|_{L^2}^2 \Big| \leq 
\| f \|_{L^2}^2 \cdot C K^2 \max\bigg\{ 
\sqrt{ \frac{\log(2/\delta)}{ n } } , ~  \frac{\log(2/\delta)}{ n }
\bigg\}
\bigg)  \geq 1 -  \delta.
\]

Let $\gamma =  C K^2  \sqrt{ \frac{\log(2/\delta)}{ n } } $. If $\gamma \leq 1$, then we see from $C \geq 1$ and $K \geq 1$ that 
$\frac{ \log(2/\delta)}{ n } \leq  1$. Hence, $CK^2 \max \{ 
\sqrt{ \frac{\log(2/\delta)}{ n } } , ~  \frac{\log(2/\delta)}{ n }
\}  = \gamma$ and 
\[
\PP \bigg(
\Big| \| f \|_n^2 -  \| f \|_{L^2}^2 \Big| \leq 
\gamma \| f \|_{L^2}^2 
\bigg)  \geq 1 -  \delta.
\]
This proves the desired result.
\end{proof}

{
\bibliographystyle{ims}
\bibliography{bib}

\begin{thebibliography}{56}
\expandafter\ifx\csname natexlab\endcsname\relax\def\natexlab#1{#1}\fi
\expandafter\ifx\csname url\endcsname\relax
  \def\url#1{\texttt{#1}}\fi
\expandafter\ifx\csname urlprefix\endcsname\relax\def\urlprefix{URL }\fi

\bibitem[{Aronszajn(1950)}]{Aro50}
\textsc{Aronszajn, N.} (1950).
\newblock Theory of reproducing kernels.
\newblock \textit{Transactions of the American mathematical society}
  \textbf{68} 337--404.

\bibitem[{Ben-David et~al.(2010)Ben-David, Blitzer, Crammer, Kulesza, Pereira
  and Vaughan}]{BBC10}
\textsc{Ben-David, S.}, \textsc{Blitzer, J.}, \textsc{Crammer, K.},
  \textsc{Kulesza, A.}, \textsc{Pereira, F.} and \textsc{Vaughan, J.~W.}
  (2010).
\newblock A theory of learning from different domains.
\newblock \textit{Machine learning} \textbf{79} 151--175.

\bibitem[{Blanchard and Massart(2006)}]{BMa06}
\textsc{Blanchard, G.} and \textsc{Massart, P.} (2006).
\newblock {Discussion: Local Rademacher complexities and oracle inequalities in
  risk minimization}.
\newblock \textit{The Annals of Statistics} \textbf{34} 2664 -- 2671.
\newline\urlprefix\url{https://doi.org/10.1214/009053606000001037}

\bibitem[{Blanchard et~al.(2019)Blanchard, Math{\'e} and M{\"u}cke}]{BMM19}
\textsc{Blanchard, G.}, \textsc{Math{\'e}, P.} and \textsc{M{\"u}cke, N.}
  (2019).
\newblock Lepskii principle in supervised learning.
\newblock \textit{arXiv preprint arXiv:1905.10764} .

\bibitem[{Cai et~al.(2021)Cai, Gao, Lee and Lei}]{CGL21}
\textsc{Cai, T.}, \textsc{Gao, R.}, \textsc{Lee, J.} and \textsc{Lei, Q.}
  (2021).
\newblock A theory of label propagation for subpopulation shift.
\newblock In \textit{Proceedings of the 38th International Conference on
  Machine Learning} (M.~Meila and T.~Zhang, eds.), vol. 139 of
  \textit{Proceedings of Machine Learning Research}. PMLR.
\newline\urlprefix\url{https://proceedings.mlr.press/v139/cai21b.html}

\bibitem[{Cai and Wei(2021)}]{CWe21}
\textsc{Cai, T.~T.} and \textsc{Wei, H.} (2021).
\newblock Transfer learning for nonparametric classification: Minimax rate and
  adaptive classifier.
\newblock \textit{The Annals of Statistics} \textbf{49}.

\bibitem[{Cameron et~al.(2008)Cameron, Gelbach and Miller}]{CGM08}
\textsc{Cameron, A.~C.}, \textsc{Gelbach, J.~B.} and \textsc{Miller, D.~L.}
  (2008).
\newblock Bootstrap-based improvements for inference with clustered errors.
\newblock \textit{The Review of Economics and Statistics} \textbf{90} 414--427.

\bibitem[{Caponnetto and Yao(2010)}]{CYa10}
\textsc{Caponnetto, A.} and \textsc{Yao, Y.} (2010).
\newblock Cross-validation based adaptation for regularization operators in
  learning theory.
\newblock \textit{Analysis and Applications} \textbf{8} 161--183.

\bibitem[{Cheng(1994)}]{Che94}
\textsc{Cheng, P.~E.} (1994).
\newblock Nonparametric estimation of mean functionals with data missing at
  random.
\newblock \textit{Journal of the American Statistical Association} \textbf{89}
  81--87.

\bibitem[{Cortes et~al.(2010)Cortes, Mansour and Mohri}]{CMM10}
\textsc{Cortes, C.}, \textsc{Mansour, Y.} and \textsc{Mohri, M.} (2010).
\newblock Learning bounds for importance weighting.
\newblock In \textit{Advances in Neural Information Processing Systems}
  (J.~Lafferty, C.~Williams, J.~Shawe-Taylor, R.~Zemel and A.~Culotta, eds.),
  vol.~23. Curran Associates, Inc.
\newline\urlprefix\url{https://proceedings.neurips.cc/paper_files/paper/2010/file/59c33016884a62116be975a9bb8257e3-Paper.pdf}

\bibitem[{Davis et~al.(2025)Davis, D{\'\i}az and Wang}]{DDW21}
\textsc{Davis, D.}, \textsc{D{\'\i}az, M.} and \textsc{Wang, K.} (2025).
\newblock Clustering a mixture of {G}aussians with unknown covariance.
\newblock \textit{Bernoulli} \textbf{31} 2105--2126.

\bibitem[{Donoho(1994)}]{Don94}
\textsc{Donoho, D.~L.} (1994).
\newblock Statistical estimation and optimal recovery.
\newblock \textit{The Annals of Statistics} \textbf{22} 238--270.

\bibitem[{Duan et~al.(2020)Duan, Jia and Wang}]{DJW20}
\textsc{Duan, Y.}, \textsc{Jia, Z.} and \textsc{Wang, M.} (2020).
\newblock Minimax-optimal off-policy evaluation with linear function
  approximation.
\newblock In \textit{International Conference on Machine Learning}. PMLR.

\bibitem[{Duan et~al.(2024)Duan, Wang and Wainwright}]{DWW24}
\textsc{Duan, Y.}, \textsc{Wang, M.} and \textsc{Wainwright, M.~J.} (2024).
\newblock Optimal policy evaluation using kernel-based temporal difference
  methods.
\newblock \textit{The Annals of Statistics} \textbf{52} 1927--1952.

\bibitem[{Feng et~al.(2024)Feng, He, Wang, Wang and Zhang}]{FHW24}
\textsc{Feng, X.}, \textsc{He, X.}, \textsc{Wang, C.}, \textsc{Wang, C.} and
  \textsc{Zhang, J.} (2024).
\newblock Towards a unified analysis of kernel-based methods under covariate
  shift.
\newblock \textit{Advances in Neural Information Processing Systems}
  \textbf{36}.

\bibitem[{Hanneke and Kpotufe(2019)}]{HKp19}
\textsc{Hanneke, S.} and \textsc{Kpotufe, S.} (2019).
\newblock On the value of target data in transfer learning.
\newblock \textit{Advances in Neural Information Processing Systems}
  \textbf{32}.

\bibitem[{Heckman(1979)}]{Hec79}
\textsc{Heckman, J.~J.} (1979).
\newblock Sample selection bias as a specification error.
\newblock \textit{Econometrica: Journal of the econometric society}  153--161.

\bibitem[{Hirshberg et~al.(2019)Hirshberg, Maleki and Zubizarreta}]{HMa19}
\textsc{Hirshberg, D.~A.}, \textsc{Maleki, A.} and \textsc{Zubizarreta, J.~R.}
  (2019).
\newblock Minimax linear estimation of the retargeted mean.
\newblock \textit{arXiv preprint arXiv:1901.10296} .

\bibitem[{Hirshberg and Wager(2021)}]{HWa21}
\textsc{Hirshberg, D.~A.} and \textsc{Wager, S.} (2021).
\newblock Augmented minimax linear estimation.
\newblock \textit{The Annals of Statistics} \textbf{49} 3206--3227.

\bibitem[{Hoerl and Kennard(1970)}]{HKe70}
\textsc{Hoerl, A.~E.} and \textsc{Kennard, R.~W.} (1970).
\newblock Ridge regression: Biased estimation for nonorthogonal problems.
\newblock \textit{Technometrics} \textbf{12} 55--67.

\bibitem[{Huang et~al.(2006)Huang, Gretton, Borgwardt, Sch\"{o}lkopf and
  Smola}]{HGB06}
\textsc{Huang, J.}, \textsc{Gretton, A.}, \textsc{Borgwardt, K.},
  \textsc{Sch\"{o}lkopf, B.} and \textsc{Smola, A.} (2006).
\newblock Correcting sample selection bias by unlabeled data.
\newblock In \textit{Advances in Neural Information Processing Systems}
  (B.~Sch\"{o}lkopf, J.~Platt and T.~Hoffman, eds.), vol.~19. MIT Press.
\newline\urlprefix\url{https://proceedings.neurips.cc/paper_files/paper/2006/file/a2186aa7c086b46ad4e8bf81e2a3a19b-Paper.pdf}

\bibitem[{Koltchinskii and Lounici(2017)}]{KLo17}
\textsc{Koltchinskii, V.} and \textsc{Lounici, K.} (2017).
\newblock Concentration inequalities and moment bounds for sample covariance
  operators.
\newblock \textit{Bernoulli} \textbf{23} 110--133.

\bibitem[{Kpotufe and Martinet(2021)}]{KMa21}
\textsc{Kpotufe, S.} and \textsc{Martinet, G.} (2021).
\newblock Marginal singularity and the benefits of labels in covariate-shift.
\newblock \textit{The Annals of Statistics} \textbf{49} 3299--3323.

\bibitem[{Kumar et~al.(2020)Kumar, Ma and Liang}]{KML20}
\textsc{Kumar, A.}, \textsc{Ma, T.} and \textsc{Liang, P.} (2020).
\newblock Understanding self-training for gradual domain adaptation.
\newblock In \textit{Proceedings of the 37th International Conference on
  Machine Learning} (H.~D. III and A.~Singh, eds.), vol. 119 of
  \textit{Proceedings of Machine Learning Research}. PMLR.
\newline\urlprefix\url{https://proceedings.mlr.press/v119/kumar20c.html}

\bibitem[{Le~Cam(1973)}]{Le73}
\textsc{Le~Cam, L.} (1973).
\newblock Convergence of estimates under dimensionality restrictions.
\newblock \textit{The Annals of Statistics}  38--53.

\bibitem[{Lee(2013)}]{Lee13}
\textsc{Lee, D.-H.} (2013).
\newblock Pseudo-label: The simple and efficient semi-supervised learning
  method for deep neural networks.
\newblock In \textit{Workshop on challenges in representation learning, ICML},
  vol.~3.

\bibitem[{Lei et~al.(2021)Lei, Hu and Lee}]{LHL21}
\textsc{Lei, Q.}, \textsc{Hu, W.} and \textsc{Lee, J.} (2021).
\newblock Near-optimal linear regression under distribution shift.
\newblock In \textit{International Conference on Machine Learning}. PMLR.

\bibitem[{Lepskii(1991)}]{Lep91}
\textsc{Lepskii, O.} (1991).
\newblock On a problem of adaptive estimation in gaussian white noise.
\newblock \textit{Theory of Probability \& Its Applications} \textbf{35}
  454--466.

\bibitem[{Li(1982)}]{Li82}
\textsc{Li, K.-C.} (1982).
\newblock Minimaxity of the method of regularization of stochastic processes.
\newblock \textit{The Annals of Statistics} \textbf{10} 937--942.

\bibitem[{Lin and Reimherr(2024)}]{LRe24}
\textsc{Lin, H.} and \textsc{Reimherr, M.} (2024).
\newblock Smoothness adaptive hypothesis transfer learning.
\newblock In \textit{Proceedings of the 41st International Conference on
  Machine Learning}. ICML'24, JMLR.org.

\bibitem[{Liu et~al.(2021)Liu, Wang and Long}]{LWL21}
\textsc{Liu, H.}, \textsc{Wang, J.} and \textsc{Long, M.} (2021).
\newblock Cycle self-training for domain adaptation.
\newblock \textit{Advances in Neural Information Processing Systems}
  \textbf{34} 22968--22981.

\bibitem[{Liu et~al.(2023)Liu, Zhang, Liao and Cai}]{LZL20}
\textsc{Liu, M.}, \textsc{Zhang, Y.}, \textsc{Liao, K.~P.} and \textsc{Cai, T.}
  (2023).
\newblock Augmented transfer regression learning with semi-non-parametric
  nuisance models.
\newblock \textit{Journal of Machine Learning Research} \textbf{24} 1--50.

\bibitem[{Ma et~al.(2023)Ma, Pathak and Wainwright}]{MPW22}
\textsc{Ma, C.}, \textsc{Pathak, R.} and \textsc{Wainwright, M.~J.} (2023).
\newblock Optimally tackling covariate shift in {RKHS}-based nonparametric
  regression.
\newblock \textit{The Annals of Statistics} \textbf{51} 738--761.

\bibitem[{Maity et~al.(2022)Maity, Sun and Banerjee}]{MSB22}
\textsc{Maity, S.}, \textsc{Sun, Y.} and \textsc{Banerjee, M.} (2022).
\newblock Minimax optimal approaches to the label shift problem in
  non-parametric settings.
\newblock \textit{Journal of Machine Learning Research} \textbf{23} 1--45.

\bibitem[{Minsker(2017)}]{Min17}
\textsc{Minsker, S.} (2017).
\newblock On some extensions of {B}ernstein's inequality for self-adjoint
  operators.
\newblock \textit{Statistics \& Probability Letters} \textbf{127} 111--119.

\bibitem[{Mou et~al.(2023)Mou, Ding, Wainwright and Bartlett}]{MDW23}
\textsc{Mou, W.}, \textsc{Ding, P.}, \textsc{Wainwright, M.~J.} and
  \textsc{Bartlett, P.~L.} (2023).
\newblock Kernel-based off-policy estimation without overlap: Instance
  optimality beyond semiparametric efficiency.
\newblock \textit{arXiv preprint arXiv:2301.06240} .

\bibitem[{Mou et~al.(2022)Mou, Wainwright and Bartlett}]{MWB22}
\textsc{Mou, W.}, \textsc{Wainwright, M.~J.} and \textsc{Bartlett, P.~L.}
  (2022).
\newblock Off-policy estimation of linear functionals: Non-asymptotic theory
  for semi-parametric efficiency.
\newblock \textit{arXiv preprint arXiv:2209.13075} .

\bibitem[{Mousavi~Kalan et~al.(2020)Mousavi~Kalan, Fabian, Avestimehr and
  Soltanolkotabi}]{MFA20}
\textsc{Mousavi~Kalan, M.}, \textsc{Fabian, Z.}, \textsc{Avestimehr, S.} and
  \textsc{Soltanolkotabi, M.} (2020).
\newblock Minimax lower bounds for transfer learning with linear and one-hidden
  layer neural networks.
\newblock \textit{Advances in Neural Information Processing Systems}
  \textbf{33} 1959--1969.

\bibitem[{Page and Gr{\"u}new{\"a}lder(2021)}]{PGr21}
\textsc{Page, S.} and \textsc{Gr{\"u}new{\"a}lder, S.} (2021).
\newblock The {G}oldenshluger--{L}epski method for constrained least-squares
  estimators over {RKHS}s.
\newblock \textit{Bernoulli} \textbf{27} 2241--2266.

\bibitem[{Pan and Yang(2010)}]{PYa10}
\textsc{Pan, S.~J.} and \textsc{Yang, Q.} (2010).
\newblock A survey on transfer learning.
\newblock \textit{IEEE Transactions on knowledge and data engineering}
  \textbf{22} 1345--1359.

\bibitem[{Pathak et~al.(2022)Pathak, Ma and Wainwright}]{PMW22}
\textsc{Pathak, R.}, \textsc{Ma, C.} and \textsc{Wainwright, M.} (2022).
\newblock A new similarity measure for covariate shift with applications to
  nonparametric regression.
\newblock In \textit{International Conference on Machine Learning}. PMLR.

\bibitem[{Reeve et~al.(2021)Reeve, Cannings and Samworth}]{RCS21}
\textsc{Reeve, H.~W.}, \textsc{Cannings, T.~I.} and \textsc{Samworth, R.~J.}
  (2021).
\newblock Adaptive transfer learning.
\newblock \textit{The Annals of Statistics} \textbf{49} 3618--3649.

\bibitem[{Schmidt-Hieber and Zamolodtchikov(2024)}]{SZa22}
\textsc{Schmidt-Hieber, J.} and \textsc{Zamolodtchikov, P.} (2024).
\newblock Local convergence rates of the nonparametric least squares estimator
  with applications to transfer learning.
\newblock \textit{Bernoulli} \textbf{30} 1845--1877.

\bibitem[{Shimodaira(2000)}]{Shi00}
\textsc{Shimodaira, H.} (2000).
\newblock Improving predictive inference under covariate shift by weighting the
  log-likelihood function.
\newblock \textit{Journal of statistical planning and inference} \textbf{90}
  227--244.

\bibitem[{Speckman(1979)}]{Spe79}
\textsc{Speckman, P.} (1979).
\newblock Minimax estimates of linear functionals in a hilbert space.
\newblock \textit{Unpublished manuscript} .

\bibitem[{Sugiyama and Kawanabe(2012)}]{SKa12}
\textsc{Sugiyama, M.} and \textsc{Kawanabe, M.} (2012).
\newblock \textit{Machine learning in non-stationary environments: Introduction
  to covariate shift adaptation}.
\newblock MIT press.

\bibitem[{Tripuraneni et~al.(2021)Tripuraneni, Adlam and Pennington}]{TAP21}
\textsc{Tripuraneni, N.}, \textsc{Adlam, B.} and \textsc{Pennington, J.}
  (2021).
\newblock Covariate shift in high-dimensional random feature regression.
\newblock \textit{arXiv preprint arXiv:2111.08234} .

\bibitem[{Vapnik(1999)}]{Vap99}
\textsc{Vapnik, V.} (1999).
\newblock \textit{The nature of statistical learning theory}.
\newblock Springer science \& business media.

\bibitem[{Vershynin(2010)}]{Ver10}
\textsc{Vershynin, R.} (2010).
\newblock Introduction to the non-asymptotic analysis of random matrices.
\newblock \textit{arXiv preprint arXiv:1011.3027} .

\bibitem[{Wahba(1990)}]{Wah90}
\textsc{Wahba, G.} (1990).
\newblock \textit{Spline models for observational data}.
\newblock SIAM.

\bibitem[{Wainwright(2019)}]{Wai19}
\textsc{Wainwright, M.~J.} (2019).
\newblock \textit{High-dimensional statistics: A non-asymptotic viewpoint},
  vol.~48.
\newblock Cambridge University Press.

\bibitem[{Wang and Kim(2023)}]{WKi23}
\textsc{Wang, H.} and \textsc{Kim, J.~K.} (2023).
\newblock Statistical inference using regularized m-estimation in the
  reproducing kernel hilbert space for handling missing data.
\newblock \textit{Annals of the Institute of Statistical Mathematics}
  \textbf{75} 911--929.

\bibitem[{Wu et~al.(2022)Wu, Zou, Braverman, Gu and Kakade}]{WZB22}
\textsc{Wu, J.}, \textsc{Zou, D.}, \textsc{Braverman, V.}, \textsc{Gu, Q.} and
  \textsc{Kakade, S.} (2022).
\newblock The power and limitation of pretraining-finetuning for linear
  regression under covariate shift.
\newblock In \textit{Advances in Neural Information Processing Systems}
  (S.~Koyejo, S.~Mohamed, A.~Agarwal, D.~Belgrave, K.~Cho and A.~Oh, eds.),
  vol.~35. Curran Associates, Inc.
\newline\urlprefix\url{https://proceedings.neurips.cc/paper_files/paper/2022/file/d5c04aa72b92c53bda5b525b60958295-Paper-Conference.pdf}

\bibitem[{Yang et~al.(2025)Yang, Zhang, Wu, Re and Su}]{YZW20}
\textsc{Yang, F.}, \textsc{Zhang, H.~R.}, \textsc{Wu, S.}, \textsc{Re, C.} and
  \textsc{Su, W.~J.} (2025).
\newblock Precise high-dimensional asymptotics for quantifying heterogeneous
  transfers.
\newblock \textit{Journal of Machine Learning Research} \textbf{26} 1--88.

\bibitem[{Zadrozny(2004)}]{Zad04}
\textsc{Zadrozny, B.} (2004).
\newblock Learning and evaluating classifiers under sample selection bias.
\newblock In \textit{Proceedings of the Twenty-First International Conference
  on Machine Learning}. ICML '04, Association for Computing Machinery, New
  York, NY, USA.
\newline\urlprefix\url{https://doi.org/10.1145/1015330.1015425}

\bibitem[{Zhang(2005)}]{Zha05}
\textsc{Zhang, T.} (2005).
\newblock Learning bounds for kernel regression using effective data
  dimensionality.
\newblock \textit{Neural computation} \textbf{17} 2077--2098.

\end{thebibliography}
}

\end{document}